\documentclass[a4paper]{article}
\usepackage[margin=1in]{geometry}
\usepackage[main=english,french]{babel}
\usepackage[T1]{fontenc}
\usepackage{amsmath}
\usepackage{amssymb,cite}
\usepackage{amsthm}
\usepackage{amsfonts}
\usepackage{arydshln}  
\usepackage{enumerate}
\usepackage{epsfig}
\usepackage{etoolbox}
\usepackage{graphicx}
\usepackage{float} 
\usepackage{subfigure} 

\usepackage{hyperref}
\hypersetup{colorlinks,
            linkcolor={blue!80!black},
            urlcolor={blue!80!black},
            citecolor=teal}
\usepackage{mathrsfs}
\usepackage{pdfsync}
\usepackage{scalerel}
\usepackage{stmaryrd}
\usepackage{upgreek}
\usepackage{url}
\usepackage[dvipsnames]{xcolor}

\usepackage{abstract}
\usepackage{stmaryrd}
\SetSymbolFont{stmry}{bold}{U}{stmry}{m}{n}
\usepackage{mathtools}
\usepackage[title]{appendix}
\usepackage{tablefootnote}

\usepackage{qcircuit}

\newcommand{\Ket}[1]{| #1 \rangle}
\newcommand{\Bra}[1]{\langle #1|}

\allowdisplaybreaks
\usepackage{amscd}
\usepackage{tikz}  
\usepackage{tikz-cd}  
\usepackage{listings}
\usepackage{framed}

\newcommand{\be}{\begin{equation}}
\newcommand{\ee}{\end{equation}}
\newcommand{\bea}{\begin{eqnarray}}
\newcommand{\eea}{\end{eqnarray}}
\newcommand{\bes}{\begin{equation*}}
\newcommand{\ees}{\end{equation*}}
\newcommand{\beas}{\begin{eqnarray*}}
\newcommand{\eeas}{\end{eqnarray*}}

\newtheorem{thm}{Theorem}[section]

\newtheorem{prop}[thm]{Proposition}
\newtheorem{cor}[thm]{Corollary}
\newtheorem{lem}[thm]{Lemma}

\newtheorem{prob}[thm]{Problem}

\newtheorem{def-prop}[thm]{Definition-Proposition}

\theoremstyle{definition}
\newtheorem{defi}[thm]{Definition}
\newtheorem{rmk}[thm]{Remark}

\numberwithin{equation}{section}

\newcommand{\nc}{\newcommand}
\nc{\on}{\operatorname}
\nc{\mb}{\mathbb}
\nc{\mbf}{\mathbf}
\nc{\mc}{\mathcal}
\nc{\mf}{\mathfrak}
\nc{\mr}{\mathrm}
\nc{\msr}{\mathscr}

\nc{\bra}{\langle}
\nc{\ket}{\rangle}
\nc{\tr}{\on{Tr}}
\nc{\Perf}{\on{Perf}}
\nc{\Fq}{{\mb F}_q}
\nc{\Newton}{X_*(T)^{+}_{\mb Q}}
\nc{\Kottwitz}{\pi_1(G_{\overline{E}})_{\Gamma}}
\nc{\Bun}{\on{Bun}}
\nc{\Gr}{\on{Gr}}
\nc{\Hck}{\mc Hck}
\nc{\Dbz}{\mathbb{D}_{\mathrm{BZ}}}
\nc{\Rihom}{R\mathscr Hom}
\nc{\red}{\textcolor{red}}

\def\a{{\boldsymbol{a}}}

\def\g{{\boldsymbol{g}}}
\def\v{{\boldsymbol{v}}}
\def\H{{\mbf H}}
\def\x{{\boldsymbol{x}}}
\def\0{{\boldsymbol{0}}}

\def\y{{\boldsymbol{y}}}
\def\z{{\boldsymbol{z}}}
\def\r{{\boldsymbol{r}}}

\nc{\Cat}{\mathbf{Cat_{\infty}}}
\nc{\Corr}{\on{Corr}(\mc C,E)}
\nc{\LZ}{\on{LZ}_{\mc D}}
\nc{\Cl}{\on{C\ell}}
\usepackage{authblk}

\newcommand{\overbar}[1]
{\mkern 1.5mu\overline{\mkern-3.0mu#1\mkern-0.5mu}\mkern 1.5mu}

\title{Quantum spectral method for gradient and Hessian estimation}

\author{Yuxin Zhang\thanks{zhangyuxin@amss.ac.cn} \, and Changpeng Shao\thanks{changpeng.shao@amss.ac.cn}}

\affil{SKLMS, Academy of Mathematics and Systems Science, Chinese Academy of Sciences, Beijing, 100190 China}

\begin{document}

\maketitle
\thispagestyle{empty}

\begin{abstract}

Gradient descent is one of the most basic algorithms for solving continuous optimization problems. In 
[\href{https://journals.aps.org/prl/abstract/10.1103/PhysRevLett.95.050501}{Jordan, PRL, 95(5):050501, 2005}], Jordan proposed the first quantum algorithm for estimating gradients of functions close to linear, with exponential speedup in the black-box model. This algorithm was greatly enhanced and developed by [\href{https://epubs.siam.org/doi/abs/10.1137/1.9781611975482.87}{Gily\'{e}n, Arunachalam, and Wiebe, SODA, pp. 1425-1444, 2019}], providing a quantum algorithm with optimal query complexity $\widetilde{\Theta}(\sqrt{d}/\varepsilon)$ for a class of smooth functions of $d$ variables, where $\varepsilon$ is the accuracy. This is quadratically faster than classical algorithms for the same problem.

In this work, we continue this research by proposing a new quantum algorithm for another class of functions, namely, analytic functions $f(\x)$ which are well-defined over the complex field. Given phase oracles to query the real and imaginary parts of $f(\x)$ respectively, we propose a quantum algorithm that returns an $\varepsilon$-approximation of its gradient with query complexity $\widetilde{O}(1/\varepsilon)$. 
As an extension, we also propose two quantum algorithms for Hessian estimation, aiming to improve quantum analogs of Newton's method. The two algorithms have query complexity $\widetilde{O}(d/\varepsilon)$ and $\widetilde{O}(d^{1.5}/\varepsilon)$, respectively, under different assumptions. 
Moreover, if the Hessian is promised to be $s$-sparse, we then have two new quantum algorithms with query complexity $\widetilde{O}(s/\varepsilon)$ and $\widetilde{O}(sd/\varepsilon)$, respectively. 
We also prove a lower bound of $\widetilde{\Omega}(d)$ for Hessian estimation in the general case.

\end{abstract}

\newpage
\pagenumbering{roman}

\tableofcontents
\newpage
\pagenumbering{arabic}
\setcounter{page}{1}

\section{Introduction}


Efficient estimation of the gradient and Hessian of functions plays a vital role across a wide range of fields, including mathematical optimization \cite{boyd2004convex}, machine learning \cite{sra2012optimization}, and others. For instance, in optimization problems, the gradient provides essential information about the direction of the steepest ascent or descent of the given function, guiding the search for optimal solutions. The Hessian matrix captures the second-order curvature and local behavior of the function, which is crucial for assessing the nature of these solutions\textemdash whether they are minima, maxima, or saddle points. They are widely used in many fundamental optimization algorithms, such as gradient descent and Newton's method. 

Given a real-valued function $f(\x):\mathbb{R}^d \rightarrow \mathbb{R}$ of $d$ variables, its gradient at point $\a$ is defined as $\nabla f(\a) = (\frac{\partial f(\x)}{\partial x_1}, \ldots, \frac{\partial f(\x)}{\partial x_d})|_{\x=\a}$.
As the gradient contains $d$ entries, gradient estimation can be computationally expensive, especially in high-dimensional cases. This challenge has stimulated significant efforts toward developing efficient quantum algorithms for gradient computation. Jordan \cite{jordan2005fast} proposed a fast quantum algorithm for numerical gradient estimation, which is a landmark contribution in this area. Particularly, if the function is close to linear, then the query complexity is $O(1)$, which is exponentially faster than classical algorithms.
In the information-theoretic sense, the classical gradient computation algorithm has a complexity of $\Omega(d)$ in general.
Later, Gily\'{e}n, Arunachalam, and Wiebe \cite{gilyen2019optimizing} significantly developed Jordan's algorithm and improved the efficiency using a higher-order finite difference formula \cite{li2005general}. Their algorithm requires $\widetilde{\Theta}(\sqrt{d}/\varepsilon)$ optimal quantum queries for a class of smooth functions, where $\varepsilon$ is the accuracy of approximating the gradient under the $\ell_\infty$-norm. This is quadratically more efficient than classical algorithms.
As a result, these fast quantum algorithms for computing the gradient
become an important subroutine of many quantum algorithms, such
as quantum algorithms for convex optimization \cite{chakrabarti2020quantum,van2020convex}, quantum tomography \cite{van2023quantum}, quantum algorithms for multiple expectation-value estimation \cite{huggins2022nearly}, and quantum reinforcement learning \cite{jerbi2023quantum}, among others.

In the exploration of quantum advantages, quantum algorithms with exponential speed-up are very attractive, and other quantum algorithms based on them can lead to large quantum speedups over classical algorithms for practically relevant applications, such as  \cite{chakrabarti2020quantum,van2020convex,van2023quantum}. These algorithms are indeed still based on Jordan's algorithm. Therefore, it is interesting to know more about the existence of exponential quantum speedups for gradient estimation.

With this motivation in mind, in this work, we propose a new quantum algorithm for gradient estimation for another class of functions. More importantly, this algorithm achieves exponential speedup over classical algorithms.
We will consider analytic functions that are allowed to take values from the complex field, and we assume we are given quantum oracles that can query the real and imaginary parts. Our main result roughly states that for any analytic function with phase oracles to query the real and imaginary parts respectively, there is a quantum algorithm that approximates its gradient, which only costs $\widetilde{O}(1/\varepsilon)$ quantum queries. Unlike previous quantum algorithms \cite{jordan2005fast,gilyen2019optimizing}, which are based on higher-order finite difference method \cite{li2005general}, our algorithm is based on the spectral method \cite{fornberg1981numerical,trefethen2000spectral}. Apart from this, another difference is the oracle, as we are dealing with complex-valued functions. We will explain this in detail in Subsection \ref{subsec: intro, comparison}.
It is interesting to note that previous quantum algorithms for solving differential equations \cite{berry2014high,childs2020quantum} already showed that the spectral method is more efficient than the higher-order finite difference method.

As a generalization of Gily\'{e}n-Arunachalam-Wiebe's algorithm and our algorithm for gradient estimation, we propose two quantum algorithms for estimating the Hessian of analytic functions under the matrix max-norm. Both are polynomially faster than classical algorithms. We will show that polynomial speedup for Hessian estimation is the best we can expect by proving a lower bound of $\widetilde{\Omega}(d)$. However, for sparse Hessians, which are very common in practice, we show that larger quantum speedups than polynomial exist.
Notice that the Hessian is widely used in many fundamental optimization algorithms, perhaps the most famous one is Newton's method. Thus, we hope our quantum algorithms can be used to accelerate Newton's method for some practical applications.

In the following subsections, we explain our results in more detail.

\subsection{Main results}

Let $f: \mb C^d \rightarrow \mb C$ be an analytic function that maps $\mb R^d$ into $\mb R$. 
Since $f$ is analytic, the latter property implies that the gradient satisfies $\nabla f(\mb R^d) \subseteq \mb R^d$. For a complex vector $\x=\x_1+i \x_2 \in \mathbb{C}^d$ and a complex function $f(\x)=f_1(\x)+i f_2(\x) \in \mathbb{C}$, where $i$ denotes the imaginary unit, we naturally store them in quantum registers as follows:
\begin{equation}\label{eq: intro, store complex numbers}
    |\boldsymbol{x}\rangle:=\left|\boldsymbol{x}_1\right\rangle\left|\boldsymbol{x}_2\right\rangle, \quad|f(\boldsymbol{x})\rangle:=\left|f_1(\boldsymbol{x})\right\rangle\left|f_2(\boldsymbol{x})\right\rangle .
\end{equation}
For real numbers $r\in \mb R$, we use the standard and most common storage method, representing $r$ as its finite-precision binary encoding. As before, we also assume phase oracle access to query $f$. Specifically, we assume
phase oracles $O_{f_1}, O_{f_2}$ for the real and imaginary parts of $f=f_1+i f_2$ separately:
\be 
\label{intro: phase oracle}
O_{f_{1}}: \Ket{\x} \rightarrow e^{i f_{1}(\x)}\Ket{\x},
\qquad 
O_{f_{2}}: \Ket{\x} \rightarrow e^{i f_{2}(\x)}\Ket{\x}.
\ee
These are our basic settings. Although the target function involves complex numbers, the representation remains the typical and natural ways in quantum computing.

\begin{prob}\label{prob: problem 1}
    Let $\varepsilon\in (0,1)$ denote the accuracy. Assume that $f: \mb C^d\rightarrow \mb C$ is an analytic function satisfying $f(\mb R^d)\subseteq \mb R$. Given access to oracles $O_{f_1}, O_{f_2}$, compute an approximate gradient $\boldsymbol{g}\in \mb R^d$ such that $\|\boldsymbol{g}-\nabla f(\0)\|_{\infty} \leq \varepsilon$.\footnote{We can also estimate the gradient at any $\x_0 \in \mathbb{R}^d$ via a linear map. Here $\|\cdot\|_\infty$ is the $\ell_\infty$-norm of vectors. For a vector $\v=(v_1,\ldots,v_d)$, we define $\|\v\|_{\infty} := \max_j |v_j|$. }
\end{prob}

When the function $f$ is only well-defined on $\mb R^d$, \cite{gilyen2019optimizing} obtained the following result for a class of functions satisfying certain smoothness conditions.

\begin{prop}[Finite difference method]
\label{prop: intro GAW's result}
    Assume that $f: \mathbb{R}^d \rightarrow \mathbb{R}$ is analytic and satisfies $\left|\partial^\alpha f(\boldsymbol{x})\right| \leq c^k k^{k/2}$ for every $\x\in \mathbb{R}^d$, $k \in \mathbb{N}$ and $\alpha \in \mb N_0^d$ such that $|\alpha|=k$. Given the phase oracle $O_f$, there exists a quantum algorithm that outputs $\boldsymbol{g} \in \mathbb{R}^d$ such that $\|\boldsymbol{g}-\nabla f(\0)\|_{\infty} \leq \varepsilon$, with query complexity $\widetilde{\Theta}(\sqrt{d}/\varepsilon)$.
\end{prop}

The above functions are known as a special kind of Gevrey class functions \cite{cornelissen2019quantum}.\footnote{Let $\sigma \in \mathbb{R}, c\in \mb R_+$ and $f: \mathbb{R}^d \rightarrow \mathbb{R}$. We call $f$ a {\bf Gevrey class function} if: (i) $f$ is smooth. (ii) The power series of $f$ around any point $\x \in \mathbb{R}^d$ converges. (iii) For all $\x \in \mathbb{R}^d, k \in \mathbb{N}$ and $\alpha \in \mathbb{N}^d$ with $|\alpha|=k$, we have
$\left|\partial^\alpha f(\x)\right| \leq c^k(k!)^\sigma.$}
In \cite{cornelissen2019quantum}, Cornelissen summarized and generalized the above result to a larger class of Gevrey class functions. Note that when $\sigma>1$, these functions lie between smooth functions and real analytic functions; for $\sigma=1$, they are equivalent to real analytic functions. Therefore, it is evident that the assumptions made in Proposition \ref{prop: intro GAW's result} correspond to Gevrey functions with $\sigma=1/2$, which is a proper subset of real analytic functions.


For Problem \ref{prob: problem 1}, our main result is as follows.

\begin{thm}[Informal form of Theorem \ref{thm: gradient estimation using spectral method}, spectral method]\label{thm: intro, gradient estimation, spectral}
    There exists a quantum algorithm that solves Problem \ref{prob: problem 1} with  $\widetilde{O}(1/\varepsilon)$ queries to $O_{f_1},O_{f_2}$.
\end{thm}


We now provide some remarks on our main result.

\begin{rmk}
The functions considered in Problem \ref{prob: problem 1}, when restricted to $\mb R^d$, 
are real analytic. Hence, they form a proper subset of the class of real analytic functions. The class of these functions and those considered in \cite{gilyen2019optimizing} intersect, but neither class is contained within the other. This illustrates that our main result\textemdash a quantum algorithm with better query complexity performance\textemdash does not contradict the optimality stated in Proposition \ref{prop: intro GAW's result}. We will provide a more detailed explanation of this in Subsection \ref{subsec: intro, comparison}.
\end{rmk}

As mentioned in \cite[below Theorem 5.2]{gilyen2019optimizing}, for multivariate polynomials $P(\x)$ of degree $O({\rm polylog}(d))$, the quantum algorithm proposed in \cite{gilyen2019optimizing} achieves polylogarithmic query complexity in dimension $d$, offering an exponential speedup over Jordan's \cite{jordan2005fast}. In comparison, Theorem \ref{thm: intro, gradient estimation, spectral} implies that if we have oracles that query the real and imaginary parts of $P(\x)$ for all $\x\in \mathbb{C}^d$, then the quantum query complexity of computing the gradient of $P$ is always polylogarithmic in dimension $d$, regardless of the degree.

\begin{rmk}
Regarding the assumption that $f$ maps $\mathbb{R}^d$ to $\mathbb{R}$, we indeed only rely on the property that $\nabla f(\x)\in \mathbb{R}^d$ for any $\x\in \mathbb{R}^d$. In this work, we find it more natural to state the assumption as $f(\mathbb{R}^d) \subseteq \mathbb{R}$, since it directly concerns the function $f$ rather than its gradient.
Additionally, to approximate $\nabla f(\0)$, our algorithm works for real analytic $f$ that can be evaluated at points in the complex plane near $\0$. It does not require $f$ to be analytic over $\mathbb{C}^d$.
\end{rmk}

Some readers might wonder why we do not consider the analytic continuation of $f$ when it is only real analytic, which would seemingly allow our algorithm to apply in the situation mentioned in Proposition \ref{prop: intro GAW's result}. There are indeed two issues with this idea: (i) computing the analytic continuation of $f(\x)$ typically involves computing its gradient; and (ii) even if the analytic continuation $\tilde{f}(\x)$ of $f(\x)$ is already known, constructing an oracle to query $\tilde{f}(\x)$ based on an oracle that queries $f(\x)$ is non-trivial. This difficulty is clearly exemplified by the simple function $\tilde{f}(\x)=f(\x)=x_t^k$, where $x_t$ is the $t$-th entry of $\x$, and $t, k$ are unknown. Given an oracle that can only evaluate $f(\x) = x_t^k$ for $x_t \in \mathbb{R}$, it is challenging to use this oracle to evaluate $\tilde{f}(\x) = (x_{t1} +i x_{t2})^k$ for $x_{t1}, x_{t2} \in \mathbb{R}$ without knowing $t, k$.

As a continuation of gradient estimation, we also present two quantum algorithms for estimating the Hessian of $f$, which is fundamental to second-order iterative methods for optimization, such as Newton's method. Below, for a matrix $A=(a_{ij})$, the max norm is defined as $\|A\|_{\max}=\max_{i,j} |a_{ij}|$. 
We propose a generalized approach for computing the Hessian by leveraging quantum gradient estimation algorithms. This generalization is applicable to both our quantum gradient estimation algorithm based on the spectral method and GAW's algorithm based on the finite difference method. As a result, we derive two corresponding quantum algorithms for Hessian estimation, building upon these two quantum gradient estimation approaches.

\begin{rmk}
    It is important to note that although the Hessian is obtained by computing the gradient of each of the $d$ components of the gradient function $\nabla f$ (i.e., the partial derivative functions), we cannot directly use the gradient algorithm $d$ times to compute the Hessian matrix of $f$ at a point. This is because the gradient algorithm requires an oracle $O_f$ for the function $f$ to compute the gradient value $\nabla f(x_0)$ at a specific point $x_0$, rather than computing the entire gradient function $\nabla f$. If we attempt to use the gradient algorithm to compute the gradient of a partial derivative function $\partial_i f$, we would need access to an oracle that evaluates $\partial_i f$, which we do not have.
\end{rmk}

\begin{thm}[Informal form of Theorem \ref{thm: hessian using spectral}, spectral method]
\label{thm: intro, Hessian using spectral}
Let $f:\mathbb{C}^d\rightarrow \mathbb{C}$ be an analytic function that maps $\mb R^d$ to $\mb R$. Then there is a quantum algorithm that uses $\widetilde{O}(d/\varepsilon)$ queries to $O_{f_1}, O_{f_2}$, and computes $\widetilde{\H}$ such that $\|\widetilde{\H} - \H_f(\0)\|_{\max} \leq \varepsilon$,  where $\H_f(\0)$ is the Hessian of $f$ at $\0$.
\end{thm}

\begin{thm}[Informal form of Theorem \ref{new thm: Hessian estimation using finite difference formula}, finite difference method]
\label{thm: intro, Hessian using finite difference}
    Let $m \in \mathbb{N}, B\in\mb R_+$ satisfy $m \geq \log(dB/\varepsilon)$.
    Assume $f:\left[-R, R\right]^d$ $\rightarrow \mathbb{R}$ is $(2 m+1)$-times differentiable and $\left|\partial_{\boldsymbol{r}}^{2 m+1} f(\x)\right| \leq B$ for all  $\x \in[-R,R]^d$ and $\r=\x/\|\x\|$.
    Then there is a quantum algorithm that uses $\widetilde{O}(d^{1.5}/\varepsilon)$ queries to $O_f$ and computes $\widetilde{\H}$ such that $\|\widetilde{\H} - \H_f(\0)\|_{\max} \leq \varepsilon$.
\end{thm}

The following result shows that this is the best we can achieve generally.

\begin{thm}[Informal forms of Propositions \ref{prop: lower bound} and \ref{prop: lower bound of real-valued Hessian estimation}]
Any quantum algorithm that approximates $\H_f(\0)$ under the matrix max-norm with constant accuracy $\varepsilon\in(0,1/2]$ must make $\widetilde{\Omega}(d)$ queries to phase oracle.
\end{thm}

The above lower bound shows that the quantum spectral method for Hessian estimation is optimal in terms of $d$.
However, it does not rule out the possibility of the existence of exponential quantum speedups for some special cases, e.g., the graph learning problem considered in \cite{montanaro_shao-tqc,lee2021quantum}. Indeed, for sparse Hessians, we can have better quantum algorithms. 
The following results further demonstrate that the algorithm based on the spectral method has significant advantages, with its query complexity being very low and, in the sparse case, independent of $d$.

\begin{thm}[Informal form of Theorem \ref{thm: Sparse Hessian estimation using spectral method}, spectral method]
\label{thm: intro new2}
    Under the same assumption as in Theorem \ref{thm: intro, Hessian using spectral}, there is a quantum algorithm that estimates $\H_f(\0)$ using 
    \begin{itemize}
    \item $\widetilde{O}(s/\varepsilon)$ queries to $O_{f_1}, O_{f_2}$ if $\H_f(\0)$ is $s$-sparse.
    \item $\widetilde{O}\left(\sqrt{m}/\varepsilon\right)$ queries to $O_{f_1}, O_{f_2}$ if $\H_f(\0)$ has at most $m$ nonzero entries.
    \end{itemize}
\end{thm}

\begin{thm}[Informal form of Theorem \ref{thm: Sparse Hessian estimation using finite difference formula}, finite difference method] 
\label{thm: intro new1}
     Under the same assumption as in Theorem \ref{thm: intro, Hessian using finite difference}, there is a quantum algorithm that estimates $\H_f(\0)$ using 
    \begin{itemize}
    \item $\widetilde{O}(sd/\varepsilon)$ queries to $O_f$ if $\H_f(\0)$ is $s$-sparse.
    \item $\widetilde{O}\left(\sqrt{m}d/\varepsilon\right)$ queries to $O_f$ if $\H_f(\0)$ has at most $m$ nonzero entries.
    \end{itemize}
     
\end{thm}

In the above, there is no need to know the positions of nonzero entries if the Hessian is $s$-sparse or contains at most $m$ terms. Classically, without knowing this information, it may still cost $O(d^2)$ queries to estimate $\H_f(\0)$. 
The above results show that when $s$ and $m$ are small, the cost is linear or even polylogarithmic in the dimension. 
Although the quantum speedup using the finite difference method is limited, our lower bound analysis in Proposition \ref{prop: lower bound of real-valued Hessian estimation} indeed shows that the quantum algorithms given in Theorem \ref{thm: intro new1} are optimal in terms of $d$.
Sparse Hessian matrices are quite common in optimization \cite{fletcher1997computing,d2638b8e-8fdc-3d6d-a044-1384db1d1f3a, coleman1984estimation,coleman1984large,gebremedhin2005color}, so the above results might be helpful in accelerating some classical optimization problems.

In the end, we remark that the function assumption in all quantum algorithms based on the spectral method is the same as in the classical case. The only difference is the oracle assumption. Namely, we have to use $O_{f_1}, O_{f_2}$ in the quantum case. However, it is not clear to us if better classical algorithms can be found under this oracle assumption.

\subsection{Summary of the key ideas}

Below, we summarize our key idea for gradient and Hessian estimation. To facilitate understanding, it is helpful to recall the fundamental procedure of previous quantum algorithms for approximating gradients \cite{jordan2005fast, gilyen2019optimizing}. For illustrative purposes, we consider the special case where $f(\x)=\boldsymbol{g}\cdot \x$ is a real-valued linear function for some unknown vector $\boldsymbol{g}=(g_1,\ldots,g_d) \in \mathbb{Z}_N^d$. 
In this case, the gradient of $f(\x)$ is simply $\g$.
Given a phase oracle $O_f:\Ket{\x} \rightarrow e^{2\pi if(\x)/N} \Ket{\x}$, we can generate the quantum state:
$$
\frac{1}{\sqrt{N^d}} \sum_{\x\in \mathbb{Z}_N^d} e^{\frac{2\pi if(\x)}{N}} \Ket{\x}
=
\bigotimes_{j=1}^d \frac{1}{\sqrt{N}} \sum_{x_j\in \mathbb{Z}_N} e^{\frac{2\pi i g_j x_j}{N}} \Ket{x_j}.
$$
We obtain $\Ket{g_1,\ldots,g_d}$ by applying the inverse quantum Fourier transform to the above state. 
This approach is also effective when $f(\x)$ is close to a linear function $\boldsymbol{g}\cdot \x$, as demonstrated in \cite[Theorem 5.1]{gilyen2019optimizing}. Therefore, a core idea of \cite{jordan2005fast,gilyen2019optimizing} is to refine the representations of $\nabla f(\0)$ using higher-order finite difference methods to obtain a function that is close to $\nabla f(\0) \cdot \x$.

Our approach follows the same principle but utilizes the spectral method \cite{lyness1967numerical,trefethen2000spectral}, which is a useful and practical method for gradient estimation in classical computation, to approximate the gradient $\nabla f(\0)$. To elaborate on this, consider a univariate function $f(x)$ that is analytic at point $x_0\in\mb C$ such that the Taylor series of $f$ at $x_0$ converges in an open neighbourhood of $x_0\in \mb C$ containing $\overbar{B}(x_0,r) \subset \mb C$, the closed disk centered at $x_0$ with radius $r \in \mb R_+$. In $\overbar{B}(x_0,r)$, $f(x)$ can be expressed by the Taylor series:
$$
f(x)=\sum_{n=0}^{ \infty} a_n(x-x_0)^n, \quad a_n=\frac{f^{(n)}(x_0)}{n!}.
$$
Then there exists a constant $\kappa$ such that $|a_n| \leq \kappa r^{-n}$ for all $n \in \mathbb{N}$, as detailed in Lemma \ref{lem: bound of derivatives}. Choose some $\delta\in(0,r)$, say $\delta=r/2$, and let $\omega = e^{-2\pi i/N}$. Then, 
$$
f_k :=f (x_0+\delta \omega^k )=\sum_{n=0}^{\infty} a_n (\delta \omega^k )^n
=\sum_{n=0}^{N-1} \omega^{k n} c_n,
$$
where $c_n=\sum_{m=0}^{\infty} a_{n+m N} \delta^{n+m N}$. Applying the inverse of the discrete Fourier transform gives us
$$
c_n=\frac{1}{N} \sum_{k=0}^{N-1} \omega^{-k n} f_k.
$$
From this, we obtain the bound:
$$
\left|a_n-\frac{c_n}{\delta^n}\right| \leq \kappa r^{-n} \sum_{m=1}^{\infty}(\delta / r)^{m N} 
=\kappa r^{-n} \frac{(\delta / r)^N}{1-(\delta / r)^N}.
$$
Particularly, for $n=1$, 
$$
\left|f'(x_0)-\frac{c_1}{\delta}\right| \leq 
\kappa r^{-1} \frac{(\delta / r)^N}{1-(\delta / r)^N}.
$$
When $N=O(\log(\kappa/\varepsilon r))$, we can approximate $f'(x_0)$ up to an additive error $\varepsilon$.

For a multivariable function $f(\x)$, to approximate $\nabla f(\0)$, we consider $h(\tau):= f(\tau \x)$ and view it as a univariate analytic function of $\tau$. It follows that $h'(0)=\nabla f(\0)\cdot \x$. Using the above analysis, which will be detailed further in Theorem \ref{thm: spectral formula for gradient}, the spectral method enables us to derive a real-valued function $F(\x)$ with a form similar to $c_1$, which can be used to approximate $\nabla f(\0) \cdot \x$.
Additionally, we can construct the phase oracle $\mathrm{O}:|\boldsymbol{x}\rangle \rightarrow e^{2 \pi i 2^{n_{\varepsilon}} F(\boldsymbol{x})}|\boldsymbol{x}\rangle$ using 
$\widetilde{O}\left( 1/\varepsilon\delta + N \right)$ queries to $O_{f_1}, O_{f_2}$, as described in Lemma \ref{lem: phase oracle conversion in spectral method}. Combining this with the previous idea described at the beginning of this subsection, we obtain a quantum algorithm that estimates $\nabla f(\0)$.

Regarding Hessian estimation (i.e., the quantum algorithms mentioned in Theorems \ref{thm: intro, Hessian using spectral} and \ref{thm: intro, Hessian using finite difference}), our basic idea is to reduce the task to gradient estimation. For illustrative simplicity, assume $f(\x)=\Bra{\x} H \Ket{\x}$ for some unknown symmetric Boolean Hermitian matrix $H$, which is the Hessian of $f$. Given a phase oracle to query $f$ on a quantum computer, a basic idea for learning $H$ is to use Bernstein-Vazirani's algorithm \cite{bernstein1993quantum}, by considering $g(\x):=\frac{1}{2}(f(\x+\y)-f(\x)) = \Bra{\x} H \Ket{\y} + \frac{1}{2}\Bra{\y} H \Ket{\y}$ for a fixed $\y$. This defines a linear function of $\x$. Via Bernstein-Vazirani's algorithm, we can learn $H \Ket{\y}$ with $1$ query to $g$, equating to $2$ queries to $f$. By selecting $\Ket{\y} = \Ket{i}$ for $i\in [d]$, we can then learn all the columns of $H$. This algorithm for learning $H$ is indeed optimal \cite{montanaro2012quantum}. With the quantum algorithms for gradient estimation, we can use a similar idea to estimate the Hessians. 

The main difficulty is the error analysis caused by the numerical methods. In practice, we only have a function $\tilde{f}(\x)$ that is close to $\x^T H \x$. To use the above idea, we hope that $\frac{1}{2}(\tilde{f}(\x+\y)-\tilde{f}(\x))$ is close to $ \Bra{\x} H \Ket{\y} + \frac{1}{2}\Bra{\y} H \Ket{\y}$ for any fixed $\y$, say $\y=\Ket{i}$. For the spectral method, this is easy to satisfy. However, the finite difference method is based on Taylor expansion, which provides a good approximation only when the norms of $\x,\y$ are small. Choosing a $\y$ with a small norm, such as $\y=\Ket{i}/\sqrt{d}$, will increase the query complexity because we will obtain an approximation of $H\y$ (which is $H\Ket{i}/\sqrt{d}$ in this example, whereas what we really want is an approximation of $H\Ket{i}$). It is hard to avoid this due to the nature of the finite difference method.
This is the main reason why the query complexity of the quantum algorithm based on the finite difference method is higher than that based on the spectral method, as seen in Theorems \ref{thm: intro, Hessian using spectral} and \ref{thm: intro, Hessian using finite difference}. 
However, it is important to note that this does not imply that the finite difference method has reached its optimal performance in Hessian estimation. While our approach shows that the spectral method outperforms the finite difference method, it remains possible that more efficient quantum algorithms could be developed using other techniques based on the finite difference method.

When the Hessian is promised to be $s$-sparse, then randomly generating $k:=\widetilde{O}(s)$ samples $\y_1,\ldots,\y_k$ is sufficient to determine $H$. Now we can determine $H$ by solving a linear system $HY=B$, where $Y=[\y_1,\ldots,\y_k]$ and $B$ is the output matrix of gradient estimation with input $Y$. With this idea, we can propose two new quantum algorithms presented in Theorems \ref{thm: intro new2} and \ref{thm: intro new1}. Due to the reason mentioned above, even in the sparse case, there is only a saving of $\sqrt{d}$ for the finite difference method. 
Similarly, in the sparse case, more efficient approaches based on the finite difference method may also exist.
Moreover, it is worth noting that our algorithm based on the spectral method achieves a result, with a query complexity of $\widetilde{O}(s/\varepsilon)$, which is independent of $d$.

\subsection{Comparison with Gily\'{e}n-Arunachalam-Wiebe's algorithm}
\label{subsec: intro, comparison}
In this subsection, we discuss and highlight the similarities and differences between the Gily\'{e}n-Arunachalam-Wiebe (hereafter abbreviated as GAW) algorithm and our algorithm for gradient estimation. Below, $G_n^d$ represents the grid defined in Definition \ref{def: grid Gn}.

\subsubsection{Similarity}

Both GAW's algorithm and our algorithm are based on the following Lemma \ref{lem: intro, GAW's gradient estimation thm}, which arises from Jordan's quantum algorithm for gradient estimation \cite{jordan2005fast}. For more details, refer to \cite[Theorem 5.1]{gilyen2019optimizing}. 

\begin{lem}\label{lem: intro, GAW's gradient estimation thm}
    Let $c \in \mb R, a, \rho, \varepsilon<M \in \mathbb{R}_{+}$, and $\x_0, \boldsymbol{g} \in$ $\mathbb{R}^d$ such that $\|\g\|_{\infty} \leq M$. Let $n_{\varepsilon}=\left\lceil\log _2(4 /a \varepsilon)\right\rceil$, $n_M=\left\lceil\log _2(3 a M)\right\rceil$, and $n=n_{\varepsilon}+n_M$. Suppose $\tilde{f}:\left(\boldsymbol{x}_0+a G_n^d\right) \rightarrow \mathbb{R}$ satisfies:
    \begin{equation}\label{eq: intro, error condition}
        \left|\tilde{f}(\boldsymbol{x}_0+a \boldsymbol{x})-\boldsymbol{g} \cdot a \boldsymbol{x} - c\right| \leq \frac{\varepsilon a}{8 \cdot 42 \pi}
    \end{equation}
    for all but a $1/1000$ fraction of the points $\boldsymbol{x} \in G_n^d$. 
    Given access to a phase oracle $\mathrm{O}:|\boldsymbol{x}\rangle \rightarrow e^{2 \pi i 2^{n_{\varepsilon}} \tilde{f}(\boldsymbol{x}_0+a \boldsymbol{x})}|\boldsymbol{x}\rangle$ acting on $\mathcal{H}=\operatorname{Span}\left\{|\boldsymbol{x}\rangle: \boldsymbol{x} \in G_n^d\right\}$, there is a quantum algorithm that can calculate a vector $\tilde{\boldsymbol{g}} \in \mathbb{R}^d$ such that
    $$
    \operatorname{Pr}\Big[\|\tilde{\boldsymbol{g}}-\boldsymbol{g}\|_{\infty}>\varepsilon\Big] \leq \rho,
    $$
    with $O\left(\log(d/\rho)\right)$ queries to $\mathrm{O}$.
\end{lem}

With this result, it becomes clear that for any function $f$, as long as we can construct a function $\tilde{f}$ based on $f$ such that condition (\ref{eq: intro, error condition}) holds for $\boldsymbol{g}=\nabla f(\x_0)$, then we can apply this lemma to approximate the gradient $\nabla f(\x_0)$.

\subsubsection{Difference 1: Approximation methods}
The first and most significant difference is the choice of functions used to approximate $\nabla f(\0)\cdot \x$ such that condition (\ref{eq: intro, error condition}) is satisfied. Without loss of generality, we set $\x_0=\0$ for convenience. The GAW algorithm deals with real-valued functions $f:\mb R^d\rightarrow \mb R$ that satisfy some specific smoothness conditions. They used the following degree-$2m$ central difference approximation:
\be
\label{intro: formula 1}
f_{(2m)}(\x)=\sum_{\substack{\ell=-m \\ \ell \neq 0}}^m \frac{(-1)^{\ell-1}}{\ell} \frac{\binom{m}{|\ell|}}{\binom{m+|\ell|}{|\ell|}} f(\ell \boldsymbol{x}) \approx \nabla f(\mathbf{0}) \cdot \boldsymbol{x} .
\ee
It was shown that
$$
|f_{(2m)}(\x) - \nabla f(\0) \cdot \x| \leq e^{-m/2} B \|\x\|^{2m+1},
$$
where $\left|\partial_{\boldsymbol{r}}^{2 m+1} f(\tau \boldsymbol{x})\right| \leq B$ for $\boldsymbol{r}=\boldsymbol{x} /\|\boldsymbol{x}\|$ and all $\tau \in[-m, m]$. 
The directional derivative $\partial_{\boldsymbol{r}}$ is defined in Definition \ref{defn: Directional derivative}.
Particularly, if $\left|\partial^\alpha f\right| \leq c^k k^{k / 2}$ for all $k\in\mb N$ and $\alpha\in \mb N_0^d$ with $|\alpha|=k$, then the error above is bounded by
$$
\sum_{k=2 m+1}^{\infty}\left(8 a c m \sqrt{d}\right)^k
$$
for all but a 1/1000 fraction of points $\boldsymbol{x} \in a G_n^d$. To satisfy condition (\ref{eq: intro, error condition}), 
it suffices to choose $a$ such that $a^{-1}=O(c m \sqrt{d}(c m \sqrt{d} / \varepsilon)^{1 / 2 m})$. This makes the query complexity of their algorithm $\widetilde{O}(\sqrt{d}/\varepsilon)$.

In comparison, our algorithm differs in that it requires sampling function values within the complex domain, i.e., evaluating some $f(\z)$ for $\z\in\mb C^d$. To be more specific, our method deals with analytic functions that can evaluate points in the complex plane near $\0$. Moreover, these functions map reals to reals.
Indeed, all we need is $\nabla f(\mb R^d)\subseteq \mb R$, which ensures the implementation of the quantum Fourier transform as an indispensable part of the algorithm. 
The key finding is the following formula for approximating gradients, derived using the spectral method:
\begin{equation}\label{eq: intro, spectral formula}
    F(\boldsymbol{x}) = \frac{1}{N \delta}\sum\limits_{k=0}^{N-1} \omega^{-k} f(\delta \omega^k \boldsymbol{x}), \quad \text{where } \omega=e^{-2\pi i /N}.
\end{equation}
For further details, see Theorem \ref{thm: spectral formula for gradient}, where we prove that the error bound is given by
$$
    |F(\boldsymbol{x})- \nabla f(\boldsymbol{0}) \cdot \boldsymbol{x}| \leq \frac{\kappa}{2r} \frac{(\delta /2r)^N}{1-(\delta /2r)^N}    
$$
for some constant $\kappa\in \mb R_+$. As analyzed in Theorem \ref{thm: gradient estimation using spectral method}, by selecting $\delta<2r$ and $N=O\left(\log (\kappa/\varepsilon)\right)$, we can ensure that the above error is bounded by the right-hand side of condition (\ref{eq: intro, error condition}). We also analyze the number of queries to $O_{f_1},O_{f_2}$ required to implement the phase oracle $\mathrm{O}:|\boldsymbol{x}\rangle \rightarrow e^{2 \pi i 2^{n_{\varepsilon}} F(\boldsymbol{x})}|\boldsymbol{x}\rangle$ in Lemma \ref{lem: phase oracle conversion in spectral method}. Finally, the overall query complexity of our algorithm is $\widetilde{O}(1/\varepsilon)$.

\begin{figure}[h]
    \centering

\begin{tikzpicture}
\draw[thick] (0,0) circle (2cm);
\draw[thick,->] (-4,0) -- (4,0) 
node[anchor=north west] {};
\draw[thick,->] (0,-3) -- (0,3) 
node[anchor=south] {};

\draw[thick] (1.414,1.414)
node[anchor=south west] {$\delta \omega^2 x$};
\draw[thick] (1.847,0.765)
node[anchor=south west] {$\delta \omega x$};
\draw[thick] (1.847,-0.765)
node[anchor=north west] {$\delta \omega^{N-1} x$};

\draw[thick] (0,0) node[anchor=south west] {0};

\draw[thick] (3.7,2.5) node[anchor=west] {Spectral method (our):};
\draw[thick] (3.7,1.5) node[anchor=west] {$\displaystyle \frac{1}{N \delta}\sum\limits_{k=0}^{N-1} \omega^{-k} f({\color{red}\delta \omega^k x}) \approx f'(0) x$};

\draw[thick,-] (3.6,0.7) -- (9,0.7) node[anchor=south] {};
\draw[thick,-] (3.6,0.7) -- (3.6,2.9) node[anchor=south] {};
\draw[thick,-] (3.6,2.9) -- (9,2.9) node[anchor=south] {};
\draw[thick,-] (9,0.7) -- (9,2.9) node[anchor=south] {};

\draw[thick] (3.7,-.8) node[anchor=west] {Finite difference method \cite{gilyen2019optimizing}:};
\draw[thick] (3.7,-2) 
node[anchor=west] {$\displaystyle \sum_{\substack{\ell=-m \\ \ell \neq 0}}^m \frac{(-1)^{\ell-1}}{\ell} \frac{\binom{m}{|\ell|}}{\binom{m+|\ell|}{|\ell|}} f({\color{blue}\ell  x}) \approx f'(0) x $};

\draw[thick,-] (3.6,-0.4) -- (10,-0.4) node[anchor=south] {};
\draw[thick,-] (3.6,-0.4) -- (3.6,-2.9) node[anchor=south] {};
\draw[thick,-] (3.6,-2.9) -- (10,-2.9) node[anchor=south] {};
\draw[thick,-] (10,-0.4) -- (10,-2.9) node[anchor=south] {};

\filldraw [black] (0,0) circle (3pt);

\filldraw [red] (2,0) circle (2pt);
\filldraw [red] (-2,0) circle (2pt);
\filldraw [red] (0,2) circle (2pt);
\filldraw [red] (0,-2) circle (2pt);

\filldraw [red] (1.414,1.414) circle (2pt);
\filldraw [red] (1.414,-1.414) circle (2pt);
\filldraw [red] (-1.414,1.414) circle (2pt);
\filldraw [red] (-1.414,-1.414) circle (2pt);

\filldraw [red] (1.847,0.765) circle (2pt);
\filldraw [red] (-1.847,-0.765) circle (2pt);
\filldraw [red] (0.765,1.847) circle (2pt);
\filldraw [red] (-0.765,-1.847) circle (2pt);

\filldraw [red] (1.847,-0.765) circle (2pt);
\filldraw [red] (-1.847,0.765) circle (2pt);
\filldraw [red] (-0.765,1.847) circle (2pt);
\filldraw [red] (0.765,-1.847) circle (2pt);

\filldraw [blue] (0.8,0) circle (2pt);
\filldraw [blue] (-0.8,0) circle (2pt);
\filldraw [blue] (1.6,0) circle (2pt);
\filldraw [blue] (-1.6,0) circle (2pt);
\filldraw [blue] (2.4,0) circle (2pt);
\filldraw [blue] (-2.4,0) circle (2pt);
\filldraw [blue] (3.2,0) circle (2pt);
\filldraw [blue] (-3.2,0) circle (2pt);

\draw[thick] (1.6,0) 
node[anchor=north] {$2x$};
\draw[thick] (0.8,-0.06) 
node[anchor=north] {$x$};

\draw[thick] (-1.6,0) 
node[anchor=north] {$-2x$};
\draw[thick] (-0.8,0) 
node[anchor=north] {$-x$};

\end{tikzpicture}

\caption{Comparison between GAW's formula and ours.
The figure illustrates the points that are used to approximate the derivative for functions of dimension $d=1$. The blue points on the $x$ axis are used in the higher-order finite difference method and the red points on the circle are used in the spectral method.}
\label{fig1}
\end{figure}
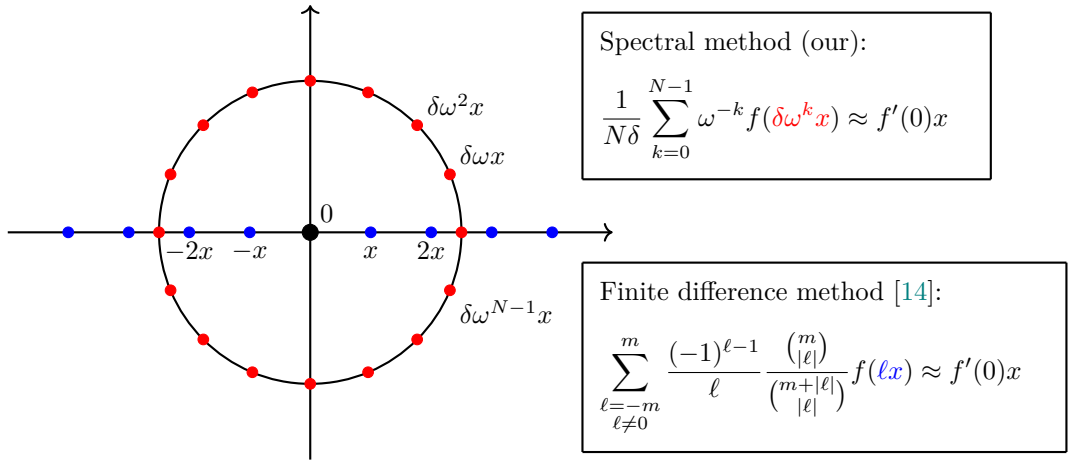

In contrast to the commonly used finite difference methods, the spectral method offers better error performance in approximating derivatives. In our opinion, one of the key reasons for this improvement is the way of selecting sampling points. As illustrated in Figure \ref{fig1}, for $1$-dimensional functions, the finite difference method samples points along a straight line, while the spectral method utilizes points on a circle around the target point $0$, intuitively demonstrating the advantage, as the derivative is a local property.
Besides, as we can see from formula (\ref{intro: formula 1}), the coefficients associated with sampling points further from the target point are smaller in the finite difference method. This implies that these points contribute less to the approximation of the derivative at the target point. In contrast, the spectral method exhibits coefficients with the same magnitude for all sampling points, indicating a more efficient utilization of information from all points. This also demonstrates the superiority of the spectral method.

\begin{rmk}

We believe that the exponential speedup achieved here can be significantly attributed to the spectral method. This method drastically reduces the error in approximating the derivative of a function at some points using the function values at sampled points, compared to the finite difference formulae. 
To our knowledge, the application of the spectral method in quantum algorithms for gradient estimation is novel. 
The advantage of the spectral method has also been observed in other applications, such as solving differential equations on a quantum computer, where the spectral method \cite{childs2020quantum} has demonstrably improved complexity results compared to the finite difference method \cite{berry2014high}.
\end{rmk}


\subsubsection{Difference 2: Oracle settings}
The second difference is the setting of oracles. 
The GAW algorithm considers real-valued functions and naturally assumes quantum access to the phase oracle $O_f:\Ket{\x} \rightarrow e^{i f(\x)} \Ket{\x} $.
In contrast, our approach deals with complex-valued functions, requiring evaluation at points in the complex domain. Thus, the standard phase oracle $O_f$ is no longer directly applicable, since it is not a unitary operator. Instead, we assume access to phase oracles for the real and imaginary parts of $f(\x)=f_1(\x)+if_2(\x)$, denoted as $O_{f_1}, O_{f_2}$, respectively, as outlined in (\ref{intro: phase oracle}). 
Similar to classical computing, the real and imaginary parts of a complex value are stored separately as their finite-precision binary encoding. Fortunately, our formula given in (\ref{eq: intro, spectral formula}) satisfies $F(\x)\in \mb R$, and we can construct the corresponding phase oracle $O_F$ using the given oracles $O_{f_1}, O_{f_2}$.
Consequently, Lemma \ref{lem: intro, GAW's gradient estimation thm} still works. It is worth noting that while our oracle assumptions differ from those in the GAW algorithm, they remain standard and natural, akin to how data is stored in classical computing.


Our oracle assumption is stronger, as in many practical applications, $f(\x)$ is real-valued and analytic in the real field.
Therefore, the standard oracle assumption is $O_f$ to query real $\x$.
When $\x\in \mathbb{C}^d$, it is usually not straightforward to construct $O_{f_1}, O_{f_2}$ from $O_f$.
However, in some cases, this is possible. A trivial example is linear functions. Moreover, for the class of functions and oracle settings we consider, we found nontrivial examples with practical relevance, as described in the following remark.

\subsection{Relevant previous results}


Apart from the two references \cite{jordan2005fast, gilyen2019optimizing} we mentioned previously, there are also many related works focusing on optimizing the quantum computation of gradients and Hessians in specific application contexts.
For instance, reference \cite{teo2023optimized} provided optimized methods for estimating gradients and Hessians in variational quantum algorithms.
Reference \cite{gao2021quantum} presented a quantum gradient algorithm for general polynomials, showcasing the potential of quantum approaches in more specialized mathematical landscapes.
A recent work \cite{apers2023quantum} provided a quantum algorithm for linear programming, which is based on interior point methods and requires efficient approximation of the Hessian and gradient of the barrier function. They provided efficient methods for computing these specific forms of the Hessian and gradient. 
Reference \cite{rebentrost2019quantum} developed quantum versions of gradient descent and Newton's method, and applied them to polynomial optimization tasks.
Reference \cite{kerenidis2020quantum} provided a quantum method for performing gradient descent and used it to speed up the solving of weighted least-squares problems.
Besides, there are related works that investigate the problem from a hardware perspective \cite{schuld2019evaluating, mari2021estimating}.

\subsection{Outline of the paper}

The rest of the paper is organized as follows:
In Section \ref{section: Preliminaries}, we present some preliminary results that will be used throughout the paper.
In Section \ref{section: Quantum spectral method for gradient estimation}, we describe our main quantum algorithm for gradient estimation.
In Section \ref{section: Quantum algorithms for Hessian estimation using spectral method}, we present two quantum algorithms for Hessian estimation using spectral method.
In Section \ref{section: Quantum algorithms for Hessian estimation using finite difference formula}, we present two quantum algorithms for Hessian estimation using finite difference method.
In Section \ref{section: Lower bounds analysis}, we prove some lower bounds of Hessian estimation.
Appendix \ref{appendix: Error analysis of finite difference formula} contains some detailed calculations required in the Hessian estimation.

\section{Preliminaries}
\label{section: Preliminaries}

\textbf{Notation.} For $n\in \mb N$, $[n]$ denotes the set $\{1, 2, \ldots, n\}$. $\mb N_0 = \mb N \cup \{0\}$. 
By $\mathbb{R}_+$, we mean the set of all positive real numbers.
Let $\mb K$ denote the field $\mb R$ or $\mb C$,\footnote{When we use $\mb K$, it implies that the statements apply to both fields of real and complex numbers.} and let $d\in \mb N$ denote the dimension of the space on which function $f: \mb K^d\rightarrow \mb K$ is defined. We use bold letters for vectors, in particular, $\0$ denotes the all-0 vector. The set $\mb K^d$ is equipped with the usual vector space structure. For vectors $\x=(x_1,\ldots,x_d) \in \mb K^d$, we write $\|\x\|=\sqrt{|x_1|^2+\cdots+|x_d|^2}$, and $\|\x\|_{\infty}=\max_{i\in [d]} |x_i|$, where $|x|$ represents the absolute value for $x\in\mb K$. The max norm of matrix $A=\left(a_{ij}\right)$ is $\|A\|_{\max}=\max_{i,j}|a_{ij}|$. With $A^\dag$, we mean the complex conjugate transpose of $A$.

For $n\in \mb N$, an $n$-dimensional multi-index is written as $\nu=(\nu_1, \nu_2, \ldots, \nu_n) \in \mb N_0^n$. The factorial of $\nu$ is $\nu!=\nu_{1}!\nu_{2}!\cdots \nu_{n}!$, and the absolute value is $|\nu|=\nu_1+\nu_2+\cdots +\nu_n$. Given a multi-index $\alpha=(\alpha_1, \ldots, \alpha_d)\in \mb N_0^d$ and $\x=(x_1,\ldots,x_d)\in \mb K^d$, we write $\x^\alpha$ for $x_1^{\alpha_1} x_2^{\alpha_2} \cdots x_d^{\alpha_d}$. For variables $\z=(z_1,\ldots,z_d)\in\mb K^d$, the higher-order partial derivative is denoted as $\partial^{\alpha}=\frac{\partial^{|\alpha|}}{\partial z_1^{\alpha_1}\cdots \partial z_d^{\alpha_d}}$. In other words, $\partial^\alpha=\partial_1^{\alpha_1} \partial_2^{\alpha_2} \ldots \partial_d^{\alpha_d}$, where $\partial_i^{\alpha_i}:=\partial^{\alpha_i} / \partial z_i^{\alpha_i}$.

\subsection{Description of oracles}

\begin{defi}[Binary oracle for complex-valued functions] \label{def: complex binary oracle}
    For $\eta \in \mb R_+$, let $f: \mb C^d \rightarrow \mb C$ be a function with an $\eta$-accurate binary oracle access acting as
    $$
    U_f^{\eta}: |\boldsymbol{x}_1\ket |\boldsymbol{x}_2\ket |\boldsymbol{0}\ket \rightarrow |\boldsymbol{x}_1\ket |\boldsymbol{x}_2\ket |\tilde{f_1}(\boldsymbol{x})\ket |\tilde{f_2}(\boldsymbol{x})\ket,
    $$
    where input $\boldsymbol{x}=\boldsymbol{x}_1+i\boldsymbol{x}_2\in \mb C^d$ and function value $f(\boldsymbol{x})=f_1(\boldsymbol{x})+if_2(\boldsymbol{x})\in \mb C$, the real numbers $f_1(\boldsymbol{x}), f_2(\boldsymbol{x})$ are stored as  their finite-precision binary encodings with precision $\eta$, which means
    $$
    |f_1(\boldsymbol{x})-\tilde{f_1}(\boldsymbol{x})|\leq \eta, \quad |f_2(\boldsymbol{x})-\tilde{f_2}(\boldsymbol{x})|\leq \eta.
    $$
    Sometimes we omit the $\eta$ and write $U_f$. We denote the cost of one query to $U_f^\eta$ as $C(\eta)$.\footnote{The cost function $C(\eta)$ depends on the specific method to obtain $U_f^\eta$. For instance, if $f$ can be calculated using classical circuits, $C(\eta)$ is $\on{polylog} (1/\eta)$, however, the cost is typically $1/\eta$ if $U_f^\eta$ is obtained using quantum amplitude estimation.}
\end{defi}

\begin{rmk}
    In this paper, we store the real and imaginary parts of all complex numbers $z\in \mb C$ separately. However, for clarity and conciseness, we denote them as $|z\ket$, representing actually $|z_1\ket |z_2\ket$ for $z=z_1+iz_2$. The same representation applies to complex vectors $\boldsymbol{z}=\boldsymbol{z}_1+i\boldsymbol{z}_2 \in \mb C^d$.
\end{rmk}

\begin{lem}[Some arithmetic operations of complex values] \label{lem: complex arithmetic}
For $a, b, c \in \mb C$, we can implement unitaries satisfying
\be
|a\ket |b\ket \rightarrow |a\ket |b+ac\ket.
\label{comp_arith_eq}
\ee
\end{lem}
\begin{proof}
    By \cite{ruiz2017quantum}, we can do the common arithmetic operations for binary encoding of real numbers. Using these operations as a foundation, we can make some adjustments to handle the case of complex numbers. Assume $a=a_1+ia_2, b=b_1+ib_2$ and $c=c_1+ic_2$, then the following is a process of implementing (\ref{comp_arith_eq}) because $b+ac=(b_1+a_1c_1-a_2c_2)+i(b_2+a_1c_2+a_2c_1)$:
    $$
    \begin{aligned}
        |a_1\ket |a_2\ket |b_1\ket |b_2\ket &\rightarrow |a_1\ket |a_2\ket |b_1+a_1 c_1\ket |b_2\ket \\
        &\rightarrow |a_1\ket |a_2\ket |b_1+a_1 c_1-a_2 c_2\ket |b_2\ket \\
        &\rightarrow |a_1\ket |a_2\ket |b_1+a_1 c_1-a_2 c_2\ket |b_2+a_1 c_2\ket \\
        &\rightarrow |a_1\ket |a_2\ket |b_1+a_1 c_1-a_2 c_2\ket |b_2+a_1 c_2+a_2 c_1\ket.
    \end{aligned}
    $$
\end{proof}

\begin{defi}[Phase oracle]
\label{def: phase oracle}
    Let $\eta\in\mathbb{R}_+$ and $f: X\rightarrow [-\pi,\pi]$ be a real-valued function defined on a set $X$. The phase oracle for $f$ acts as
    $$
    O_f^\eta: |x\ket \rightarrow e^{i\tilde{f}(x)}|x\ket,
    $$
    where $|\tilde{f}(x)-f(x)|\leq \eta$ for all $x\in X$. Sometimes, we will ignore $\eta$ and just write $O_f$. If $f$ is complex-valued with real and imaginary parts $f_1, f_2$, we can similarly define $O_{f_1}$ and $O_{f_2}$ to query $f$.
\end{defi}

The query model is commonly used in the field of quantum query complexity. It is widely known that the binary oracle and phase oracle for real-valued functions can be converted to each other, using the phase-kickback technique and quantum phase estimation. In the case of complex-valued functions, we have a similar result.

\begin{prop}[Binary oracle and phase oracle]
\label{prop: Binary oracle and phase oracle}
    Let $f: \mb C^d \rightarrow \mb C$ be a function.

\begin{itemize}
    \item Given $U_f^\eta$, we can obtain the phase oracles $O_{f_1}^\eta, O_{f_2}^\eta$ for its real and imaginary parts with $2$ applications of $U_f^\eta$. 
    \item Conversely, suppose $f_1(\x), f_2(\x) \in[-\pi,\pi]$ for all $\x\in \mb C^d$, then given $O_{f_1}^\eta, O_{f_2}^\eta$, we can obtain $U_f^{\eta+\varepsilon}$ with $O(1/\varepsilon)$ applications of $O_{f_1}^\eta, O_{f_2}^\eta$.
\end{itemize}

\end{prop}

\begin{proof}
    Note that for any $y\in\mb R$ represented as binary expansion, it is easy to implement $|y\ket \rightarrow e^{iy}|y\ket$ using controlled phase gates. Therefore, given
    $
    U_f^\eta: |\boldsymbol{x}\ket |\0\ket \rightarrow |\boldsymbol{x}\ket |\tilde{f_1}(\boldsymbol{x})\ket |\tilde{f_2}(\boldsymbol{x})\ket
    $
    for $\boldsymbol{x} \in \mb C^d$, we can obtain $O_{\tilde{f_1}}, O_{\tilde{f_2}}$ in this way with the same accuracy. Each is obtained with 2 applications of $U_f^\eta$.
    
    Conversely, suppose that we are given phase oracles. Following the process of quantum phase estimation, the phase oracle conjugated with the quantum Fourier transform (QFT) acting on the output register can give us a binary query. Note that the QFT can be performed properly since $f_1(\boldsymbol{x}), f_2(\boldsymbol{x}) \in \mb R$. The complexity is the same as that of quantum phase estimation.
\end{proof}

\begin{defi}[Fractional query oracle] \label{def: fractional oracle}
Let $r\in[-1,1], \eta \in \mathbb{R}_+$ and $f: X\rightarrow [-\pi,\pi]$ be a real-valued function defined on a set $X$.  The fractional query oracle $O_{rf}$ is defined as
\[
O_{rf}^\eta: \Ket{x} \rightarrow e^{i r \tilde{f}(x)} \Ket{x},
\]
where
$|\tilde{f}(x) - f(x)|\leq \eta$ for all $x\in X$. As usual, we sometimes ignore $\eta$ when it makes no difference.
\end{defi}

The following result follows from Circuit 4.2.18 in \cite{cornelissen2019quantum} (also see \cite[Corollary 34]{gilyen2019quantum}).

\begin{prop}[Phase oracle to fractional oracle]
\label{lemma: phase to fractional}
Assume that $\|f\|_{\infty} \leq 1/2$, $r\in [0,1/48]$ and $\eta>0$, then there is a quantum circuit that implements $O_{rf}$ up to error $\eta$ using $O(\log(1/\eta))$ applications of $O_f$ and its controlled form.
\end{prop}

If $1/48 < r < 1$ and $r$ is an integer multiple of $r'\in[0,1/48]$, then we can apply the above result to construct $O_{rf}$ by considering it as $r/r'$-th power of $U_{r'f}$. So the query complexity will be $O((r/r')\log(1/\eta))$.

\begin{rmk}
Although $U_f, O_f$ are commonly used oracles in the design of quantum algorithms, another more practical oracle is the probability oracle. Assume $f:\mathbb{R}^d \rightarrow[0,1]$, the probability oracle acts as 
\be
P_f: \Ket{x} \Ket{0} \rightarrow \Ket{x} \left( \sqrt{f(x)}\Ket{0} \Ket{\psi_0} + \sqrt{1-f(x)}\Ket{1} \Ket{\psi_1} \right),
\ee
where $\Ket{\psi_0}, \Ket{\psi_1}$ are arbitrary (normalized) quantum states. The right-hand side of the above state appears widely in many quantum algorithms. 
As shown in \cite[Corollary 4.1]{gilyen2019optimizing}, there is an effective procedure to convert a probability oracle to a phase oracle. This result will be used below, so state it here.
\end{rmk}

\begin{lem}[Probability oracle to phase oracle]
\label{lem: Probability oracle to phase oracle}
Suppose that we are given a probability oracle $P_g$ for $g(x):X \rightarrow [0,1]$. Let $t \in \mathbb{R}$ and $f(x)=t g(x)$ for all $x \in X$. We can implement an $\varepsilon$-approximate phase oracle $O_f$ with query complexity $O(|t| + \log(1/\varepsilon))$, i.e., this many uses of $P_g$ and its inverse.
\end{lem}

Conversely, as shown in \cite[Lemma 16 (the arXiv version)]{gilyen2019optimizing}, we have the following conversion from phase to probability oracle.

\begin{lem}[Phase oracle to probability oracle]
\label{lem: Phase oracle to probability oracle}
Let $\varepsilon, \delta\in(0,1/2)$ and $f:X \rightarrow [\delta,1-\delta]$. Given access to the phase oracle $O_f$, then we can implement a probability oracle $P_f$ up to accuracy $\varepsilon$ using $O(\delta^{-1} \log(1/\varepsilon))$ invocations of (controlled) $O_f$ and $O_f^\dag$.
\end{lem}

For a complex-valued function  $f(\x)=f_1(\x) + i f_2(\x)$, an alternative representation is using polar coordinates, i.e., $f(\x) = r(\x) e^{i\theta(\x)}$. So it is natural to consider two other phase oracles:
\be
O_r: \Ket{\x} \rightarrow e^{ir(\x)} \Ket{\x}, \quad 
O_\theta: \Ket{\x} \rightarrow e^{i\theta(\x)} \Ket{\x}.
\ee
We below show that given $O_r, O_\theta$, we can efficiently construct $O_{f_1}, O_{f_2}$. This implies that the oracle assumption on $O_r, O_\theta$ is stronger.

\begin{prop}
\label{prop: polar oracles to standard oracles}
Let $\varepsilon, \delta \in (0,1/2)$.
Assume $r(\x), \cos(\theta(\x)), \sin(\theta(\x)) \in [-1+2\delta,1-2\delta]$ for all $\x$.
Then given $O_r, O_\theta$, we can construct $O_{f_1}, O_{f_2}$ up to error $\varepsilon$ with $O(\delta^{-1}\log^4(1/\varepsilon))$ applications of $O_r, O_\theta$, including $O_r^\dag, O_\theta^\dag$ and their controlled forms.
\end{prop}

To prove this claim, we need to invoke some results on block-encoding \cite{gilyen2019quantum}.

\begin{defi}[Block-encoding]
Let $A$ be an operator and $\alpha\geq \|A\|$, then a unitary $U$ is called a $\alpha$-block-encoding of $A$ if it has the form
\[
U = \begin{pmatrix}
A/\alpha & * \\
*  & *
\end{pmatrix}.
\]
\end{defi}

Intuitively, $U$ is a block-encoding of $A$ if for any state $\Ket{\psi}$, we have $U \Ket{0}\Ket{\psi} = \Ket{0} \otimes (A/\alpha) \Ket{\psi} + \Ket{0}^\bot$ for some orthogonal part $\Ket{0}^\bot$. Sometimes, more than 1 ancillary qubits are required.

Given block-encodings of operators, the linear combination of unitaries (LCU) is a useful technique to construct block-encoding of their linear combinations. For example, for $i\in\{0,1\}$, let $U_i$ be a block-encoding of $A_i$ with parameter $\alpha_i=1$, then by LCU, it is easy to check that the following process constructs a block-encoding of $A_0+A_1$ with $\alpha=2$:
\begin{enumerate}
    \item Prepare $\Ket{+} \Ket{0,\psi} = \frac{1}{\sqrt{2}} \Ket{0} \Ket{0,\psi} + \frac{1}{\sqrt{2}} \Ket{1} \Ket{0,\psi}$.
    \item Apply $\Ket{0} \Bra{0} \otimes U_0 + \Ket{1} \Bra{1} \otimes U_1$ to obtain
     $\frac{1}{\sqrt{2}} \Ket{0} U_0\Ket{0,\psi} + \frac{1}{\sqrt{2}} \Ket{1} U_1\Ket{0,\psi}$.
    \item Apply Hadamard gate to the first register to obtain 
    $\frac{1}{2} \Ket{0} (U_0+U_1) \Ket{0,\psi} + \frac{1}{2} \Ket{1} (U_0-U_1)\Ket{0,\psi}$.
\end{enumerate}

Given a block-encoding of a Hermitian matrix, Low and Chuang showed how to implement an optimal Hamiltonian simulation.

\begin{lem}[Hamiltonian simulation, see \cite{low2019hamiltonian}]
\label{lem: Hamiltonian simulation}
Suppose that $U$ is a block-encoding of the Hamiltonian $H$ with parameter $\alpha$. Then one can implement an $\varepsilon$-precise approximation of the Hamiltonian simulation unitary $e^{i t H}$ with $O(\alpha |t| + \log(1/\varepsilon))$ uses of controlled-$U$ or its inverse.
\end{lem}

\begin{lem}
\label{lem: product of phase oracles}
Let $\varepsilon,\delta\in(0,1/2)$ and $f,g:X\rightarrow [-1+2\delta, 1-2\delta]$. Given phase oracles $O_f, O_g$, we can construct $O_{fg}$ up to accuracy $\varepsilon$ with $O(\delta^{-1} \log^3(1/\varepsilon))$ calls to $O_f, O_g$.
\end{lem}

\begin{proof}
Let $f'=(1-f)/2, g'=(1-g)/2$. From $O_f, O_g$, it costs $O(\log(1/\varepsilon))$ to construct $O_{f'}, O_{g'}$ up to accuracy $\varepsilon$ by Lemma \ref{lemma: phase to fractional}. Moreover, $f',g':X\rightarrow [\delta, 1-\delta]$.
By Lemma \ref{lem: Phase oracle to probability oracle}, we can construct $P_{f'}$ and $P_{g'}$
\[
P_{f'} : \Ket{\x} \Ket{0} \rightarrow \Ket{\x} \left(\sqrt{f'(\x)} \Ket{0} + \sqrt{1-f'(\x)} \Ket{1} \right),
\]
\[
P_{g'} : \Ket{\x} \Ket{0} \rightarrow \Ket{\x} \left(\sqrt{g'(\x)} \Ket{0} + \sqrt{1-g'(\x)} \Ket{1} \right),
\]
up to accuracy $\varepsilon$ using $O(\delta^{-1} \log(1/\varepsilon))$ queries to $O_{f'}, O_{g'}$. Now, it is easy to check that $P_{f'}^\dag(I\otimes Z)P_{f'}$ is a 1-block-encoding of the ${\rm diag}\{1-2f'(\x)\} = {\rm diag}\{f(\x)\}$. Similarly, $P_{g'}^\dag(I\otimes Z)P_{g'}$ is a 1-block-encoding of ${\rm diag}\{g(\x)\}$. Given these two block-encodings, we can construct a 1-block-encoding of the product $H:={\rm diag}\{f(\x) g(\x)\}$. Finally, we use Lemma \ref{lem: Hamiltonian simulation} to implement $e^{i H}$ up to accuracy $\varepsilon$ with $O(\log(1/\varepsilon))$ applications of the block-encoding. This gives the phase oracle $O_{fg}$.
\end{proof}

We are now ready to prove Proposition \ref{prop: polar oracles to standard oracles}.

\begin{proof}[Proof of Proposition \ref{prop: polar oracles to standard oracles}]
Note that $f_1(\x)=r(\x) \cos(\theta(\x)), f_2(\x)=r(\x) \sin(\theta(\x))$. 
We consider the following process:

\begin{enumerate}
    \item Prepare $\Ket{\x} \Ket{+} = \frac{1}{\sqrt{2}} \Ket{\x} \Ket{0} + \frac{1}{\sqrt{2}} \Ket{\x} \Ket{1}$.
    \item Apply control-$O_\theta$, i.e., $\Ket{0} \Bra{0} \otimes O_\theta + \Ket{1} \Bra{1} \otimes O_\theta^{-1}$, we then obtain
     $\frac{1}{\sqrt{2}} e^{i\theta(\x)} \Ket{\x} \Ket{0} + \frac{1}{\sqrt{2}} e^{-i\theta(\x)} \Ket{\x} \Ket{1}$.
    \item Apply Hadamard gate to the second register to obtain $\Ket{\x} (\cos(\theta(\x)) \Ket{0} + \sin(\theta(\x)) \Ket{1})$.
\end{enumerate}
The above defines a process to implement a probability oracle for the function $\cos(\theta(\x))$. As it acts like a rotation, we denote it as $R$.
By considering $R+R^T$ using LCU, we will obtain a 1-block-encoding of ${\rm diag}(\cos(\theta(\x)))$. Similarly, by considering $X(R-R^T)$ we will have a 1-block-encoding of ${\rm diag}(\sin(\theta(\x)))$. Here $X=\Ket{0}\Bra{1} + \Ket{1}\Bra{0}$ is the Pauli-$X$. From these block-encodings, by Lemma \ref{lem: Hamiltonian simulation}, we can construct phase oracles
\[
O_{1}: \Ket{\x} \rightarrow e^{i\cos(\theta(\x))} \Ket{\x}, \quad 
O_{2}: \Ket{\x} \rightarrow e^{i\sin(\theta(\x))} \Ket{\x}
\]
up to accuracy $\varepsilon$ using $O(\log(1/\varepsilon))$ applications of $O_\theta$.
Finally by Lemma \ref{lem: product of phase oracles}, we can construct the phase oracles of $r(\x) \cos(\theta(\x))$ and $r(\x) \sin(\theta(\x))$ from $O_r, O_1, O_2$.
\end{proof}

\subsection{Multivariable calculus}

This paper extensively utilizes concepts from multivariable analysis. Here, we provide a brief overview of notation, definitions, and properties, to facilitate understanding. For further information, we recommend referring to the following classic textbooks: \cite{krantz2002primer} delves into real analytic functions, and \cite{gunning2022analytic, range1998holomorphic, krantz2001function,laurent2010holomorphic} encompass holomorphic functions in several complex variables.

\begin{defi}[Directional derivative]
\label{defn: Directional derivative}
    Given a function $f:\mb R^d \rightarrow \mb R$ that is $n$-times differentiable at $\x\in \mb R^d$, then the $n$-th order directional derivative along the direction $\boldsymbol{r}\in \mb R^d$ is
    $$
    \partial_{\boldsymbol{r}}^n f(\boldsymbol{x})=\frac{d^n}{d \tau^n} f(\boldsymbol{x}+\tau \boldsymbol{r})=\sum_{\alpha\in \mb N_0^d, |\alpha|=n} \frac{n!}{\alpha!} \boldsymbol{r}^\alpha \cdot \partial^\alpha f(\x).
    $$
\end{defi}

Finite difference formulae, derived from the Taylor expansion, can be used to approximate derivatives of functions. In \cite{li2005general}, Li provided general finite difference formulae and error analysis for higher-order derivatives of functions. Therefore, we can use them to obtain an approximate formula for the Hessian of multivariable functions, as defined below.

\begin{defi}[Finite difference formula for Hessian]
\label{def: finite difference formula of Hessian}
        Let $m\in\mb N$. The degree-$2m$ central difference approximation of the second derivative of a function $f:\mb R^d \rightarrow \mb R$ is
        $$
        f_{(2 m)}(\boldsymbol{x}):=\sum_{\substack{t=-m}}^m  a_{t}^{(2m)} f(t \boldsymbol{x}) \approx \boldsymbol{x}^{\mathrm{T}} \mathbf{H}_{f}(\mathbf{0}) \boldsymbol{x},
        $$
        where $\mbf{H}_f$ is the Hessian matrix of $f$. The coefficients for $t\in \{-m, \dots,m\}\backslash \{0\}$ are
        $$
        a_t^{(2m)}:=\frac{(-1)^{t-1}\cdot 2}{t^2} \frac{m! m!}{(m+t)! (m-t)!}
        $$
        and $a_0^{(2m)}=- \sum\limits_{\substack{t=-m, t \neq 0}}^{m} a_t^{(2m)} $.
\end{defi}

\subsubsection{Holomorphic functions in several complex variables}

In this part, we outline the definitions and related concepts of multivariable complex functions, including their derivatives, analyticity, and other pertinent properties. These topics are foundational in analysis and are provided here for reference, in case some readers are less familiar with the differential properties of multivariable complex functions.

For $d\in \mb N$, we define an isomorphism of $\mb R$-vector spaces between $\mb C^d$ and $\mb R^{2d}$ by setting $z_j=x_j+iy_j$ for all $j\in[d]$. The holomorphic and anti-holomorphic differential operations are given by
$$
\left\{\begin{array}{l} \vspace{.2cm}
\displaystyle \frac{\partial}{\partial z_j}=\frac{1}{2}\left(\frac{\partial}{\partial x_j}-i \frac{\partial}{\partial y_j}\right), \\
\displaystyle \frac{\partial}{\partial \bar{z}_j}=\frac{1}{2}\left(\frac{\partial}{\partial x_j}+i \frac{\partial}{\partial y_j}\right),
\end{array}\right.
$$
for $j\in[d]$.

\begin{defi}[Page 5 of \cite{range1998holomorphic}, Definition 1.1 of \cite{laurent2010holomorphic}]
    Let $\Omega \subset \mb C^d$ be an open set and function $f: \Omega\rightarrow \mb C$. The function $f$ is said to be holomorphic on $\Omega$ if $f$ is differentiable and satisfies the system of partial differential equations
    $$
    \frac{\partial f}{\partial \bar{z}_j}(\boldsymbol{c})=0 \quad \text { for every } \boldsymbol{c} \in \Omega \text { and } j\in [d],
    $$
    which is called the homogeneous Cauchy-Riemann system. As for real-valued functions, we denote $\nabla f(\boldsymbol{c})=\left(\frac{\partial f}{\partial z_1}(\boldsymbol{c}), \ldots, \frac{\partial f}{\partial z_d}(\boldsymbol{c})\right)$.
\end{defi}

\begin{lem}[Hartogs's Theorem, Theorem 2.9 of \cite{laurent2010holomorphic}] 
    Let $f: \Omega\rightarrow \mb C$ be a function defined on $\Omega \subset \mb C^d$. If $f$ is holomorphic, then $f$ is holomorphic in each variable separately. More precisely, the functions $f_j$ defined by $z \mapsto f\left(z_1, \ldots, z_{j-1}, z, z_{j+1}, \ldots, z_d\right)$ are holomorphic as functions of one complex variable, which means the following limit exists\footnote{Note that $z_j\in \mb C$ tends to $c_j$ in all directions in the complex plane, and the limit exists meaning that the values obtained in different directions are all the same.}
    $$
    \frac{\partial f}{\partial z_j}(\boldsymbol{c}):=\lim_{z_j\rightarrow c_j}\frac{f(c_1,\ldots,c_{j-1},z_j,c_{j+1},\ldots,c_d)-f(\boldsymbol{c})}{z_j-c_j}
    $$
    for each $j\in[d]$ and any $\boldsymbol{c}=(c_1,\ldots,c_d)\in \Omega$. Conversely, if $f$ is holomorphic in each variable separately, then Hartogs's theorem\footnote{Under the additional hypothesis that $f$ is continuous, the statement is easier to prove, known as Osgood's lemma.} ensures that the function $f$ itself is holomorphic.
\end{lem}

\subsubsection{Analytic functions}

\begin{defi}[Power series]
    Let $a_{\nu}\in \mb K=\mathbb{R}$ or $\mathbb{C}$ for $\nu=(\nu_1, \ldots, \nu_d)\in \mb N_0^d$. A power series in $d$ variables $\z=(z_1, \ldots, z_d)\in\mb K^d$ centered at the point $\boldsymbol{c}=(c_1, \ldots, c_d)\in \mb K^d$ is a series of the form
    \begin{equation}\label{eq: power series}
        \sum_{\nu \in \mathbb{N}_0^d} a_{\nu} (\z-\boldsymbol{c})^{\nu},
    \end{equation}
    where $(\z-\boldsymbol{c})^{\nu}$ is a multi-index notation representing $(z_1-c_1)^{\nu_1} \cdots (z_d-c_d)^{\nu_d}$.
\end{defi}
The domain of convergence of multivariable power series is not as straightforward as in the single-variable case. Below, we will give some basic results about the convergence.

\begin{defi}[Polydisc]\label{def: polydisc}
    The (open) polydisc $D(\boldsymbol{c},\boldsymbol{r})$ of multi-radius $\boldsymbol{r}=(r_1,\ldots,r_d)\in\mb R_+^d$ and center $\boldsymbol{c}\in \mb K^d$ is
    $$
    D(\boldsymbol{c},\boldsymbol{r})=\{\z=(z_1,\ldots,z_d)\in \mb K^d: |z_j-c_j| < r_j \text{ for all }j=1,\ldots,d \}.
    $$
    The closed polydisc is $\overbar{D}(\boldsymbol{c},\r)=\{\z\in\mb K^d:|z_j-c_j|\leq r_j \text{ for all }j\in[d]\}$.
\end{defi}

\begin{lem}[Abel's Lemma, Lemma 1.15 of \cite{range1998holomorphic}] 
    Given $a_\nu \in \mb K$ for all $\nu\in \mb N_0^d, \boldsymbol{c}\in\mb K^d$, and the power series \eqref{eq: power series} above, there exists $\boldsymbol{r}\in \mb R_+^d$ such that
    $$
    \sum_{\nu \in \mathbb{N}_0^d}\left|a_\nu\right| \boldsymbol{r}^\nu<\infty.
    $$
    Hence, the power series $\sum a_{\nu} (\z-\boldsymbol{c})^{\nu}$ converges on the polydisc $D(\boldsymbol{c},\boldsymbol{r})$.
\end{lem}

\begin{defi}[Analytic functions]
    Let $f: \Omega \rightarrow \mb K$ be a function defined on an open set $\Omega\subset \mb K^d$. We say that $f$ is analytic on $\Omega$ if for any point $\boldsymbol{p}\in \Omega$, it can be written as a power series
    $$
    f(\z)=\sum_{\nu\in\mb N_0^d} a_{\nu} (\z-\boldsymbol{p})^{\nu}
    $$
    in which the coefficients $a_{\nu}\in\mb K$, and the series converges to $f(\z)$ for $\z \in D(\boldsymbol{p},\boldsymbol{r})$ for some $\boldsymbol{r}\in \mb R_+^d$. Indeed, the coefficients $a_{\nu}=\frac{\partial^{\nu}f(\boldsymbol{p})}{\nu!}$. In other words, $f$ is said to be analytic at point $\boldsymbol{p}$ if its Taylor series at $\boldsymbol{p}$ converges to $f$ in some polydisc $D(\boldsymbol{p},\boldsymbol{r})$.
\end{defi}

\begin{lem}[Osgood's Lemma, Theorem 2 of Chapter 1 of \cite{gunning2022analytic}] 
    Let $\Omega \subset \mb K^d$ be an open set, and function $f:\Omega\rightarrow \mb K$. If $f$ is continuous and analytic in each variable separately, then $f$ itself is analytic.
\end{lem}

\subsection{Spectral method} \label{section: spectral method}

The spectral method is powerful in computing the first, second, and higher-order derivatives of analytic functions, whether real or complex-valued \cite{fornberg1981numerical,trefethen2000spectral,lyness1971algorithm,lyness1968differentiation,lyness1967numerical}. It offers better error performance compared to the finite difference method.

\begin{lem}[Corollary 1.1.10 of \cite{krantz2002primer}]
\label{lem: constant in power series}
The power series $\sum_n a_n (x-\alpha)^n$ has radius of convergence $\rho$ if and only if for each $0<r<\rho$, there exists a constant $0<\kappa=\kappa(r)$ such that $|a_n|\leq \kappa r^{-n}$ for all $n$. 
\end{lem}
\begin{proof}
    Suppose that $0<\varlimsup\limits_{n \rightarrow \infty} \left|a_n\right|^{1 / n}=A<\infty$, then $\rho=A^{-1}$. By definition, 
    $$
    \lim _{n \rightarrow \infty}(\sup _{m \geq n} |a_m|^{1/m})=\rho^{-1}.
    $$
    For any $0<r<\rho$, let $\varepsilon=r^{-1}-\rho^{-1} >0$, there exists $N=N(r)$ such that whenever $n>N$,
    $$
    \sup _{m \geq n} |a_m|^{1/m}\leq\rho^{-1}+\varepsilon=r^{-1},
    $$
    which means $|a_m|^{1/m}\leq r^{-1}$ for all $m>N$. Let $\kappa=\max\{1, |a_1| r, \ldots, |a_N| r^N \}$, then for all $n\in \mb N$, we have $|a_n|\leq \kappa r^{-n}.$   
\end{proof}

We can make the constant $\kappa$ more precisely for analytic functions by the Cauchy inequality.

\begin{lem}[Cauchy's inequality, Chapter 2.5 of \cite{titchmarsh1939theory}]
\label{lem: bound of derivatives}
Assume that $f$ is analytic at $x_0$ with the radius of convergence $r$. 
Let $\delta<r$ and $M_\delta$ be the upper bound of $|f(x)|$ on the circle $C_\delta=\{x\in \mathbb{C}:|x-x_0|=\delta\}$, then
$$
|f^{(n)}(x_0)| \leq \frac{n! M_\delta}{\delta^n}
$$
for all $n\in \mb N$.
\end{lem}

\begin{proof}
    By Cauchy's differentiation formula, we have
    \be
    \label{Cauchy integral}
    f^{(n)}(x_0)=\frac{n !}{2 \pi i} \oint_{C_\delta} \frac{f(x)}{(x-x_0)^{n+1}} d x.
    \ee
    Hence, 
    $$
    |f^{(n)}(x_0)| \leq \frac{n !}{2 \pi} \cdot \frac{M_\delta}{\delta^{n+1}} \cdot 2\pi \delta =\frac{n! M_\delta}{\delta^n}.
    $$
\end{proof}

In the above proof, if we approximate the integral (\ref{Cauchy integral}) by an $N$-point trapezoidal rule, we can then approximate $\delta^n f^{(n)}(x_0)/n!$ with $N^{-1} \sum_{k=0}^{N-1} e^{2\pi i k n/N} f(x_0 + \delta e^{-2\pi i k /N})$. 
Another way to see this is from the Taylor series of $f(x)$ at $x_0$, where the error analysis is relatively more clear.

Now let $f$ be a function that is analytic at point $x_0\in\mb C$, and the Taylor series of $f$ at $x_0$ converges in an open neighbourhood of $x_0\in \mb C$ that contains $\overbar{B}(x_0,r) \subset \mb C$, which is the closed disk with center $x_0$ and radius $r \in \mb R_+$. In $\overbar{B}(x_0,r)$, $f$ can be computed by the Taylor series:
$$
f(x)=\sum_{n=0}^{ \infty} a_n(x-x_0)^n, \quad a_n=\frac{f^{(n)}(x_0)}{n !}.
$$
According to Lemma \ref{lem: bound of derivatives}, there is a constant $\kappa=M_r>0$ such that 
$$
|a_n| \leq \kappa r^{-n} , \quad  \forall n \in \mathbb{N}.
$$
Here $M_r$ denotes the maximum value of $|f(x)|$ on the circle centered at $x_0$ with radius $r$. 

Choose some $\delta \in (0,r), N\in \mb N$, consider the following $N$ values of $f$
$$
f_k:=f(x_0+\delta \omega^k), \quad k=0, \ldots, N-1,
$$
where $\omega=e^{-\frac{2 \pi i}{N}}$. Then we have
$$
\begin{aligned}
f_k =\sum_{n=0}^{\infty} a_n(\delta \omega^k)^n 
 =\sum_{n=0}^{N-1} \omega^{k n}\left(\sum_{m=0}^{\infty} a_{n+m N} \delta^{n+m N}\right).
\end{aligned}
$$
Denote $c_n=\sum_{m=0}^{\infty} a_{n+m N} \delta^{n+m N}$, then 
$$
f_k=\sum_{n=0}^{N-1} \omega^{k n} c_n,
$$
which is an application of the discrete Fourier transform (DFT). The inverse DFT gives
$$
c_n=\frac{1}{N} \sum_{k=0}^{N-1} \omega^{-k n} f_k = \frac{1}{N} \sum_{k=0}^{N-1} \omega^{-k n} f(x_0+\delta e^{-2\pi i k/N}).
$$
This formula coincides with the formula obtained by the trapezoidal rule mentioned above.
Thus, $c_n$ can be obtained through $N$ function evaluations of $f$. Consequently, an approximation of $a_n=\frac{f^{(n)}(x_0)}{n !}$ can be provided by $c_n/\delta^n$, and the error is
\be
\left|a_n-\frac{c_n}{\delta^n}\right| \leq \kappa r^{-n} \sum_{m=1}^{\infty}(\delta / r)^{m N} 
=\kappa r^{-n} \frac{(\delta / r)^N}{1-(\delta / r)^N}.
\label{eq: error of spectral method}
\ee
Particularly, when $n=1$, 
\be
\left|f'(x_0)-\frac{c_1}{\delta}\right| \leq 
\kappa r^{-1} \frac{(\delta / r)^N}{1-(\delta / r)^N}.
\label{eq: error of spectral method for gradient}
\ee

There are two ways to look at the approximation (\ref{eq: error of spectral method}). If we fix $N$ and consider the behavior of the error when the radius $\delta$ approaches 0, then we derive
$a_n = c_n/\delta^n + O(\kappa \delta^N)$. This is an approximation of $a_n$ of order $N$ with respect to $\delta$. We can also fix $\delta$ and let $N\rightarrow \infty$, we now obtain $a_n = c_n/\delta^n + O(\kappa e^{-\alpha N})$, where $\alpha = \log(r/\delta)> 0$. This approximation is exponential with respect to $N$. In this work, we will use the second way to determine $N$ to achieve the desired accuracy $\varepsilon$ by setting $\alpha$ as a constant. As a result, $N = O(\log(\kappa/\varepsilon))$.

\subsection{Grid of sampling points}
Every algorithm that estimates some information based on function evaluations must describe a set of sampling points. For clarity and ease of comparison between previous methods and ours, we describe a common grid of points here.

\begin{defi}[Grid of sampling points]
\label{def: grid Gn}
    Let $d, n \in \mb N$, we define
    \be
    \label{eq for gird}
    G_n^d:= \bigtimes_{i=1}^d \left\{\frac{1}{2^n} \cdot k_i+\frac{1}{2^{n+1}}: \; k_i\in \{-2^{n-1},\dots, 2^{n-1}-1\} \right\} \subset \mb R^d.
    \ee
    Note that $G_n^d$ is the grid around the original point with side length $1$, and for $a \in \mb R_+$, we refer to $aG_n^d$ as the grid with side length $a$. It is clear that the elements $\boldsymbol{x}=\left(x_1,\dots,x_d\right) \in G_n^d$ are just vectors whose components $x_i$ are uniformly distributed in $(-1/2,1/2)$ with intervals of $1/2^n$.
\end{defi}

From the analysis presented in Lemma \ref{lem: GAW's gradient estimation thm}, $d$ denotes the dimension (or the number of variables) of $f$, and $n$ relates to the accuracy of estimating the gradient of $f$.

Note that every binary string $b=(b_0, \ldots, b_{n-1})\in \{0,1\}^n$ can be viewed as an integer in $[0, 2^{n})\cap \mb Z$, i.e., $b_0+2b_1+\cdots+2^{n-1} b_{n-1}$. There is a natural bijection between $[0, 2^{n})\cap \mb Z$ and $[-2^{n-1}, 2^{n-1})\cap \mb Z$, which leads to 
    a natural bijection between $[0, 2^{n})\cap \mb Z$ and the set $G_n$ defined above for $d=1$. Hence, an $n$-qubit basis state $|b_0\ket\cdots |b_{n-1}\ket$ can represent a corresponding element in $G_n$. In this work, we will label $n$-qubit basis states by elements in $G_n$ for convenience.

We define the quantum Fourier transform acting on state $|x\ket$ for $x \in G_n$ as 
\be
\label{QFT for Gn}
\mathrm{QFT}_{G_n}:|x\rangle \rightarrow \frac{1}{\sqrt{2^n}} \sum_{k \in G_n} e^{2 \pi i 2^n x k}|k\rangle.
\ee
This unitary is equivalent to the usual quantum Fourier transform $\mathrm{QFT}_{2^n}$ up to conjugation with a tensor product of $n$ single-qubit phase gates, so we can interchangeably use these two equivalent transforms, see \cite[Claim 5.1]{gilyen2019optimizing}.

\section{Quantum spectral method for gradient estimation}
\label{section: Quantum spectral method for gradient estimation}

Most previous work, e.g., \cite{jordan2005fast, gilyen2019optimizing}, about computing gradients through quantum algorithms focused on real-valued functions. A key technique is the quantum Fourier transform.
However, complex numbers play a crucial role in quantum mechanics \cite{karam2020complex}, where observables are Hermitian operators, the quantum Fourier transform is frequently applied, and the time evolution of states is described by unitaries, among others. Complex-valued functions also have plenty of applications in mathematics and engineering.

For a complex-valued function $f:\mb C \rightarrow \mb C$, to utilize the spectral method, we explore functions satisfying $f'(\mb R) \subseteq \mb R$. 
Note that if $f:\mb C \rightarrow \mb C$ is analytic and maps $\mb R$ into $\mb R$, then for all $n \in \mb N$, we have $f^{(n)}(\mb R) \subseteq \mb R.$
This can be easily deduced from the uniqueness property of the limit, and multivariable functions exhibit similar behavior by definition. This constraint on functions may seem somewhat contrived at first glance, but it has numerous applications and examples, such as polynomials, exponentials, and other basic functions. In fact, any analytic real-valued function can be locally extended to the complex plane, which is quite natural through Taylor expansion. Moreover, the derivatives we consider are also local properties that align well with this perspective.




Given oracle access to a function that can evaluate its values, if we aim to approximate its derivatives at a point, it is necessary to provide a formula expressing the derivatives in terms of function values. In  \cite{gilyen2019optimizing}, they used a finite difference formula for approximating the first-order derivative, derived from Taylor expansion, for this purpose. In this work, we employ the spectral method for analytic functions, as introduced above, which arises from the Cauchy differentiation formula.

\begin{thm} \label{thm: spectral formula for gradient}
    Let $N\in \mb N$ and $\omega=e^{-\frac{2\pi i}{N}}$. 
    Let $f:\mb C^d \rightarrow \mb C$ be an analytic function that maps $\mb R^d$ into $\mb R$. Assume the Taylor series of $f$ converges to itself in a closed polydisc $\overbar{D}(\0,\r)$ with $\r\in \mb R_+^d$. Denote $r=\min_{j\in [d]} r_j$, then for all $\boldsymbol{x}\in G_n^d$ and $\delta \in (0,2r)$,
    $$
    F(\boldsymbol{x}):=\frac{1}{N \delta}\sum\limits_{k=0}^{N-1} \omega^{-k} f(\delta \omega^k \boldsymbol{x})
    $$
    gives an approximation for $\nabla f(\boldsymbol{0}) \cdot \boldsymbol{x}$ with error
    \bes
    \left|F(\boldsymbol{x})- \nabla f(\boldsymbol{0}) \cdot \boldsymbol{x}\right| \leq \frac{\kappa}{2r} \frac{(\delta /2r)^N}{1-(\delta /2r)^N}
    \ees
    for some constant $\kappa=\max_{\tau} |f(\tau\x)|>0$ with $\tau$ ranging over the circle centered at $0$ with radius $2r$. 
    Moreover, $F(\boldsymbol{x})\in \mb R$.
\end{thm}

\begin{proof}
    Given $\x \in G_n^d \subset \mb R^d$, define $h(\tau):=f(\tau \boldsymbol{x})$ as a function over the field $\mathbb{C}$ which is also analytic and maps $\mb R$ to $\mb R$, then
    $$
    h^{(k)}(\tau)=\frac{d^k f(\tau \x)}{d \tau^k}=\partial_{\x}^k f(\tau \x).
    $$
    In particular, $h^{\prime}(0)=\nabla f(\0) \cdot \x$.
    Since the Taylor series of $f$ converges to itself in $\overbar{D}(\0,\r)$, the Taylor series of $h$ converges as well for all $\tau$ satisfying $|\tau x_j|\leq r_j, j\in [d]$. 
    More precisely, 
    take $\tilde{r}:=2r$, then for all $\x\in G_n^d$, the Taylor series of $h$ converges to itself in an open set containing the closed ball $\overbar{B}(0,\tilde{r})\subset \mb C$. By (\ref{eq: error of spectral method for gradient}), 
    \[
    \left|h'(0) - \frac{c_1}{\delta} \right| \leq \kappa \tilde{r}^{-1} \frac{(\delta/\tilde{r})^N}{1-(\delta/\tilde{r})^N},
    \]
    where
    \[
    c_1 = \frac{1}{N} \sum_{k=0}^{N-1} \omega^{-k} h(\delta \omega^k)
    = \frac{1}{N} \sum_{k=0}^{N-1} \omega^{-k} f(\delta \omega^k \x)
    \]
    and $\kappa$ is the maximum value of $|h(\tau)|$ on the circle centered at $0$ with radius $\tilde{r}$.
    Finally, we obtain the claimed error bound in this theorem by noting that $F(\x)=c_1/\delta$. Moreover, $h^{(n)}(\mb R)\subseteq \mb R$ for all $n\in \mb N$ since $h(\mb R)\subseteq \mb R$. Consequently, $F(\boldsymbol{x})\in \mb R$ since $c_1 \in \mb R$ by its definition.
\end{proof}

\begin{lem}[Theorem 5.1 of \cite{gilyen2019optimizing}] 
\label{lem: GAW's gradient estimation thm}
    Let $c \in \mb R, a, \rho, \varepsilon<M \in \mathbb{R}_{+}$, and $\y, \boldsymbol{g} \in$ $\mathbb{R}^d$ such that $\|\boldsymbol{g} \|_{\infty} \leq M$. Let $n_{\varepsilon}=\left\lceil\log _2(4 /a \varepsilon)\right\rceil$, $n_M=\left\lceil\log _2(3 a M)\right\rceil$ and $n=n_{\varepsilon}+n_M$. Suppose $f:\left(\boldsymbol{y}+a G_n^d\right) \rightarrow \mathbb{R}$ is such that
    $$
        |f(\boldsymbol{y}+a \boldsymbol{x})-\boldsymbol{g} \cdot a \boldsymbol{x} - c| \leq \frac{\varepsilon a}{8 \cdot 42 \pi}
    $$
    for all but a $1/1000$ fraction of the points $\boldsymbol{x} \in G_n^d$. If we have access to a phase oracle $$\mathrm{O}:|\boldsymbol{x}\rangle \rightarrow e^{2 \pi i 2^{n_{\varepsilon}} f(\boldsymbol{y}+a \boldsymbol{x})}|\boldsymbol{x}\rangle$$ acting on $\mathcal{H}=\operatorname{Span}\left\{|\boldsymbol{x}\rangle: \boldsymbol{x} \in G_n^d\right\}$, then we can calculate a vector $\tilde{\boldsymbol{g}} \in \mathbb{R}^d$ such that
    $$
    \operatorname{Pr}\left[\|\tilde{\boldsymbol{g}}-\boldsymbol{g}\|_{\infty}>\varepsilon\right] \leq \rho,
    $$
    with $O\left(\log \big(\frac{d}{\rho}\big)\right)$ queries to $\mathrm{O}$.
\end{lem}

In the above lemma, the gate complexity is $O\left(d \log\big(\frac{d}{\rho}\big) \log\big(\frac{M}{\varepsilon}\big)\log \log\big(\frac{d}{\rho}\big) \log\log\big(\frac{M}{\varepsilon}\big)\right)$.
Here we briefly summarize this quantum algorithm, which is also a key subroutine of our quantum algorithm below, for computing $\tilde{\g}$. 
We set $a=1, \y=0$ for convenience.
Denote $N=2^n$.
With $\mathrm{O}$, we can generate the following state
\[
|\phi\rangle:=\frac{1}{\sqrt{N^d}} \sum_{\boldsymbol{x} \in G_n^d} e^{2\pi i  f(\boldsymbol{x})N}|\boldsymbol{x}\rangle.
\]
By the condition that $f(\x)$ is close to $\g \cdot \x + c$, the above state is close to
\beas
|\psi\rangle &:=& \frac{1}{\sqrt{N^d}} \sum_{\boldsymbol{x} \in G_n^d} e^{2\pi i (\boldsymbol{g} \cdot \boldsymbol{x}+c)N}|\boldsymbol{x}\rangle \\
&=& e^{2\pi i cN} \left(\frac{1}{\sqrt{N}} \sum_{x_1 \in G_n} e^{2\pi ig_1 \cdot x_1N}\left|x_1\right\rangle\right) \otimes \cdots \otimes\left(\frac{1}{\sqrt{N}} \sum_{x_d \in G_n} e^{2\pi i g_d \cdot x_dN}\left|x_d\right\rangle\right).
\eeas
Finally, we apply the inverse of the quantum Fourier transform on $G_n^d$ to $\Ket{\phi}$, which will be a state close to $\Ket{\tilde{\g}}$ with high probability. In the above, because of the representation of grid points in $G_n$, it is $e^{2\pi ig_1 \cdot x_1N}$ rather than $e^{2\pi ig_1 \cdot x_1/N}$ in the standard quantum Fourier transform.

In Theorem \ref{thm: spectral formula for gradient}, we already proved that $F(\x)$ is close to $\nabla f(\0)\cdot \x$. To use the above result to estimate $\nabla f(\0)$, we need to construct the phase oracle O to query $F(\x)$. We below consider this construction in two cases: using the binary oracle $U_f$ and using phase oracles $O_{f_1}, O_{f_2}$.

Below we will set $a=1$. Different from the finite difference method used in \cite{gilyen2019optimizing}, $a$ is not important for us in the spectral method. This will make a big difference in gradient and Hessian estimation.


\begin{lem}[Oracle construction using binary oracle]
\label{lem: binary oracle conversion in spectral method}
    Given an analytic function $f:\mb C^d \rightarrow \mb C$ that maps $\mb R^d$ into $\mb R$, and the access to the oracle $U_f^\eta$ defined in Definition \ref{def: complex binary oracle}, then for $F(\x)$ defined in Theorem \ref{thm: spectral formula for gradient} and $n_\varepsilon=\left\lceil\log _2(4 /\varepsilon)\right\rceil$, we can get an $\eta'$-precision oracle 
    $$
    \mathrm{O}:|\boldsymbol{x}\rangle \rightarrow e^{2 \pi i 2^{n_{\varepsilon}} F(\boldsymbol{x})}|\boldsymbol{x}\rangle
    $$
    required in Lemma \ref{lem: GAW's gradient estimation thm}, with $O(N)$ queries to $U_f^\eta$ and $\eta/\eta' \in O(\varepsilon\delta)$.
\end{lem}
\begin{proof}
    By Theorem \ref{thm: spectral formula for gradient}, we have $F(\boldsymbol{x})=\frac{1}{N \delta}\sum_{k=0}^{N-1} \omega^{-k} f(\delta \omega^k \boldsymbol{x}) \in \mb R$ for any $\boldsymbol{x}\in \mb R^d$, particularly for $\x \in G_n^d$. Using Lemma \ref{lem: complex arithmetic}, we can implement unitaries for any $k \in [N]$ satisfying
    $$
    V_{k}: |\boldsymbol{x}\ket |\boldsymbol{0}\ket \rightarrow |\boldsymbol{x}\ket |\delta \omega^k \boldsymbol{x}\ket,
    $$
    and
    $$
    W_k: |f(\delta \omega^k \boldsymbol{x})\ket |\boldsymbol{0}\ket \rightarrow |f(\delta \omega^k \boldsymbol{x})\ket |\omega^{-k}f(\delta \omega^k \boldsymbol{x})\ket.
    $$
    
    Using the above unitaries and oracle $U_f$ (and its inverse $U_f^{-1}$), we can construct an oracle for any $k_1, k_2\in [N]$ (we omit the tensor product with the identity mapping for brevity):
    $$
    \begin{aligned}
        |\boldsymbol{x}\ket |\boldsymbol{0}\ket|\boldsymbol{0}\ket|\boldsymbol{0}\ket &\xrightarrow{V_{k_1}} &&|\boldsymbol{x}\ket |\delta \omega^{k_1}\boldsymbol{x} \ket |\boldsymbol{0}\ket|\boldsymbol{0}\ket\\ 
        &\xrightarrow{U_f} &&|\boldsymbol{x}\ket |\delta \omega^{k_1}\boldsymbol{x} \ket |f(\delta \omega^{k_1}\boldsymbol{x})\ket |\boldsymbol{0}\ket \\
        &\xrightarrow{W_{k_1}} &&|\boldsymbol{x}\ket |\delta \omega^{k_1}\boldsymbol{x} \ket |f(\delta \omega^{k_1}\boldsymbol{x})\ket |\omega^{-k_1}f(\delta \omega^{k_1}\boldsymbol{x})\ket.
    \end{aligned}
    $$
    Note that the two intermediate registers serve as auxiliary components, with the first one representing the input $\boldsymbol{x}$ and the last one storing the desired outcome. 
    
    As before, we denote $f_k :=f(\delta\omega^k\boldsymbol{x})$. Then the inverse $V_{k_1}^{-1} \cdot U_f^{-1}$ allows us to reset the auxiliary registers to the original states, enabling the proceeding operations:\footnote{The first arrow is similar to the above using $V_{k_2}$ and $U_f$, and the second one follows the complex-value operations describing in Lemma \ref{lem: complex arithmetic}.}
    $$   
    \begin{aligned}
        |\boldsymbol{x}\ket |\boldsymbol{0}\ket |\boldsymbol{0}\ket |\omega^{-k_1}f_{k_1}\ket &\longrightarrow &&|\boldsymbol{x}\ket |\delta \omega^{k_2}\boldsymbol{x}\ket |f_{k_2}\ket |\omega^{-k_1}f_{k_1}\ket \\
        &\longrightarrow &&|\boldsymbol{x}\ket |\delta \omega^{k_2}\boldsymbol{x}\ket |f_{k_2}\ket |\omega^{-k_1}f_{k_1}+\omega^{-k_2}f_{k_2} \ket\\
        &\longrightarrow &&|\boldsymbol{x}\ket |\boldsymbol{0}\ket |\boldsymbol{0}\ket |\omega^{-k_1}f_{k_1}+\omega^{-k_2}f_{k_2}\ket.
    \end{aligned}
    $$
     
    Running the oracle described above from $k=0$ to $N-1$, we obtain
    $$
    \begin{aligned}
        |\boldsymbol{x}\ket |\boldsymbol{0}\ket &\longrightarrow |\boldsymbol{x}\ket |\sum\limits_{k=0}^{N-1} \omega^{-k} f(\delta \omega^k \boldsymbol{x})\ket \\
        &\longrightarrow |\boldsymbol{x}\ket |\frac{1}{N\delta}\sum\limits_{k=0}^{N-1} \omega^{-k} f(\delta \omega^k \boldsymbol{x})\ket \\
        &\longrightarrow |\boldsymbol{x}\ket |F(\boldsymbol{x})\ket \\
        &\longrightarrow |\boldsymbol{x}\ket |2\pi \cdot 2^{n_\varepsilon}F(\boldsymbol{x})\ket.
    \end{aligned}
    $$
    Since $2\pi \cdot 2^{n_\varepsilon} F(\boldsymbol{x}) \in \mb R$, standard oracle conversion gives us the oracle
    $$
    \mathrm{O}: |\boldsymbol{x}\ket \rightarrow e^{2\pi i 2^{n_\varepsilon}F(\boldsymbol{x})}|\boldsymbol{x}\ket.
    $$
    The overall process shows that the oracle $U_f$ and its inverse are called a total number of $4N$ times.
\end{proof}

\begin{rmk}
    Note that the binary oracle $U_f^\eta$ gives us an $\eta$-precision approximation $\tilde{f}(\x)$ instead of the exact $f(\x)$. Therefore, the oracle $\mathrm{O}$, which we constructed using $U_f^\eta$ in the above lemma, also exhibits a precision $\eta'$. It is easy to see that $\eta' \approx \eta/\varepsilon\delta$, implying that the cost of each application of O is $C(\varepsilon\delta \eta')$.
\end{rmk}

\begin{lem}[Oracle construction using phase oracle]
\label{lem: phase oracle conversion in spectral method}
    Let $f:\mb C^d \rightarrow \mb C$ be an analytic function that maps $\mb R^d$ into $\mb R$, with real and imaginary parts $f_1, f_2$. Suppose we have access to phase oracles $O_{f_1}^\eta, O_{f_2}^\eta$. Then for $F(\x)$ defined in Theorem \ref{thm: spectral formula for gradient} and $n_\varepsilon=\left\lceil\log _2(4 /\varepsilon)\right\rceil$, we can get the oracle up to accuracy $\eta$
    $$
    \mathrm{O}:|\boldsymbol{x}\rangle \rightarrow e^{2 \pi i 2^{n_{\varepsilon}} F(\boldsymbol{x})}|\boldsymbol{x}\rangle
    $$
    required in Lemma \ref{lem: GAW's gradient estimation thm}, with $O\left( \frac{1}{\varepsilon\delta} + N\log\big(\frac{N}{\eta}\big)\right)$ queries to each $O_{f_1}^\eta, O_{f_2}^\eta$.
\end{lem}
\begin{proof}
    By Theorem \ref{thm: spectral formula for gradient}, we have $F(\boldsymbol{x}):=\frac{1}{N \delta}\sum_{k=0}^{N-1} \omega^{-k} f(\delta \omega^k \boldsymbol{x}) \in \mb R$ for any $\boldsymbol{x}\in \mb R^d$. We write $c_k:=\cos (-2\pi k/N)$ and $s_k:=\sin (-2\pi k/N)$ for all $k\in [N]$, then $\omega^{-k}=c_k+i s_k$. As $F(\x)$ is real, we obtain that
    $$
    F(\x)=\frac{1}{N\delta} \sum_{k=0}^{N-1} \left(c_k f_1(\x_k)-s_k f_2(\x_k) \right)
    $$
    where $\x_k=\delta\omega^k \x$. Observe that
    $$
    e^{2\pi i 2^{n_\varepsilon}F(\x)}|\x\ket=e^{i \frac{2\pi 2^{n_\varepsilon}}{N\delta} \sum_{k=0}^{N-1} (c_k f_1(\x_k) - s_k f_2(\x_k) )} |\x\ket.
    $$
    Denote the scaling factor $S:=\frac{2\pi 2^{n_\varepsilon}}{N\delta}$, then implementing $\mathrm{O}$ requires
    $$
    \sum_{k=0}^{N-1} \left\lceil |c_k| S\right\rceil \leq NS= \frac{\pi}{2 \varepsilon\delta},
    $$
    queries to $O_{f_1}$ and at most $N$ queries to fractional oracle $O_{t f_1}$ defined in Definition \ref{def: fractional oracle}. 
    
    Given $O_{f_1}$, by Lemma \ref{lemma: phase to fractional} we can implement $O_{t f_1}$ up to error $\eta'$ using $O(\log(1/\eta'))$ applications of $O_{f_1}$. 
    Thus, the overall error in the above procedure is $O(N\eta')$ as there are $N$ terms in the summation.
    Setting $\eta'=\eta/N$ leads to the claimed result. The same argument holds for $O_{f_2}$ as well.
\end{proof}

With the above preliminaries, we are now prepared to present our main theorem for gradient estimation.

\begin{thm}[Quantum spectral method for gradient estimation]
\label{thm: gradient estimation using spectral method}
    Let $\rho, \varepsilon \in \mb R_+$, and let $f:\mb C^d \rightarrow \mb C$ be an analytic function that maps $\mb R^d$ into $\mb R$. 
    Assume that the Taylor series of $f$ converges to itself in a closed polydisc $\overbar{D}(\0,\r)$ for some $\r\in \mb R_+^d$. Denote $r=\min_{j\in [d]} r_j$ and $\kappa = \max_{|\x|\in \overbar{D}(\0,\r)}|f(\x)|$.
    Given access to the oracle $U_f^\eta$ as in Definition \ref{def: complex binary oracle}, or phase oracles $O_{f_1}^\eta, O_{f_2}^\eta$ as in Definition \ref{def: phase oracle}, there exists a quantum algorithm that computes an approximate gradient $\boldsymbol{g}$ such that $\|\boldsymbol{g}-\nabla f(\boldsymbol{0})\|_\infty \leq \varepsilon$ with probability at least $1-\rho$. The algorithm requires:
    \begin{itemize}
        \item $O\left(\log \big(\frac{d}{\rho}\big) \log \big(\frac{\kappa}{\varepsilon r}\big) \right)$ queries to \( U_f \) with precision $\eta \in O(\varepsilon r)$, or
        \item $O\left(\frac{1}{\varepsilon r}\log\big(\frac{d}{\rho}\big) +\log\big(\frac{d}{\rho}\big)\log\big(\frac{\kappa}{\varepsilon r}\big)\log\big(\frac{1}{\eta} \log\big(\frac{\kappa}{\varepsilon r}\big) \big) \right)$ queries to $O_{f_1}$ and $O_{f_2}$.
    \end{itemize}
\end{thm}

\begin{proof}
    By Theorem \ref{thm: spectral formula for gradient}, we have
    $$
    |F(\boldsymbol{x}) - \nabla f(\boldsymbol{0})\cdot \boldsymbol{x} | \leq \frac{\kappa }{2r} \frac{(\delta /2r)^N}{1-(\delta /2r)^N} =E 
    $$
    for all $\x \in G_n^d$.
    We now set $\delta=r$. By choosing
    $$
    N \geq \frac{\log(1+4\cdot42 \pi \kappa/\varepsilon r)}{\log(2r/\delta)} = O\left(\log(\kappa/\varepsilon r)\right),
    $$
    we can ensure that
    \[
    E \leq \frac{\varepsilon}{8 \cdot 42 \pi},
    \]
    which is required in Lemma \ref{lem: GAW's gradient estimation thm} with $a=1$. 
    By Lemma \ref{lem: binary oracle conversion in spectral method}, it costs $O(N)$ to construct the oracle O to query $F(\x)$ from $U_f^\eta$, and by Lemma \ref{lem: phase oracle conversion in spectral method}, the cost is $O\left( \frac{1}{\varepsilon\delta} + N\log\big(\frac{N}{\eta}\big) \right)$ from $O_{f_1}^\eta, O_{f_2}^\eta$. 
    The oracle $\mathrm{O}$ obtained using $U_f^\eta$ has precision $\eta'=8 \pi \eta/\varepsilon\delta$. To ensure the algorithm operates correctly, $\eta'$ should be approximately $1/42$, which implies $\eta = O(\varepsilon\delta)$.
    The claims in this theorem are now from Lemma \ref{lem: GAW's gradient estimation thm}.
\end{proof}

In the above theorem, the first claim seems better than the second one. Indeed, in Jordan's algorithm \cite{jordan2005fast} and GAW's algorithm \cite{gilyen2019optimizing}, if we also use the binary oracle, then these algorithms also achieve exponential speedups in terms of $d$. However, as already analyzed in \cite{gilyen2019optimizing} and can also be seen from the first claim, it requires a high level of accuracy to implement $U_f^\eta$ with $\eta=O(\varepsilon r)$. The cost of each application of $U_f^\eta$ is $C(\eta)=C(\varepsilon r)$. If this oracle is obtained from phase estimation, then the cost $C(\varepsilon r) = 1/\varepsilon r$, which coincides with our second claim. 
In some other cases, as explained in \cite{gilyen2019optimizing}, the dependence on the precision can be as bad as $1/d$.
From the viewpoint of applications, the phase oracle is more commonly used and easier to obtain, especially from the probability oracle.



\section{Quantum algorithms for Hessian estimation using spectral method}
\label{section: Quantum algorithms for Hessian estimation using spectral method}

In this section, we extend our idea to estimate Hessians of functions. Our basic idea is to estimate each column of the Hessian by reducing the task to a gradient estimation problem. This technique is commonly used in quantum algorithms \cite{montanaro_shao-tqc,lee2021quantum}. In the simplest case, given an oracle to query $f(\x)=\Bra{\x} H \Ket{\x}$ for some unknown matrix $H$. For any fixed $\y$, by considering $f(\x+\y)-f(\x)$, we will obtain an oracle to compute $g(\x) := \Bra{\x} H \Ket{\y}$.
This function is linear, hence the vector $H \Ket{\y}$ can be learned using the quantum algorithm for gradient estimation. Below, we extend this idea to general functions.

\subsection{Estimating Hessians in the general case}

The following result generalizes Lemma \ref{lem: GAW's gradient estimation thm} for gradient estimation to Hessian estimation.

\begin{prop}[From gradient to Hessian] 
\label{thm: vector to matrix}
    Let $c \in \mb R, \rho, \varepsilon< M \in \mb R_+$ as in Lemma \ref{lem: GAW's gradient estimation thm}. Given a matrix $H \in \mb R^{d \times d}$ such that $\|H\|_{\rm{max}} \leq M$, if we have $h: \mb R^d \rightarrow \mb R$ such that
    $$
    |h(\boldsymbol{z})-\bra \boldsymbol{z}|H|\boldsymbol{z}\ket - c|\leq \frac{ \varepsilon}{8 \cdot 42 \pi}
    $$
    for almost all inputs $\z=\boldsymbol{x}+\y$, where $\x\in G_n^d$ and $\y\in \{0,1\}^d$, and access to its phase oracle $O_h: |\boldsymbol{z}\ket \rightarrow e^{ih(\boldsymbol{z})} |\boldsymbol{z}\ket$. Then we can calculate a matrix $\widetilde{H}$ with probability at least $1-\rho$ such that
    $$
    \|\widetilde{H}-H\|_{\rm{max}} \leq \varepsilon 
    $$
    using $O\left(\frac{d}{\varepsilon} \log\big(\frac{d}{\rho}\big)\right)$ queries to the phase oracle $O_h$.
\end{prop}

\begin{proof}
    For any fixed $\boldsymbol{y}\in \{0,1\}^d$, let $F_{\y}(\boldsymbol{x}):= \frac{1}{2} \left(h(\boldsymbol{x}+\boldsymbol{y})-h(\boldsymbol{x}) \right)$, then using two queries to oracle $O_h$, we can obtain the following oracle
    $$
    O_{F_{\y}}: |\boldsymbol{x}\ket \rightarrow e^{i F_{\y}(\boldsymbol{x})} |\boldsymbol{x}\ket.
    $$
    It is easy to check that 
    \[
    \left| F_{\y}(\boldsymbol{x}) - \Bra{\x}H\Ket{\y} - \frac{1}{2}\Bra{\y}H\Ket{\y} \right|\leq \frac{\varepsilon}{8\cdot 42 \pi}.
    \]
    This satisfies the condition described in Lemma \ref{lem: GAW's gradient estimation thm} with $a=1$. From $O_{F_{\y}}$, we can implement the oracle described in Lemma \ref{lem: GAW's gradient estimation thm} with $O(1/\varepsilon)$ applications of $O_{F_{\y}}$. Now, by Lemma \ref{lem: GAW's gradient estimation thm}, there is a quantum algorithm that approximates $H\Ket{\y}$. In particular, if we choose $\y = \Ket{i}$, i.e., the $i$-th column of the identity matrix, then there is a quantum algorithm that returns $\g_i$ satisfying
    \[
    \operatorname{Pr}\Big[\|\g_i - H\Ket{i}\|_{\infty} \leq \varepsilon \Big] \geq 1 - \frac{\rho'}{d},
    \]
    with query complexity $O(\log(d^2/\rho'))$. Here $H\Ket{i}$ is exactly the $i$-th column of $H$. Let $\widetilde{H}$ be the matrix generated by all $\g_i$, then\footnote{Here we used the facts that $\ln(x)\geq 1-x^{-1}$ and $e^{-x}\geq 1-x$.}
    $$
    \on{Pr}\left[\|\widetilde{H}-H\|_{\rm{max}} \leq \varepsilon  \right] \geq \left(1-\frac{\rho'}{d}\right)^d \geq 1-\frac{d\rho'}{d-\rho'}.
    $$
    Setting $\rho'=d\rho/(d+\rho)$, then the success probability becomes $1-\rho$.
    Putting it all together, we obtain the claimed query complexity.
\end{proof}

Similar to Lemma \ref{lem: GAW's gradient estimation thm}, $M$ affects the gate complexity. The algorithm described above may seem natural; however, as we will demonstrate in Proposition \ref{prop: lower bound} below, it is indeed close to optimal for constant $\varepsilon$. Prior to that, we use the above result to estimate the Hessian of a general function using the spectral method.
    
\begin{thm}[Hessian estimation using spectral method]
\label{thm: hessian using spectral}
    Assuming the same conditions as in Theorem \ref{thm: gradient estimation using spectral method}, 
    then there is a quantum algorithm that computes an approximate matrix $\widetilde{\H}$ such that $\|\widetilde{\H}-\H_f(\boldsymbol{0})\|_{\rm{max}} \leq \varepsilon$ with probability at least $1-\rho$, where $\H_f(\boldsymbol{0})$ is the Hessian of $f$ at $\0$. The algorithm requires:
    \begin{itemize}
        \item $O\left(d \log \big(\frac{d}{\rho}\big) \log \big(\frac{\kappa}{\varepsilon r^2}\big) \right)$ queries to \( U_f \) with precision $\eta \in O(\varepsilon r)$, or
        \item $O \left(d \log\big(\frac{d}{\rho}\big) \left(\frac{1}{\varepsilon r}+\log\big(\frac{\kappa}{\varepsilon r^2}\big)\log\big(\frac{1}{\eta} \log\big(\frac{\kappa}{\varepsilon r^2}\big) \big)\right) \right)$ queries to $O_{f_1}$ and $O_{f_2}$.
    \end{itemize}
\end{thm}

\begin{proof}
    Following the procedure in the proof of Theorem \ref{thm: spectral formula for gradient} by considering $h''(0)$, we derive the following approximate formula by (\ref{eq: error of spectral method})
    \be
    \label{0627}
    \bra \boldsymbol{z}| \H_f(\boldsymbol{0}) |\boldsymbol{z}\ket =h''(0)
    \approx \frac{2}{N\delta^2}\sum_{k=0}^{N-1}\omega^{-2k}f(\delta\omega^k\boldsymbol{z})
    \ee
    for $\z=\x+\y$, where $\x\in G_n^d$ and $\y\in\{0,1\}^d$. To satisfy the corresponding convergence condition, we choose $\tilde{r}=2r/3$ here, and $\delta=\tilde{r}/2$ as before. The error is bounded by $2\kappa \tilde{r}^{-2}  \frac{(\delta/\tilde{r})^N}{1-(\delta/\tilde{r})^N}$. 
    The oracle for the right-hand side of (\ref{0627}) can be constructed similarly to Lemmas \ref{lem: binary oracle conversion in spectral method} and 
    \ref{lem: phase oracle conversion in spectral method}.
    We can choose an appropriate $N= O\left(\log \big(\frac{\kappa}{\varepsilon r^2}\big)\right)$ to satisfy the condition described in Proposition \ref{thm: vector to matrix}. Now following the procedure in Proposition \ref{thm: vector to matrix}, invoking the gradient estimation algorithm $d$ times can approximate the Hessian matrix $\H_f(\0)$. 

    Finally, we can obtain the claimed overall query complexity of approximating $\H_f(\0)$ by combining the results of Proposition \ref{thm: vector to matrix}, along with Lemmas \ref{lem: binary oracle conversion in spectral method} and \ref{lem: phase oracle conversion in spectral method}.
\end{proof}

\subsection{Estimating sparse Hessians}

Estimating sparse Hessian is useful in optimization \cite{fletcher1997computing,d2638b8e-8fdc-3d6d-a044-1384db1d1f3a,coleman1984estimation,coleman1984large,gebremedhin2005color}.
In this section, we show that sparse Hessians can be estimated more efficiently on a quantum computer.
Below, when we say $H$ is $s$-sparse, we mean it contains at most $s$ nonzero entries in each row and column.

\begin{lem}[Schwartz-Zippel Lemma]
Let $P\in R[x_{1},x_{2},\ldots ,x_{n}]$ be a non-zero polynomial of total degree $d \geq 0$ over an integral domain $R$. Let $S$ be a finite subset of $R$ and let $r_1, r_2, \ldots, r_n$ be selected at random independently and uniformly from $S$. Then
\[
\Pr[P(r_{1},r_{2},\ldots ,r_{n})=0]\leq {\frac {d}{|S|}}.
\]
\end{lem}
To apply the above result, we modify the grid of sampling points as follows: Let $q>2$ be a prime number, define
\be
\label{new sampling points}
S_q:=\left\{ \frac{k}{q} : k=-\frac{q-1}{2}, \cdots, \frac{q-1}{2} \right\} \subset ( -\frac{1}{2}, \frac{1}{2} ).
\ee
To clarify and avoid any confusion with $G_n$, we change the notation to $S_q$. Below, with $\mathbb{Z}_q$, we mean $\{-\frac{q-1}{2}, \cdots, \frac{q-1}{2}\}$, which is an integral domain.

The following result is a generalization of \cite[Theorems 15 and 17]{montanaro_shao-tqc} for $\mathbb{Z}_q$.

\begin{lem}
\label{lem: key lemma for sparse Hessian}
Assume that $d\in\mathbb{N}, q$ is a prime number.
Let $H=(h_{ij})_{d\times d} \in S_q^{d\times d}$ be a symmetric matrix. Let $f(\x)=\x^T H \x$ and the oracle be $O_f: \Ket{\x} \rightarrow e^{i q f(\x)} \Ket{\x}$ for any $\x\in S_q^d$.
\begin{itemize}
    \item If $H$ is $s$-sparse, then there is a quantum algorithm that returns $H$ with probability at least $0.99$ using $O\left(s\log (qd)\right)$ queries.
    \item If $H$ has at most $m$ nonzero entries, then there is a quantum algorithm that returns $H$ with probability at least $0.99$ using $O\left(\sqrt{m \log(qd)}\right)$ queries.
\end{itemize}
\end{lem}

\begin{proof}
For any fixed $\y\in \mathbb{Z}_q^d$, let $g(\x)=\frac{1}{2} (f(\x+\y)-f(\x))=\x \cdot (H\y)+\frac{1}{2} f(\y)$. With $O(1)$ applications of $O_f$, we have the following state
\be
\label{0819-eq1}
\frac{1}{\sqrt{q^d}}\sum_{\x\in S_q^d} e^{2\pi i q g(\x) } \Ket{\x}
=
\frac{e^{\pi i q f(\y)}}{\sqrt{q^d}} \sum_{\x\in S_q^d} e^{2\pi i q \x \cdot H\y} \Ket{\x}.
\ee 
To apply the quantum Fourier transform, we consider the first register of the above state: here we assume that $x_1={k_1}/{q}$ and $h_{1j}={\ell_{1j}}/{q}$ for some $k_1, \ell_{1j}\in \mathbb{Z}_q$,
\[
\sum_{x_1 \in S_q} e^{2\pi i q x_1 \sum_{j=1}^d h_{1j} y_j} \Ket{x_1}
=
\sum_{k_1 \in \mathbb{Z}_q} e^{2\pi i q \frac{k_1}{q}  \frac{\sum_{j=1}^d \ell_{1j} y_j}{q} } \Ket{k_1}.
\]
Applying the inverse of the quantum Fourier transform, we obtain $\sum_{j=1}^d \ell_{1j} y_j \mod \mathbb{Z}_q$. Here, to apply the quantum Fourier transform, it is important to make sure that $\y\in\mathbb{Z}_q^d$ rather than $S_q^d$.

For convenience, denote $H=L/q$, where $L\in \mathbb{Z}_q^{d\times d}$. As a result of the above analysis, we obtain $L\y \mod \mathbb{Z}_q$ for any given $\y\in\{0,1\}^d \subset   \mathbb{Z}_q^d$. We now randomly choose $k$ samples $\y_1,\ldots,\y_k\in\{0,1\}^d$, then we obtain $L\y_1,\ldots,L\y_k \mod \mathbb{Z}_q$ using $O(k)$ quantum queries.

Note that for any $\x\neq \x'$, by the Schwartz-Zippel lemma, the probability that $\x\cdot \y_i = \x' \cdot \y_i$ for all $i$ is at most $2^{-k}$. By the union bound, for any $\x\in \mathbb{Z}_q^d$, the probability that there exists an $\x'\in \mathbb{Z}_q^d$ such that $\x'\neq \x, |\x'|_0\leq s$ and $\x\cdot \y_i = \x' \cdot \y_i$ for all $i$ is bounded by\footnote{Here $|\x'|_0$ refers to the Hamming weight of $\x'$, i.e., the number of nonzero entries of $\x'$.}
\[
\sum_{l=0}^s \binom{d}{l} (q-1)^l 2^{-k}
=O\left(s \binom{d}{s} q^{s}  2^{-k}\right)
=O\left(2^{\log(s) + s \log(de/s)+ s \log q -k  } \right).
\]

We apply this bound to all rows of $L$ via a union bound, then the probability that for any row $\x$ of $L$, there is an $\x'\in \mathbb{Z}_q$ such that $\x'\neq \x, |\x'|_0\leq s$ and $\x\cdot \y_i = \x' \cdot \y_i$ for all $i$ is bounded by $O\left(d \cdot 2^{\log(s) + s \log(de/s)+ s \log q -k  } \right).$ It suffices to choose $k=O(s\log(qd/s))$ to make this probability arbitrary small. 
From $\y_1,\ldots,\y_k$ and $L\y_1,\ldots,L\y_k$, we then can determine $L$ with high probability.
This proves the first claim.

For the second claim, we first apply the above algorithm by setting $s=\sqrt{m/\log (qd)}$. This learns all rows of $L$ with at most $\sqrt{m/\log (qd)}$ nonzero entries and identifies all dense rows with more than $\sqrt{m/\log (qd)}$ nonzero entries. And there are at most $\sqrt{m\log (qd)}$ dense rows. For those rows, we use (\ref{0819-eq1}) to learn them by setting $\y$ as the corresponding standard basis vector. The overall cost is $O\left(\sqrt{m \log(qd)} \right)$.
\end{proof}

\begin{prop}
\label{prop: real case}
Assume that $d\in\mathbb{N}, \eta \in \mb R_+$ and $q$ is a prime number.
Let $H\in \mathbb{R}^{d\times d}$ and $\|H\|_{\max}\leq 1/2$. 
Assume $H = \widetilde{H} + E$, where $\widetilde{H}\in S_q^{d\times d}$ and $\|E\|_{\max} \leq  \eta$. 
Let $f(\x)=\x^T H \x$ and the oracle be $O_f: \Ket{\x} \rightarrow e^{i qf(\x)} \Ket{\x}$ for any $\x\in S_q^d$.
\begin{itemize}
    \item If $H$ is $s$-sparse, then there is a quantum algorithm that returns $\widetilde{H} $ with probability at least $0.99$ using $O(s\log^2 (qd))$ queries if $\eta=o(1/s q)$.
    \item If $H$ has at most $m$ nonzero entries, then there is a quantum algorithm that returns $\widetilde{H} $ with probability at least $0.99$ using $O\left(\sqrt{m}\log^{1.5} (qd)\right)$ queries if $\eta=o\left(1/\sqrt{m} q\right)$.
\end{itemize}
\end{prop}

In the above, the assumption $\|H\|_{\max}\leq 1/2$ ensures the existence of the decomposition $H=\widetilde{H}+E$. 

\begin{proof}
Our goal is to choose an appropriate $\eta$ such that $\widetilde{H}$ can be recovered with high probability. We use the same idea as the algorithm presented in Lemma \ref{lem: key lemma for sparse Hessian}, which is mainly based on (\ref{0819-eq1}). From the proof of Lemma \ref{lem: key lemma for sparse Hessian}, it is sufficient to prove the first claim. So without loss of generality, we assume that $H$ is $s$-sparse. For the state (\ref{0819-eq1}), up to a global phase, it equals
\beas
\Ket{\psi}=\frac{1}{\sqrt{q^d}}\sum_{\x\in S_q^d} e^{ 2\pi i q \x \cdot H\y} \Ket{\x} 
&=& 
\frac{1}{\sqrt{q^d}}\sum_{\x\in S_q^d} e^{2\pi i q \x \cdot  (\widetilde{H}\y+E\y)} \Ket{\x} \\
&=& 
\bigotimes_{j=1}^d \frac{1}{\sqrt{q}} \sum_{x_j\in S_q} 
e^{2\pi i q x_j(\widetilde{H}\y)_j}
e^{2\pi i q x_j(E\y)_j} 
\Ket{x_j} .
\eeas
For each $j$, we have that
\beas
\left\|\frac{1}{\sqrt{q}} \sum_{x_j\in S_q} 
(e^{2\pi i q x_j(E\y)_j}  - 1)
\Ket{x_j} \right\|^2
&=&
\frac{4}{q} \sum_{x_j\in S_q} \sin^2(\pi q x_j (E\y)_j) \\
&\leq&
\frac{4}{q} \sum_{x_j\in S_q} \pi^2 q^2 x_j^2 (E\y)_j^2 \\
&<& 
\frac{\pi^2 q^2 s^2 \eta^2}{4}.
\eeas
In the last step, we use the facts that $|x_j| < 1/2, |y_j|\leq 1$ and $E$ is $s$-sparse with $\|E\|_{\max}\leq \eta$. So we choose $\eta$ such that $s q \eta = o(1)$ is a small constant. With this choice, for each $j$,
\[
\left\|\frac{1}{\sqrt{q}} \sum_{x_j\in S_q} 
e^{2\pi i q x_j(\widetilde{H}\y)_j}
e^{2\pi i q x_j(E\y)_j} 
\Ket{x_j} 
-
\frac{1}{\sqrt{q}} \sum_{x_j\in S_q} 
e^{2\pi i q x_j(\widetilde{H}\y)_j}
\Ket{x_j} 
\right\| = o(1).
\]
From the state $\frac{1}{\sqrt{q}} \sum_{x_j\in S_q} 
e^{2\pi i q x_j(\widetilde{H}\y)_j}
\Ket{x_j}$, we can recover $(\widetilde{H}\y)_j$ with probability 1 using the quantum Fourier transform.  
Since $\Ket{\psi}$ is a product state and $j\in [d]$ are independent of each other, we can recover all $(\widetilde{H}\y)_j$ with high probability using $O(\log d)$ measurements. Namely, for each $\y$, we can obtain $\widetilde{H}\y$.
With this result and the analysis for Lemma \ref{lem: key lemma for sparse Hessian}, we obtain the claimed results.
\end{proof}

\begin{rmk}
In the above proposition, the condition $\|E\|_{\max} \leq o(1/sq)$ indicates that entries of $H$ have the form ${k}/{q} + o(1/sq)$ for $k \in \mathbb{Z}_q$. This condition might be strong, as the error is usually bounded by $1/2q$ rather than $o(1/sq)$. However, from the choice of $\eta$ in the above proof, we can relax this to $\|E\|_{\max} \leq 1/2q$ by considering $\y\in\{0,1\}^d/(\alpha s)$ for some large constant $\alpha$. In this case, the oracle we need becomes $\Ket{\x} \rightarrow e^{i s q f(\x)}\Ket{\x}$. Consequently, the query complexity in our final theorem below (i.e., Theorem \ref{thm: Sparse Hessian estimation using spectral method}) will be multiplied by a factor of $O(s)$. If $s$ is not excessively large, then this should not be an issue.
In Proposition \ref{prop: real case}, the choice of $\eta$ is indeed not ideal. This remark clarifies that $\eta$ can be adjusted to $1/2q$, though this would introduce an additional factor of $s$ in the complexity. The latter choice is standard, while the former is somewhat restrictive.
\end{rmk}

As shown below, the above result also holds when $f(\x)$ is close to a quadratic form $\x^TH\x$. We consider this under two cases: $\|H\|_{\max}\leq 1/2$ and $\|H\|_{\max}\leq M/2$ for some general $M>1$.

\begin{prop}
\label{prop: a little more general 1}
Assume that $d,\ell \in\mathbb{N}, \varepsilon,\eta \in \mathbb{R}_+, q=O(1/\varepsilon)$ is prime.
Let $H\in \mathbb{R}^{d\times d}$ and $\|H\|_{\max}\leq 1/2$. 
Assume that $H = \widetilde{H} + E$, where $\widetilde{H}\in S_q^{d\times d}$ and $\|E\|_{\max} \leq \eta$. 
Assume that $|f(\z)-\z^T H \z|\leq O(\varepsilon)$, where $\z=\x+\y$, holds for all but a 1/1000 fraction of points $\x\in S_q^d$ and all $\y\in\{0,1\}^d$. Let the oracle be $O_f: \Ket{\z} \rightarrow e^{i q f(\z)} \Ket{\z}$ for any $\z\in  S_q^d+\{0,1\}^d$.
\begin{itemize}
    \item If $H$ is $s$-sparse, then there is a quantum algorithm that returns $\widetilde{H}$ with probability at least $0.99$ using $O(s \log^2 (qd))$ queries if $\eta=o(1/sq)$.
    \item If $H$ contains $m$ nonzero entries, then there is a quantum algorithm that returns $\widetilde{H}$ with probability at least $0.99$ using $O\left(\sqrt{m}\log^{1.5} (qd)\right)$ queries if $\eta=o\left(1/\sqrt{m}q\right)$.
\end{itemize}
\end{prop}

\begin{proof}
The key idea of the proof of Proposition \ref{prop: real case} is based on the state (\ref{0819-eq1}) generated by using the oracle to query $h(\x):=\x^T H \x$. 
By assumption, for any $\y\in\{0,1\}^d$, when the condition $|f(\z)-h(\z)|\leq O(\varepsilon)$ is satisfied for all $\z=\x+\y$ with $\x\in W \subset S_q^d$, and the size $\#(W)\geq 999 q^d/1000$, it is straightforward to verify that
\beas
&& \left\|\frac{1}{\sqrt{q^d}}\sum_{\x\in W} e^{ 2\pi i q \frac{1}{2} (h(\x+\y)-h(\x))} \Ket{\x}
-
\frac{1}{\sqrt{q^d}}\sum_{\x\in W} e^{ 2\pi i q \frac{1}{2} (f(\x+\y)-f(\x))} \Ket{\x}\right\|^2 \\
&\leq& 
\frac{1}{q^d}
\sum_{\x\in W}
q^2 \pi^2 \Big|(h(\x+\y)-h(\x)) - (f(\x+\y)-f(\x) )\Big|^2 \\
&=& O(q^2 \varepsilon^2).
\eeas
For $\x\notin W$, we have
\beas
&& \left\|\frac{1}{\sqrt{q^d}}\sum_{\x\notin W} e^{ 2\pi i q \frac{1}{2} (h(\x+\y)-h(\x))} \Ket{\x}
-
\frac{1}{\sqrt{q^d}}\sum_{\x\notin W} e^{ 2\pi i q \frac{1}{2} (f(\x+\y)-f(\x))} \Ket{\x}\right\|^2 \\
&\leq& 
\frac{1}{q^d}
\sum_{\x\notin W}
4 q^2 \leq \frac{4}{1000}.
\eeas
To ensure these estimates are bounded by a small constant, it suffices to choose $q=O(1/\varepsilon)$, which is exactly what we claim in the statement of this proposition. Thus, using $O_f$, we also obtain a state close to $\Ket{\psi}$ required in the proof of Proposition \ref{prop: real case}. 
\end{proof}

\begin{cor}
\label{cor: a little more general 2}
Assume that $d\in\mathbb{N}, \varepsilon,\eta \in \mathbb{R}_+, M>1$, and $q=O( M/\varepsilon)$ is prime.
Let $H\in  \mathbb{R}^{d\times d}$ and $\|H\|_{\max}\leq M/2$. 
If $H = M \cdot \widetilde{H} + M \cdot  E$, where $\widetilde{H}\in  S_q^{d\times d}$ and $\|E\|_{\max} \leq \eta$, and assume that $|f(\z)-\z^T H \z|\leq O(\varepsilon)$, where $\z=\x+\y$, holds for all but a 1/1000 fraction of points $\x\in S_q^d$ and all $\y\in\{0,1\}^d$. Let the oracle be $O_f: \Ket{\z} \rightarrow e^{i  f(\z)} \Ket{\z}$ for any $\z\in S_q^d+\{0,1\}^d$.
\begin{itemize}
    \item If $H$ is $s$-sparse, then there is a quantum algorithm that returns $\widetilde{H}$ with probability at least $0.99$ using $O\left(\frac{s}{\varepsilon}\log^2 (qd)\right)$ queries if $\eta=o(1/sq)$.
    \item If $H$ contains $m$ nonzero entries, then there is a quantum algorithm that returns $\widetilde{H}$ with probability at least $0.99$ using $O\left(\frac{\sqrt{m}}{\varepsilon}\log^2 (qd)\right)$ queries if $\eta=o\left(1/\sqrt{m}q\right)$.
\end{itemize}
\end{cor}

\begin{proof}
It suffices to consider $\tilde{f}(\x) = f(\x)/M$ and $h(\x)=\x^T H \x/M$. Then for any $\x\in S_q^d$ and $\y\in\{0,1\}^d$, we have $|\tilde{f}(\x+\y) - h(\x+\y)|\leq O(\varepsilon/M)$. To use 
Proposition \ref{prop: a little more general 1}, the oracle we need is 
\[
O_{\tilde{f}}: \Ket{\x} \rightarrow e^{i q \tilde{f}(\x)} \Ket{\x}
= e^{i {f(\x)}/{\varepsilon} } \Ket{\x} = O_f^{1/\varepsilon} \Ket{\x}
\]
for any $\x \in S_q^d$. The results now follow from Proposition \ref{prop: a little more general 1}.
\end{proof}

Finally, combining Proposition \ref{cor: a little more general 2} and the idea of the proof of Theorem \ref{thm: hessian using spectral}, we obtain the following result. It coincides with Theorem \ref{thm: hessian using spectral} in the worst cases.

\begin{thm}[Sparse Hessian estimation using spectral method]
\label{thm: Sparse Hessian estimation using spectral method}
    Under the same assumption as in Theorem \ref{thm: hessian using spectral}.
    Assume that the Hessian $H=\H_f(\boldsymbol{0})$ satisfies $\|H\|_{\max}\leq M/2$ and $H = M \cdot \widetilde{H} + M \cdot  E$, where $\widetilde{H}\in  S_q^{d\times d}$ and $\|E\|_{\max} \leq \eta$, where $q=O(M/\varepsilon)$ is prime.
     Then there is a quantum algorithm that returns $\widetilde{H}$ exactly using 
    \begin{itemize}
    \item $\widetilde{O}(s/\varepsilon)$ queries to $O_{f_1}, O_{f_2}$ if $H$ is $s$-sparse and $\eta=o(1/sq)$.
    \item $\widetilde{O}\left(\sqrt{m}/\varepsilon\right)$ queries to $O_{f_1}, O_{f_2}$ if $H$ has at most $m$ nonzero entries and $\eta=o\left(1/\sqrt{m}q\right)$.
    \end{itemize}
\end{thm}

Different from Theorem \ref{thm: hessian using spectral}, which is about absolute error estimation, Theorem \ref{thm: Sparse Hessian estimation using spectral method} estimates $H$ under the relative error.


\section{Quantum algorithms for Hessian estimation using finite difference formula}
\label{section: Quantum algorithms for Hessian estimation using finite difference formula}

In this section, we will use the finite difference formula to estimate the Hessian. From our results below, we will see that the spectral method can demonstrate significant advantages over the finite difference method for Hessian estimation. The key ideas are similar, however, the error analysis for the finite difference method is more involved, as we will see shortly.

\subsection{Estimating Hessians in the general case}

The following result is similar to but slightly weaker than Proposition \ref{thm: vector to matrix}. 

\begin{prop}
\label{new prop: gradient to Hessian}
    Let $c \in \mb R, a, \rho, \varepsilon< M \in \mb R_+$ as in Lemma \ref{lem: GAW's gradient estimation thm}. Given a matrix $H \in \mb R^{d \times d}$ such that $\|H\|_{\rm{max}} \leq M$, if we have $h: \mathbb{R}^d \rightarrow \mb R$ such that
    \be\label{property 0}
    \left| \frac{1}{2} \Big( h(\x+\y) - h(\x) \Big) -
     \x^T H \y - \frac{1}{2} \y^T H \y \right|
     \leq O(a\varepsilon),
    \ee
    for all but a $1/1000$ fraction of points $\boldsymbol{x} \in aG_n^d$ and any fixed $\y\in\{\Ket{i}: i\in[d]\}$, and access to its phase oracle $O_h: |\boldsymbol{x}\ket \rightarrow e^{ih(\boldsymbol{x})} |\boldsymbol{x}\ket$, then there is a quantum algorithm that calculates a matrix $\widetilde{H}$ with probability at least $1-\rho$ such that
    $$
    \|\widetilde{H}-H\|_{\rm{max}} \leq \varepsilon 
    $$
    using $O\left(\frac{d}{a\varepsilon} \log\big(\frac{d}{\rho}\big)\right)$ queries to the phase oracle $O_h$.
\end{prop}

\begin{proof}
The condition in Proposition \ref{thm: vector to matrix} implies (\ref{property 0}), and the latter is indeed the condition used in the proof of Proposition \ref{thm: vector to matrix}. Therefore, the claim here follows directly from Proposition \ref{thm: vector to matrix}.
\end{proof}

As given in Definition \ref{def: finite difference formula of Hessian},
in the finite difference method, which is based on Taylor expansion, we have $f_{(2m)}(\x) \approx \x^T H_f(\0) \x$, see Lemma \ref{app: lemma 4} in the appendix. The error between them is highly affected by the norm of $\x$. This is reasonable as the Hessian is a local property, so we need to choose $\x$ having a small norm. In Proposition \ref{new prop: gradient to Hessian}, to satisfy (\ref{property 0}), one natural option is letting $f_{(2m)}(\x+\y) \approx (\x+\y)^T H_f(\0) (\x+\y)$ when $\y=\Ket{i}$. Generally, this is not possible as $\x+\y$ can have a large norm. The following result states a case in which this is possible.

\begin{thm} 
\label{new thm: Hessian estimation using finite difference formula}
    Let $m \in \mathbb{N}, R,B\in\mb R_+$ satisfy $m \geq \log(dB/\varepsilon)$.
    Assume that $f:[-R, R]^d$ $\rightarrow \mathbb{R}$ is $(2 m+1)$-times differentiable and
    \be
    \label{property 1}
    \left|\partial_{\boldsymbol{r}}^{2 m+1} f(\x)\right| \leq B \quad \text { for all } \x \in[-R, R], \text{ where } \r=\x/\|\x\|.
    \ee
    Given access to its phase oracle $O_f: |\x\ket \rightarrow e^{if(\x)} |\x\ket$, then there is a quantum algorithm that can estimate its Hessian at point $\boldsymbol{0}$ with accurate $\varepsilon$ and probability at least $1-\rho$, using $\widetilde{O}\left(d^{1.5}m/\varepsilon+md/R\varepsilon\right)$ queries to $O_f$.
\end{thm}

\begin{proof}
    By Lemma \ref{app: lemma 4} in the appendix, the condition in Proposition \ref{new prop: gradient to Hessian} is satisfied for $f_{(2m)}(\x)$ and $\y\in\{\Ket{i}: i\in[d]\}$. Moreover, from Lemma \ref{app: lemma 4}, $a\approx 1/\sqrt{d}m$. Another natural restriction on $a$ is $ma/2\leq R$, i.e., $a\leq 2R/m$. By Lemma \ref{lem: bound of error for 1 dim fin diff formula}, we have
    $\sum_{t=1}^m |a_{t}^{(2 m)} |<{\pi^2}/{6}.$ So we can implement a phase oracle to query $f_{(2m)}$ with $O(m)$ applications of $O_f$. The result now follows from Proposition \ref{new prop: gradient to Hessian}.
\end{proof}

The condition (\ref{property 1}) is quite strong because $B$ usually increases with respect to $m$, such as for $f(x)=e^{\alpha x}$ with $B=O(\alpha^m)$. In this case, the condition $m \geq \log(dB/\varepsilon)$ is hard to satisfy. Thus, the above theorem is useful when $B$ is small. 

One may realize that all these conditions in the above theorem are mainly caused by the choice of $\y=\Ket{i}$, which has a norm of $1$. Indeed, we can also choose $\y=a\Ket{i}$ such that it has a norm of the same order as $\x$ in Proposition \ref{new prop: gradient to Hessian}. With this choice, we will obtain an approximation of $a H\Ket{i}$, see the proof of Proposition  \ref{thm: vector to matrix}. This means that the complexity in Proposition \ref{new prop: gradient to Hessian} will be multiplied by an extra factor of $O(1/a)$. In summary, if we choose $\y=a\Ket{i}$, then the query complexity will be $\widetilde{O}(d/a^2\varepsilon)$. In the finite difference method, $a \approx 1/\sqrt{d}$. So the overall query complexity is $\widetilde{O}(d^2/\varepsilon)$. This indeed shows no quantum speedup in terms of $d$ over classical algorithms. However, the restriction on $B$ can be relaxed, and the dependence on $\varepsilon$ is better than classical algorithms. We summarize this in the following theorem.



\begin{thm} 
\label{thm: hessian using difference formula}
    Let $\rho, \varepsilon, c \in \mb R_+$, and suppose $f: \mathbb{R}^d \rightarrow \mathbb{R}$ is analytic such that for every $k \in \mathbb{N}$ and $\alpha \in \mb N_0^d$ with $|\alpha|=k$,
    $$
    \left|\partial^\alpha f(\boldsymbol{0})\right| \leq c^k k^{\frac{k}{2}} .
    $$
    Given access to its phase oracle $O_f: |\x\ket \rightarrow e^{if(\x)} |\x\ket$, then there is a quantum algorithm that can estimate its Hessian at point $\boldsymbol{0}$ with accurate $\varepsilon$ and probability at least $1-\rho$, using $O\left(\frac{c d^{2}}{\varepsilon} \log\big(\frac{c\sqrt{d}}{\varepsilon}\big) \log\big(\frac{d}{\rho}\big)\right)$ queries to $O_f$.
\end{thm}

\begin{proof}
    For $\boldsymbol{x} \in \mb R^d$, using the finite difference method, we have an approximation of $\bra \boldsymbol{x}| \mathbf{H}_{f}(\0) |\boldsymbol{x}\ket$ by Definition \ref{def: finite difference formula of Hessian}, where $\H_f(\0)$ denotes the Hessian of function $f$ at point $\boldsymbol{0}$:
    $$
    f_{(2 m)}(\boldsymbol{x}):=\sum_{\substack{t=-m}}^m  a_{t}^{(2m)} f(t \boldsymbol{x}).
    $$
    By Lemma \ref{lem: bound of error in fin diff formula} in the appendix, for at least $999/1000$ fraction of inputs $\x,\y\in aG_n^d$, we have $| f_{(2 m)}(\boldsymbol{x}+\y)$ $-(\x+\y)^{\mathrm{T}} \H_f(\0) (\x+\y)| \leq \sum_{k=2 m+1}^{\infty}(13 a c m \sqrt{d})^k$. 
    As a result, we have 
    \[
    \left|\frac{1}{2}\Big(f_{(2 m)}(\x+\y)-f_{(2 m)}(\x)\Big) - \x^{\mathrm{T}} \H_f(\0) \y - \frac{1}{2} \y^{\mathrm{T}} \H_f(\0) \y\right| \leq 
    \sum_{k=2 m+1}^{\infty}(13 a c m \sqrt{d})^k.
    \]
    This is similar to the condition (\ref{property 0}) stated in Proposition \ref{new prop: gradient to Hessian}.
    We now choose $a$ such that
    $a^{-1}:=14 cm \sqrt{d}(196 \cdot 8 \cdot 42 \pi c m \sqrt{d} / \varepsilon)^{\frac{1}{2m}} = \widetilde{O}(\sqrt{d}) $, then we have $13 a c m \sqrt{d}=\frac{13}{14}(196 \cdot 8 \cdot 42 \pi c m \sqrt{d} / \varepsilon)^{-\frac{1}{2m}}$. Moreover, 
    \beas
        \sum_{k=2 m+1}^{\infty}(13 a c m \sqrt{d})^k &=& (13 a c m \sqrt{d})^{2 m+1} \sum_{k=0}^{\infty}(13 a c m \sqrt{d})^k \\
        & \leq& \frac{\varepsilon}{196 \cdot 8 \cdot 42 \pi c m \sqrt{d}}(196 \cdot 8 \cdot 42 \pi c m \sqrt{d} / \varepsilon)^{-\frac{1}{2 m}} \sum_{k=0}^{\infty}\left(\frac{13}{14}\right)^k \\
        &=& \frac{\varepsilon}{14 cm \sqrt{d} \cdot 8 \cdot 42 \pi}(196 \cdot 8 \cdot 42 \pi cm \sqrt{d} / \varepsilon)^{-\frac{1}{2 m}} \\
        &=&\frac{\varepsilon a}{8 \cdot 42 \pi}.
    \eeas
    Therefore, the function $f_{(2m)}(\x)$ satisfies condition (\ref{property 0}). Now following the procedure described in the proof of Proposition \ref{new prop: gradient to Hessian}, invoking the gradient estimation algorithm $d$ times can approximate the Hessian matrix $\H_f(\0)$. 
    Note that now $\y\in aG_n^d$, in the procedure, we have to choose $\y=\frac{a}{2}\Ket{i}$. Using the notation in that proof, the error between $\g_i$ and $H\Ket{i}$ is bounded by $O(\sqrt{d} \varepsilon)$, so we have to choose a much smaller $\varepsilon$ to ensure this error is small. This leads to an extra factor of $\sqrt{d}$ in the overall complexity.
\end{proof}

\subsection{Estimating sparse Hessians}

As shown in Lemma \ref{lem: key lemma for sparse Hessian}, in the sparse case, we need to choose $\y\in\{0,1\}^d \subset \mathbb{Z}_q^d$ randomly. This requirement is too strong for Proposition \ref{new prop: gradient to Hessian} as $\y=\Ket{i}$. Fortunately, Lemma \ref{app: lemma 4} used in the proof of Theorem \ref{new thm: Hessian estimation using finite difference formula} holds for any $\y\neq 0$ with $\|\y\|=1$. As a result, we will replace $\y\in\{0,1\}^d$ with 
$\y\in \{\pm 1/\sqrt{d}\}^d$ below. All the arguments remain the same.

\begin{lem}
Assume that $d\in\mathbb{N}, q$ is a prime number.
Let $H=(h_{ij})_{d\times d} \in S_q^{d\times d}$ be a symmetric matrix. Let $f(\x)=\x^T H \x$ and the oracle be $O_f: \Ket{\x} \rightarrow e^{i q \sqrt{d} f(\x)} \Ket{\x}$ for any $\x\in S_q^d$.
\begin{itemize}
    \item If $H$ is $s$-sparse, then there is a quantum algorithm that returns $H$ with probability at least $0.99$ using $O\left(s\log (qd)\right)$ queries.
    \item If $H$ has at most $m$ nonzero entries, then there is a quantum algorithm that returns $H$ with probability at least $0.99$ using $O\left(\sqrt{m \log(qd)}\right)$ queries.
\end{itemize}
\end{lem}

\begin{proof}
The proof is similar to that of Lemma \ref{lem: key lemma for sparse Hessian}. 
As $\y$ is normalized, we choose $\y=(\pm 1/\sqrt{d},\cdots,$ $\pm 1/\sqrt{d})$.  The oracle is changed accordingly.
\end{proof}

\begin{prop}
\label{prop: sparse Hessian new 1}
Assume that $d\in\mathbb{N}, \eta \in \mb R_+$ and $q$ is a prime number.
Let $H\in \mathbb{R}^{d\times d}$ and $\|H\|_{\max}\leq 1/2$. 
Assume that $H = \widetilde{H} + E$, where $\widetilde{H}\in S_q^{d\times d}$ and $\|E\|_{\max} \leq  \eta$. 
Let $f(\x)=\x^T H \x$ and the oracle be $O_f: \Ket{\x} \rightarrow e^{i q\sqrt{d} f(\x)} \Ket{\x}$ for any $\x\in S_q^d$.
\begin{itemize}
    \item If $H$ is $s$-sparse, then there is a quantum algorithm that returns $\widetilde{H} $ with probability at least $0.99$ using $O(s\log^2 (qd))$ queries if $\eta=\min\{ 1/2q, o(\sqrt{d} /s q)\}$.
    \item If $H$ has at most $m$ nonzero entries, then there is a quantum algorithm that returns $\widetilde{H} $ with probability at least $0.99$ using $O\left(\sqrt{m}\log^{1.5} (qd)\right)$ queries if $\eta=\min\{1/2q$, $ o(\sqrt{d}/\sqrt{m} q)\}$.
\end{itemize}
\end{prop}

\begin{proof}
The proof here is similar to that of Proposition \ref{prop: real case}. We will use the notation in the proof of Proposition \ref{prop: real case} below. As $\y$ is normalized, we have
\beas
\left\|\frac{1}{\sqrt{q}} \sum_{x_j\in S_q} 
\left(e^{2\pi i q x_j(E\y)_j}  - 1\right)
\Ket{x_j} \right\|^2
<
\frac{\pi^2 q^2 s^2 \eta^2}{4 d}.
\eeas
So it suffices to choose $\eta$ such that the above estimate is $o(1)$, i.e., $\eta = o(\sqrt{d}/sq)$. Because of the decomposition $H = \widetilde{H} + E$, we also have $\eta \leq 1/2q$.
\end{proof}

\begin{prop}
\label{prop: finite difference sparse 1}
Assume that $d,\ell \in\mathbb{N}, \varepsilon,\eta \in \mathbb{R}_+, q=O(1/\varepsilon)$ is prime.
Let $H\in \mathbb{R}^{d\times d}$ and $\|H\|_{\max}\leq 1/2$. 
Assume that $H = \widetilde{H} + E$, where $\widetilde{H}\in S_q^{d\times d}$ and $\|E\|_{\max} \leq \eta$. 
Let
\[
\left|\frac{1}{2}\Big( f(\x+\y)-f(\x) \Big) - \frac{1}{2}\Big((\x+\y)^T H (\x+\y) - \x^T H \x \Big) \right|
\leq O(\varepsilon)
\]
for all but a 1/1000 fraction of points $\x\in S_q^d$ and all $\y\in\{\pm1/\sqrt{d}\}^d$. Let the oracle be $O_f: \Ket{\x} \rightarrow e^{i q \sqrt{d} f(\x)} \Ket{\x}$ for any $\x\in  S_q^d$.
\begin{itemize}
    \item If $H$ is $s$-sparse, then there is a quantum algorithm that returns $\widetilde{H}$ with probability at least $0.99$ using $O(s \log^2 (qd))$ queries if  $\eta=\min\{ 1/2q, o(\sqrt{d} /s q)\}$.
    \item If $H$ contains $m$ nonzero entries, then there is a quantum algorithm that returns $\widetilde{H}$ with probability at least $0.99$ using $O\left(\sqrt{m}\log^{1.5} (qd)\right)$ queries if $\eta=\min\{1/2q$, $ o(\sqrt{d}/\sqrt{m} q)\}$.
\end{itemize}
\end{prop}

\begin{proof}
The proof is similar to that of Proposition \ref{prop: a little more general 1}. Here, we should use Proposition \ref{prop: sparse Hessian new 1} instead.
\end{proof}

\begin{prop}
\label{prop: finite difference sparse 2}
Assume that $d\in\mathbb{N}, a,\varepsilon,\eta \in \mathbb{R}_+, M>1, q=O( M/\varepsilon)$ is prime.
Let $H\in  \mathbb{R}^{d\times d}$ and $\|H\|_{\max}\leq M/2$. 
Assume that $H = M \cdot \widetilde{H} + M \cdot  E$, where $\widetilde{H}\in  S_q^{d\times d}$ and $\|E\|_{\max} \leq \eta$. 
Let 
\[
\left|\frac{1}{2}\Big( f(\x+\y)-f(\x) \Big) - \frac{1}{2}\Big((\x+\y)^T H (\x+\y) - \x^T H \x \Big) \right|
\leq O(a \varepsilon)
\]
for all but a 1/1000 fraction of points $\x\in a S_q^d$ and all $\y\in\{\pm1/\sqrt{d}\}^d$. 
Let the oracle be $O_f: \Ket{\x} \rightarrow e^{i \frac{f(a\x+\y) - f(a\x)}{2}} \Ket{\x}$ for any $\x\in S_q^d$.
\begin{itemize}
    \item If $H$ is $s$-sparse, then there is a quantum algorithm that returns $\widetilde{H}$ with probability at least $0.99$ using $O\left(\frac{s\sqrt{d}}{a\varepsilon}\log^2 (qd)\right)$ queries if $\eta=\min\{ 1/2q, o(\sqrt{d} /s q)\}$.
    \item If $H$ contains $m$ nonzero entries, then there is a quantum algorithm that returns $\widetilde{H}$ with probability at least $0.99$ using $O\left(\frac{\sqrt{md}}{a\varepsilon}\log^{1.5}(qd)\right)$ queries if $\eta=\min\{o(\sqrt{d}/\sqrt{m} q)$, $1/2q\}$.
\end{itemize}
\end{prop}

\begin{proof}
For any $\x\in S_q^d$, define 
\[
\tilde{f}(\x) = \frac{f(a\x+\y) - f(a\x)}{2a M}, \quad
h(\x) = \frac{(a\x+\y)^TH(a\x+\y) - a^2\x^TH\x}{2a M}.
\]
Then by assumption, for any $\x\in S_q^d$, we have 
$|\tilde{f}(\x) - h(\x) | \leq O(\varepsilon/M).$
So we choose $q=O(M/\varepsilon)$.
To use 
Proposition \ref{prop: finite difference sparse 1}, the oracle we need is 
\[
O_{\tilde{f}}: \Ket{\x} \rightarrow e^{i q \sqrt{d}  \tilde{f}(\x)} \Ket{\x}
= e^{i  \sqrt{d} \frac{f(a\x+\y) - f(a\x)}{2a \varepsilon} } \Ket{\x} 
= O_f^{\sqrt{d}/a\varepsilon} \Ket{\x}
\]
for any $\x \in S_q^d$. The results now follow from Proposition \ref{prop: finite difference sparse 1}.
\end{proof}

Finally, combining Proposition \ref{prop: finite difference sparse 2} and Theorem \ref{new thm: Hessian estimation using finite difference formula}, we obtain the following result.

\begin{thm}[Sparse Hessian estimation using finite difference formula] 
\label{thm: Sparse Hessian estimation using finite difference formula}
     Under the same assumption as in Theorem \ref{new thm: Hessian estimation using finite difference formula}. Assume that the Hessian $H=\H_f(\boldsymbol{0})$ satisfies $\|H\|_{\max}\leq M/2$ and $H = M \cdot \widetilde{H} + M \cdot  E$, where $\widetilde{H}\in  S_q^{d\times d}$ and $\|E\|_{\max} \leq \eta$, where $q=O(M/\varepsilon)$ is prime.
     Then there is a quantum algorithm that returns $\widetilde{H}$ exactly using 
    \begin{itemize}
    \item $\widetilde{O}(sd/\varepsilon)$ queries to $O_f$ if $H$ is $s$-sparse and $\eta=\min\{ 1/2q, o(\sqrt{d} /s q)\}$.
    \item $\widetilde{O}\left(\sqrt{m}d/\varepsilon\right)$ queries to $O_f$ if $H$ has at most $m$ nonzero entries and $\eta=\min\{1/2q$, $ o(\sqrt{d}/\sqrt{m} q)\}$.
    \end{itemize}
\end{thm}

The above theorem does not coincide with Theorem \ref{new thm: Hessian estimation using finite difference formula} in the worst case. This is mainly caused by the choice of $\y$. In Theorem \ref{new thm: Hessian estimation using finite difference formula}, $\y=\Ket{i}$, while in Theorem \ref{thm: Sparse Hessian estimation using finite difference formula}, $\y\in\{\pm1/\sqrt{d}\}^d$. However, it still saves a factor of $\sqrt{d}$ when the Hessian is sparse.

In the above theorem, there is no need to know the positions of nonzero entries. Classically, if this information is not known, it may still cost $O(d^2)$ queries to compute $H$. So the above quantum algorithm can achieve quadratic quantum speedups over the classical counterparts. In addition, as shown in Proposition \ref{prop: lower bound of real-valued Hessian estimation} below, the above quantum algorithm is optimal in terms of $d$. So quadratic speedup is the best we can expect for the finite difference method.



\subsection{Some final comments}

As we can see from the algorithms presented in this section and the previous section, error analysis is much easier for the spectral method than for the finite difference method. The condition in Proposition \ref{thm: vector to matrix} is easily satisfied by the spectral method. This is the main reason why quantum algorithms based on the spectral method outperform those based on the finite difference method. As analyzed in Subsection \ref{subsec: intro, comparison}, this difference essentially arises from the fundamental techniques of these two methods: the spectral method is based on the Cauchy integral formula, while the finite difference method is based on Taylor expansion.

Since the finite difference formula is based on Taylor expansion, achieving sufficiently small errors requires relatively small magnitudes of variables. In contrast, the spectral method has superior error performance. Additionally, as we can see in Subsection \ref{subsec: intro, comparison}, in the finite difference formula, the function values at points further from the target point have smaller linear combination coefficients, thus contributing less information. However, their larger magnitudes contribute more to the error, which is an undesirable situation.

\section{Lower bounds analysis}
\label{section: Lower bounds analysis}

For Hessian estimation, a lower bound of $\Omega(d/\log d)$ follows from \cite[Proposition 1]{montanaro2012quantum}. 

\begin{lem}
\label{lem: learning multilinear function}
    Let $f:\mathbb{F}_q^d \rightarrow \mathbb{F}_q$ be a multilinear degree $k$ polynomial over finite field $\mathbb{F}_q$. Then any quantum query algorithm that learns $f$ with bounded error must make $\Omega(d^{k-1})$ queries to $U_f$.
\end{lem}

If we are only given $O_f$, then to apply $U_f$ once, by quantum phase estimation (also see Proposition \ref{prop: Binary oracle and phase oracle}), we have to use $O(\log q)$ applications of $O_f$ as $\log q$ bits are enough to specify $f(\x)$ for any $\x\in \mathbb{F}_q^n$. So by Lemma \ref{lem: learning multilinear function}, a lower bound with respect to phase oracle $O_f$ is $\Omega(d^{k-1}/\log q)$. 

\begin{prop} \label{prop: lower bound}
Let $f$ be a real-valued (or complex-valued) function of $d$ variables with a phase oracle access $O_f$ (or $O_{f_1}, O_{f_2}$), then any quantum algorithm that outputs the Hessian of $f$ up to a constant error $\varepsilon\in(0,1/2)$ under matrix-max norm must make $\Omega(d/\log d)$ queries to $O_f$ (or $O_{f_1}, O_{f_2}$) to vectors in $G_n^d$.
\end{prop}

\begin{proof}
We first consider the lower bound for real-valued functions with oracle $O_f$.
Let $H$ be an unknown $d\times d$ Boolean Hermitian matrix such that the diagonal entries are 0.
Consider the function $f(\x)=\frac{1}{2}\langle \x|H|\x\rangle=\sum_{i<j} H_{ij} x_i x_j$, which is a typical form of multilinear degree-2 polynomials. 
For Hessian estimation, we will query points from $G_n^d$. Each point $\x=(x_1,\ldots,x_d)\in G_n^d$ satisfies $x_j = \frac{k_j}{2^n} + \frac{1}{2^{n+1}}\in(-1/2,1/2)$ for some $k_j \in \{-2^{n-1},\ldots,2^{n-1}-1\}$. To view them as vectors from $\mathbb{F}_q^d$ for some $q$ that will be determined shortly, we multiply them by $2^{n+1}$, i.e., we introduce $y_j = 2^{n+1} x_j= 2 k_j + 1\in(-2^n,2^n)$. 

Now let $g(\y)=f(2^{n+1}\x)$. For any $\x$, we have $|g(\y)|<(d^2-d) 2^n$.
We now choose $q=2(d^2-d) 2^n$ so that $g(\y)$ does not change even after modulo $q$ in the field $\mathbb{F}_q$. Now by Lemma \ref{lem: learning multilinear function}, the query complexity of computing $H$ exactly is lower bounded by $\Omega(d/\log q)=\Omega(d/(n+\log d))$. 
If $\varepsilon\in(0,1/2)$, then an $\varepsilon$-approximation of $H$ also allows us to determine $H$ as we know it is a Boolean matrix.
From Lemma \ref{lem: GAW's gradient estimation thm} and the proof of Theorem \ref{thm: vector to matrix}, $n=O(\log (d/\varepsilon))$. So we obtain a lower bound of $\Omega(d/\log d)$ when $\varepsilon \in (0,1/2)$ is a constant. 

Next, we consider lower bounds analysis for complex-valued functions. The argument is very similar to the above. By Theorem \ref{thm: spectral formula for gradient}, in the spectral method, the complex vectors we queried have the form $\delta \omega^k \x$, which is just a complex scaling of $\x$ for $\x\in G_n^d$. As a result, for the function we constructed above, we have
$f(\delta \omega^k \x) = \delta^2 \omega^{2k} f(\x)$. 
The real and imaginary parts are $f_1(\delta \omega^k \x)=\delta^2\cos(2\pi k/N) f(\x)$ and $f_2(\delta \omega^k \x)=\delta^2\sin(2\pi k/N) f(\x)$ respectively. They all contain the information of $f(\x)$. On a quantum computer, we represent them using finite-precision binary encodings, so similar to the above argument by scaling $f$ approximately, we can also obtain a lower bound of $\Omega(d/\log d)$. The scaling here depends on the precision of representing $f_1,f_2$, which further depends on $\delta^2 \cos(2\pi k/N)$ and $\delta^2 \sin(2\pi k/N)$. Note that $N=O(\log(\kappa/\varepsilon))$ for $\kappa=\max|f(\x)|=O(d^2)$, and from the proof of Theorem \ref{thm: gradient estimation using spectral method}, $\delta$ depends on the radius of convergence rate, which can be set as a constant for the function $f(\x)=\frac{1}{2}\langle \x|H|\x\rangle$. So $q$ should be set as large as $O(d^22^n\times \kappa/\varepsilon)=O(d^4 2^n)$ as $\varepsilon$ is constant in our setting. None of these changes the lower bounds much.
\end{proof}

For real-valued functions with oracle $O_f$, we can prove a slightly better lower bound of $\Omega(d/\varepsilon)$. We here use a similar idea of the lower bound proof of gradient estimation given in \cite{gilyen2019optimizing}.

\begin{lem}[Theorem 1.2 in \cite{gilyen2019optimizing}] \label{lem: lower bound in GAW}
    Let $\mathcal{F}$ be a finite set of functions $X \rightarrow \mathbb{R}$. Suppose we have access to $f \in \mathcal{F}$ via a phase oracle such that
    $$
    \mathrm{O}_f:|x\rangle \rightarrow e^{i f(x)}|x\rangle \quad \text { for every } x \in X .
    $$
    Let $f_* \in \mathcal{F}$ be fixed. Given $\mathrm{O}_f$ for some unknown $f \in \mathcal{F}$, determining whether $f=f_*$ has query complexity
    $$
    \Omega\left(\sqrt{|\mathcal{F}|-1} \Bigg/ \sqrt{\max\limits_{x \in X} \sum\limits_{f \in \mathcal{F}}\left|f(x)-f_*(x)\right|^2} \right) .
    $$
\end{lem}

To prove a lower bound in the situation of Theorem \ref{thm: hessian using difference formula}, we construct a family of functions $\mathcal{F}$ for which the functions can be well distinguished by calculating their Hessians with accuracy $\varepsilon$ in the matrix max norm. Then, with the help of Lemma \ref{lem: lower bound in GAW} we show that this requires $\Omega(d / \varepsilon)$ queries.

\begin{lem}\label{lem: function construction for lower bound}
    Let $d\in \mb N, \varepsilon, c \in \mb R_+$ and define the following functions that map $\mb R^d$ into $\mb R: f_*(\x):=0$ and $ f_{j,k}(\boldsymbol{x}):=\varepsilon x_j x_k e^{-c\|\boldsymbol{x}\|^2/2} $ for all $j,k \in [d]$. Consider the family of functions $\mathcal{F}:=\left\{f_{j,k}(\boldsymbol{x}):j,k\in [d]\right\}$, then $|\mc F|=d^2$ and for all $\x \in \mb R^d$ we have
    $$
    \sum_{j,k \in[d]}\left|f_{j,k}(\boldsymbol{x})-f_*(\boldsymbol{x})\right|^2 \leq \frac{4 \varepsilon^2}{e^2 c^2}.
    $$
\end{lem}
\begin{proof}
    $$
    \begin{aligned}
    \sum_{j,k \in[d]}\left|f_{j,k}(\boldsymbol{x})-f_*(\boldsymbol{x})\right|^2 & =\sum_{j,k \in[d]}\left| \varepsilon x_j x_k e^{-c\|\boldsymbol{x}\|^2 / 2}\right|^2 \\
    & =\varepsilon^2\|\boldsymbol{x}\|^4 e^{-c\|\boldsymbol{x}\|^2} \\
    & \leq \frac{4\varepsilon^2}{e^2 c^2}.
    \end{aligned}
    $$   
    The last inequality uses $ze^{-\sqrt{z}}\leq 4/e^2$ with $z:=c^2\|\x\|^4$.
\end{proof}

\begin{lem}\label{lem: smoothness of lower bound function}
    Let $d, t \in \mb N, 1\leq c\in \mb R_+$ and $\x\in \mb R^d$. Then the functions $f_{j,k}(\x):=c x_j x_k e^{-c\|\boldsymbol{x}\|^2/2}$ satisfy the following: for any $\alpha \in \mb N_0^d, |\alpha|=t$, 
    $$
    |\partial^\alpha f_{j,k}(\0)| \leq c^t t^{\frac{t}{2}},
    $$
    and $\H_{f_{j,k}}(\0)=c (E_{jk}+E_{kj})$, where $E_{jk}$ is the $(0,1)$-matrix of dimension $d$ such that only the $(j,k)$-th entry equals $1$.
\end{lem}

\begin{proof}
Observe that
$$
f(\boldsymbol{x})=\sqrt{c} x_j e^{-\frac{c x_j^2}{2}}\cdot \sqrt{c} x_k e^{-\frac{c x_k^2}{2}} \cdot\prod_{i \neq j,k}^d e^{-\frac{c x_i^2}{2}}.
$$
From the Taylor series $e^{-\frac{cx^2}{2}}=\sum_{\ell=0}^{\infty}\left(-\frac{1}{2}\right)^{\ell} \frac{(\sqrt{c} x)^{2 \ell}}{\ell!}$, we have that for $t\in \mb N$,
$$
\left.\partial_i^t e^{-\frac{c x_i^2}{2}}\right|_{x_i=0}=\left\{\begin{array}{ll}
\left(-\frac{1}{2}\right)^{\ell} (\sqrt{c})^{2 \ell} \frac{(2 \ell)!}{\ell!} & \quad t=2 \ell \\
0 &\quad  t=2 \ell+1
\end{array},\right.
$$
and
$$
\left.\quad \partial_j^t \sqrt{c} x_j e^{-\frac{c x_j^2}{2}}\right|_{x_j=0}=\left\{\begin{array}{ll}
0 &\quad t=2 \ell \\
\left(-\frac{1}{2}\right)^{\ell} (\sqrt{c})^{2 \ell+1} \frac{(2 \ell+1)!}{\ell!} & \quad t=2 \ell+1
\end{array} .\right.
$$
Also observe that, for $\ell \geq 0$, we have $(\sqrt{c})^{\ell}\leq c^{\ell}$, and
$$
\frac{(2 \ell)!}{\ell!} \leq(2 \ell)^{\ell} \quad \text { and } \quad\left(\frac{1}{2}\right)^{\ell} \frac{(2 \ell+1)!}{\ell!} \leq(2 \ell+1)^{\ell+1 / 2}.
$$
The statements of the lemma follow by combining the above results.
\end{proof}

\begin{prop}[Lower bound for sparse Hessian estimation of real-valued functions]
\label{prop: lower bound of real-valued Hessian estimation}
    Let $\varepsilon,c,d\in \mb R_+$ such that $\varepsilon\leq c$, and for an arbitrary finite set $G \subset \mathbb{R}^d$, let
    $$
    \mathcal{H}=\underset{\x \in G}{\operatorname{Span}}\{|\boldsymbol{x}\rangle: \boldsymbol{x} \in G\}.
    $$
    Suppose $\mathcal{A}$ is a $T$-query quantum algorithm, given query access to phase oracle $\mathrm{O}_f:|\boldsymbol{x}\rangle \rightarrow e^{i f(\boldsymbol{x})} |\boldsymbol{x}\rangle$, acting on $\mathcal{H}$, for analytic functions $f: \mathbb{R}^d \rightarrow \mathbb{R}$ satisfying
    $$
    \left|\partial^\alpha f(\mathbf{0})\right| \leq c^t t^{\frac{t}{2}} \quad \text { for all } t \in \mathbb{N}, \alpha \in \mb N_0^d, |\alpha|=t
    $$
    such that $\mathcal{A}$ computes an $\varepsilon$-approximation of the Hessian of $f$ at $\mathbf{0}$ under the matrix max-norm, succeeding with probability at least $2 / 3$, then $T>c d/ \varepsilon \in \Omega(d/\varepsilon)$.
\end{prop}
\begin{proof}
    We defined a set of functions $\mathcal{F}=\left\{f_{j,k}(\boldsymbol{x}):j,k \in[d]\right\}$ and $f_*=0$ as shown in Lemma \ref{lem: function construction for lower bound}. By Lemma \ref{lem: smoothness of lower bound function}, every $f\in \mc F$ satisfies $|\partial^\alpha f(\0)| \leq c^t t^{\frac{t}{2}}$ for all $t\in \mb N$ and $\alpha \in \mb N_0^d, |\alpha|=t$. Suppose we are given the phase oracle $O_f=O_{f_{j,k}}$ for some unknown $j,k\in [d]$ or $O_f=O_{f_*}$. Since $\H_{f_*}(\0)=\0$ and $\H_{f_{j,k}}(\0)=2\varepsilon (E_{jk}+E_{kj})$, using $T$-query algorithm $\mc A$, one can determine which function the given oracle corresponds to with success probability at least 2/3. In particular, we can distinguish the case $f=f_*$ from $f\in \mc F$, and thus by Lemmas \ref{lem: lower bound in GAW} and \ref{lem: function construction for lower bound} we get that
    $$
    T \geq \frac{\sqrt{d^2}}{2\varepsilon/ ec}   > \frac{c d}{\varepsilon}.
    $$
    This completes the proof.
\end{proof}

\section{Conclusions}

In this work, we proposed a new quantum algorithm for gradient estimation for analytic functions that can take values from the complex field. The algorithm achieves exponential speedup over classical algorithms. We also generalized our algorithm and Gily\'{e}n-Arunachalam-Wiebe's algorithm for gradient estimation to Hessian estimation, both achieving polynomial speedups over classical algorithms. A lower bound shows that polynomial speedup is the best we can expect. 
For the estimation of sparse Hessians, we proposed two new quantum algorithms with better performance. Exponential quantum speedup is achieved for the quantum algorithm based on the spectral method.
After this research, we feel there are some interesting questions that deserve further study:

\begin{itemize}
    \item Can we find some applications of our algorithms for gradient estimation and Hessian estimation? 
    Jordan's algorithm and Gily\'{e}n-Arunachalam-Wiebe's algorithm for gradient estimation have already been used to speed up some practical problems \cite{chakrabarti2020quantum,van2020convex,van2023quantum,huggins2022nearly,jerbi2023quantum}. Similar to Jordan's algorithm, our algorithm also achieves exponential speedup, so it is interesting to find some useful applications for it. The Hessian matrix widely appears in second-order optimization algorithms, such as Newton's method, so it is also interesting to know if our quantum algorithms for Hessian estimation can be used to speed up some problems using Newton's method. Sparse Hessian is very common in optimization \cite{fletcher1997computing,d2638b8e-8fdc-3d6d-a044-1384db1d1f3a,coleman1984estimation,coleman1984large,gebremedhin2005color}. It is interesting to see if our quantum algorithms for sparse Hessian estimation can be used to accelerate some classical optimization problems.
    
    \item What is the right bound for Hessian estimation? 
    In Section \ref{section: Lower bounds analysis}, we obtained a lower bound of $\widetilde{\Omega}(d)$ in two settings. For complex-valued functions, our algorithm is optimal up to a logarithmic factor. However, for real-valued functions, there is still a gap. We are not able to close this gap in the current research. 
    
    \item Can the finite difference method perform better for Hessian estimation?
    The algorithms based on the spectral method have better performance; however, the oracle assumption is stronger than that in the finite difference method. If better quantum algorithms can be found using the finite difference method, then this should be very helpful for practical applications.
\end{itemize}

\section*{Acknowledgement}

The research was supported by the National Key Research Project of China under Grant No. 2023YFA1009403, and No. 2020YFA0712300.
We would like to thank Shantanav Chakraborty and Tongyang Li for their helpful comments on an earlier version of this paper.

\begin{appendices}
\section{Error analysis of finite difference formula}
\label{appendix: Error analysis of finite difference formula}

\begin{lem} \label{lem: bound of coefficients in fin diff formula}
    For all $m \in \mb N$ and $k \geq 2m$, we have
    $$
    \sum_{t=-m}^m\left|a_{t}^{(2 m)} t^{k+1}\right| \leq 24 e^{-\frac{7 m}{6}} m^{k+\frac{3}{2}},
    $$
    where $a_{t}^{(2 m)}$ is defined in Definition \ref{def: finite difference formula of Hessian}.
\end{lem}
\begin{proof}
    $$
    \begin{aligned}
        \sum_{t=-m}^m\left|a_{t}^{(2 m)} t^{k+1}\right| &\leq 4 \sum_{t=1}^m \frac{2}{t^2} \frac{m! m!}{(m+t)! (m-t)!} t^{k+1} \\
        & \leq 8m \cdot \max_{t\in [m]} \frac{m! m!}{(m+t)! (m-t)!} t^{k-1}.
    \end{aligned}
    $$
    Using a similar argument to the proof of \cite[Lemma 34 (the arXiv version)]{gilyen2019optimizing}, we obtain
    $$
    \frac{m! m!}{(m+t)! (m-t)!} t^{k-1} \leq 3 \sqrt{m} e^{-\frac{7}{6}m} m^k.
    $$
    This completes the proof.
\end{proof}

\begin{lem} \label{lem: bound of error for 1 dim fin diff formula}
    Let $x \in \mb R_+, m \in \mb N$ and suppose that $f: [-mx,mx] \rightarrow \mb R$ is $(2m+1)$-times differentiable. Then there exists some constant $c\in \mb R$ such that
    $$
    \begin{aligned}
        \left|f''(0) x^2 - f_{(2m)}(x)-c\right| &=\Big|f''(0) x^2 -\sum_{t=-m}^m a_t^{(2m)} f(tx)\Big| \\
        & \leq 4 e^{-\frac{m}{2}}\left\|f^{(2 m+1)}\right\|_{\infty}|x|^{2 m+1} 
    \end{aligned}
    $$
    where $\left\|f^{(2 m+1)}\right\|_{\infty}:=\sup \limits_{\xi \in[-m x, m x]}\left|f^{(2 m+1)}(\xi)\right|$ and $a_{t}^{(2 m)}$ is defined in Definition \ref{def: finite difference formula of Hessian}. Moreover
    $$
    \sum_{t=1}^m\big|a_{t}^{(2 m)}\big|<\sum_{t=1}^m \frac{1}{t^2} < \frac{\pi^2}{6}.
    $$
\end{lem}
\begin{proof}
    For any $y\in \mb R$, Taylor's theorem gives us
    $$
    f(y)=\sum_{j=0}^{2 m} \frac{f^{(j)}(0)}{j !} y^j+\frac{f^{(2 m+1)}(\xi)}{(2 m+1) !} y^{2 m+1}
    $$
    for some $\xi$ between $0$ and $y$. Denote the first term on the right-hand side of the above equation as $p(y/x)$, and let $z:=y/x$. Then we have
    \begin{equation}\label{eq: def of p}
        p(z)=\sum_{j=0}^{2 m} \frac{f^{(j)}(0)}{j !}(z x)^j=f(z x)-\frac{f^{(2 m+1)}(\xi)}{(2 m+1) !}(z x)^{2 m+1}.
    \end{equation}
    Note that $p''(0)=f''(0) x^2$. 
    Since $\deg(p)\leq 2m$, the Lagrange interpolation formula leads to
    $$
    p(z)=\sum_{t=-m}^m p(t) \prod_{\substack{i=-m \\ i \neq t}}^m \frac{z-i}{t-i} .
    $$
    Now we aim to compute the 2nd-order derivative of the above function. Note that
    $$
    \prod_{\substack{i=-m \\ i \neq t}}^m \frac{1}{t-i} = \frac{(-1)^{m-t}}{(m+t)! (m-t)!}.
    $$
    We below handle the cases for $t \neq 0$ and $t=0$ respectively:
    \begin{itemize}
        \item For $t \neq 0$, we have
        $$\prod_{\substack{i=-m \\ i \neq t}}^m (z-i)=\prod_{\substack{i=1 \\ i \neq |t|}}^m (z^2-i^2)\cdot z(z+t).$$
        Let $h(z):=\prod\limits_{\substack{i=1 \\ i \neq |t|}}^m (z^2-i^2)$ and $g(z):=z(z+t)$. Note that $h(z)$ is a polynomial with only even powers, which can be written as $h(z)=z^{2m-2}+\cdots+c_1 z^2 +c_0$. It is easy to see that $g(0)=0, f'(0)=0, g''(0)=2$, and
        $$
        f(0)=c_0=\prod_{\substack{i=1 \\ i \neq |t|}}^m (-i^2)=(-1)^{m-1}\cdot \frac{(m!)^2}{t^2}.
        $$
        Then $[h(z)g(z)]''=h''(z)g(z)+2h'(z)g'(z)+h(z)g''(z)$ leads to
        $$
        \Big(\prod_{\substack{i=-m \\ i \neq t}}^m (z-i)\Big)''\Big|_{z=0}= \frac{2\cdot (-1)^{m-1} (m!)^2}{t^2}.
        $$

        \item For $t=0$, we have
        $$
        \prod_{\substack{i=-m \\ i \neq t}}^m (z-i)=\prod_{i=1}^m (z^2-i^2)=z^{2m}+\cdots+d_1 z^2 +d_0.
        $$
        Its 2nd-derivative at $z=0$ is just
        $$
        2d_1=2\sum_{j=1}^m \prod_{\substack{i=1\\i \neq j}}^m (-i^2)=2\cdot (-1)^{m-1}(m!)^2 \cdot \sum_{j=1}^m \frac{1}{j^2}.
        $$
    \end{itemize}
    Combining the above results, we obtain that
    \begin{equation}\label{eq: 2nd derivative of p}
        \begin{aligned}
        p''(0)&= 2\cdot (-1)^{t-1} \sum_{j=1}^m \frac{1}{j^2}\cdot p(0)+\sum_{\substack{t=-m\\t \neq 0}}^m \frac{2\cdot (-1)^{t-1}}{t^2} \frac{(m!)^2}{(m+t)!(m-t)^2}\cdot p(t)\\
        &=2\cdot (-1)^{t-1} \sum_{j=1}^m \frac{1}{j^2}\cdot p(0)+\sum_{\substack{t=-m\\t \neq 0}}^m a_t^{(2m)} p(t),
    \end{aligned}
    \end{equation}
    where $a_t^{(2m)}$ is the coefficients given in Definition \ref{def: finite difference formula of Hessian}.
    Therefore,
    $$
    \begin{aligned}
        f''(0)x^2&=p''(0) \\
        &\stackrel{\eqref{eq: 2nd derivative of p}}{=}2\cdot (-1)^{t-1} \sum_{j=1}^m \frac{1}{j^2}\cdot p(0)+\sum_{\substack{t=-m\\t \neq 0}}^m a_t^{(2m)} p(t) \\
        &\stackrel{\eqref{eq: def of p}}{=}2\cdot (-1)^{t-1} \sum_{j=1}^m \frac{1}{j^2}\cdot f(0)+\sum_{\substack{t=-m\\t \neq 0}}^m a_t^{(2m)} \left(f(tx)-\frac{f^{(2 m+1)}(\xi_t)}{(2 m+1) !}(t x)^{2 m+1}\right).
    \end{aligned}
    $$

    Let $c:=\big(a_0^{(2m)}-2\cdot (-1)^{t-1} \sum_{j=1}^m \frac{1}{j^2}\big) f(0)$, which is a constant since $|c|\leq \big( \big|a_0^{(2m)}\big|+2\sum_{j=1}^m \frac{1}{j^2} \big) f(0)\leq 4 \sum_{j=1}^m \frac{1}{j^2} f(0)\leq \frac{2 \pi^2}{3} f(0).$ Hence, 
    $$
    \begin{aligned}
        \left|f_{(2m)}(x)-f''(0)x^2-c\right| &=\Bigg| \sum_{\substack{t=-m\\t \neq 0}}^m a_t^{(2m)} \frac{f^{(2 m+1)}(\xi_t)}{(2 m+1) !}(t x)^{2 m+1} \Bigg| \\
        & \leq \sum_{\substack{t=-m\\t \neq 0}}^m \left| a_t^{(2m)} t^{2m+1}\right| \cdot \frac{\left\|f^{(2 m+1)}\right\|_{\infty}}{(2 m+1) !} \cdot |x|^{2m+1} \\
        & \leq 24 e^{-\frac{7 m}{6}} m^{2m+3 / 2} \frac{\left\|f^{(2 m+1)}\right\|_{\infty}}{(2 m+1) !} |x|^{2m+1} \\
        & \leq 12 e^{-\frac{7 m}{6}} m^{2m+1 / 2} \frac{\left\|f^{(2 m+1)}\right\|_{\infty}}{\sqrt{4 \pi m}(2 m / e)^{2 m}}|x|^{2 m+1} \\
        & \leq 4 e^{-\frac{7 m}{6}}\left(\frac{e}{2}\right)^{2 m}\left\|f^{(2 m+1)}\right\|_{\infty}|x|^{2 m+1} \\
        & \leq 4 e^{-\frac{m}{2}}\left\|f^{(2 m+1)}\right\|_{\infty}|x|^{2 m+1},
    \end{aligned}
    $$
    which is as claimed.
\end{proof}

Based on Lemma \ref{lem: bound of error for 1 dim fin diff formula}, we now prove a version for higher dimensional functions.

\begin{lem}
\label{app: lemma 3}
    Let $m \in \mathbb{N}, B>0, \boldsymbol{x} \in \mathbb{R}^d$. 
    Assume that $f:\left[-m\|\boldsymbol{x}\|_{\infty}, m\|\boldsymbol{x}\|_{\infty}\right]^d$ $\rightarrow \mathbb{R}$ is $(2 m+1)$-times differentiable and
    $$
    \left|\partial_{\boldsymbol{r}}^{2 m+1} f(\tau \boldsymbol{x})\right| \leq B \quad \text { for all } \tau \in[-m, m], \text{ where } \boldsymbol{r}:=\frac{\boldsymbol{x} }{\|\boldsymbol{x}\|},
    $$
    then there exists some constant $c\in \mb R$ such that
    $$
    \left|f_{(2 m)}(\boldsymbol{x})-\boldsymbol{x}^{\mathrm{T}} \mathbf{H}_f(\mathbf{0}) \boldsymbol{x}-c\right| \leq 4B e^{-\frac{m}{2}}\|\boldsymbol{x}\|^{2 m+1} .
    $$
\end{lem}

\begin{proof} 
    Consider the function $h(\tau):=f(\tau \x)$, then
    $$
    \begin{aligned}
        \left|f_{(2m)}(\x)-\boldsymbol{x}^{\mathrm{T}} \mathbf{H}_f(\mathbf{0}) \boldsymbol{x}-c\right| &=\left|h_{(2m)}(1)-h''(0)\right| \\
        &\leq 4 e^{-\frac{m}{2}}\sup _{\tau \in[-m, m]}\left|h^{(2 m+1)}(\tau)\right| \\
        & = 4 e^{-\frac{m}{2}}\sup _{\tau \in[-m, m]}\left|\partial_{\boldsymbol{x}}^{2 m+1} f(\tau \boldsymbol{x})\right| \\
        & = 4 e^{-\frac{m}{2}} \sup _{\tau \in[-m, m]}\left|\partial_{\boldsymbol{r}}^{2 m+1} f(\tau \boldsymbol{x})\right|\|\boldsymbol{x}\|^{2 m+1} \\
        & \leq 4B e^{-\frac{m}{2}}\|\boldsymbol{x}\|^{2 m+1} ,
    \end{aligned}
    $$
    as claimed.
\end{proof}

\begin{lem}
\label{app: lemma 4}
Under the same assumption as in  Lemma \ref{app: lemma 3}, let $\x\in aG_n^d, \y \in \mathbb{R}^d$ with $\|\y\|=1$. If $m \geq \log(dB/\varepsilon)$ and $a\approx O(1/\sqrt{d}m)$, we then have
\be
\label{app: diff of Hessian}
\left| \frac{1}{2} \Big( f_{(2 m)}(\x+\y) - f_{(2 m)}(\x) \Big) -
\x^T \mathbf{H}_f(\mathbf{0}) \y - \frac{1}{2} \y^T \mathbf{H}_f(\mathbf{0}) \y \right|
\leq O(a\varepsilon).
\ee
\end{lem}

\begin{proof}
By Lemma \ref{app: lemma 3}, for any $\y$, we have
\[
\text{LHS of } \eqref{app: diff of Hessian} \leq 4B e^{-\frac{m}{2}} (\|\boldsymbol{x}\|^{2 m+1} + \|\x+\boldsymbol{y}\|^{2 m+1}).
\]
When $\x\in a G_n^d$ and $\y\in \mathbb{R}^d$ with $\|\y\|=1$, we have
\beas
\text{LHS of } \eqref{app: diff of Hessian}
&\leq& 4B e^{-\frac{m}{2}} \left( \left(\frac{a\sqrt{d}}{2}\right)^{2 m+1} + \left(\frac{a}{2}\sqrt{d}+1\right)^{2 m+1} \right) \\
&\leq& 8B e^{-\frac{m}{2}} \left(\frac{a}{2}\sqrt{d}+1\right)^{2 m+1}
\eeas
It is hoped that the above is smaller than $\eta a \varepsilon$ with $\eta=1/8\cdot 42 \cdot \pi$. We assume that $a\sqrt{d}/2 <1 $ and $(2m+1)  \frac{a\sqrt{d}}{2} <1$, then by $\ln(1+x) \leq x$ and $e^x \leq 1+2x$ when $0<x<1$, we have
\beas
8B e^{-\frac{m}{2}}  \left(\frac{a}{2}\sqrt{d} + 1\right)^{2 m+1}
&=& 8B e^{-\frac{m}{2}} 
e^{(2m+1) \ln(1+\frac{a\sqrt{d} }{2})} \\
&\leq& 8B e^{-\frac{m}{2}} e^{(2m+1)  \frac{a\sqrt{d} }{2}} \\
&\leq& 8B e^{-\frac{m}{2}} (1 + (2m+1)  a\sqrt{d} ).
\eeas
Let 
\[
8B e^{-\frac{m}{2}}  (1 + (2m+1)  a\sqrt{d} ) \leq \eta a \varepsilon,
\]
i.e.,
\[
a(\eta \varepsilon - 8B e^{-\frac{m}{2}}  (2m+1) \sqrt{d} )
\geq 8B e^{-\frac{m}{2}}.
\]
So we obtain
\beas
8B e^{-\frac{m}{2}} \leq 
(\eta \varepsilon - 8B e^{-\frac{m}{2}}  (2m+1) \sqrt{d} ) \frac{2}{\sqrt{d}(2m+1)}.
\eeas
This means
\beas
12B e^{-\frac{m}{2}} \sqrt{d}(2m+1) \leq 
\eta \varepsilon .
\eeas
Thus when $m$ is large enough, say $m\approx \log(dB/\varepsilon)$, the above holds. As a result, we can choose $a<2/\sqrt{d}(2m+1)\approx 1/\sqrt{d}\log(d/\varepsilon)$. 
\end{proof}

\begin{lem}\label{lem: bound of error in fin diff formula}
    Let $m\in \mb N, a,c\in \mb R_+$. Suppose $f:\mb R^d\rightarrow \mb R$ is analytic and for all $k\in\mb N$, $\alpha \in \mb N_0^d, |\alpha|=k$
    $$
    \left|\partial^\alpha f(\0)\right| \leq c^k k^{\frac{k}{2}},
    $$
    then we have
    $$
    \left| f_{2m}(\boldsymbol{x})-\boldsymbol{x}^{\mathrm{T}} \H_f(\0) \boldsymbol{x}\right| \leq \sum_{k=2 m+1}^{\infty}(13 a c m \sqrt{d})^k
    $$
    for all but $1/1000$ fraction of points $\boldsymbol{x}\in a G_n^d$.
\end{lem}
\begin{proof}
    Since $f$ is analytic, we can write it as its Taylor series:
    $$
    f(\boldsymbol{x})=\sum_{k=0}^{\infty} \sum_{\alpha \in[d]^k} \frac{\boldsymbol{x}^\alpha \cdot \partial_\alpha f(\mathbf{0})}{k !} .
    $$
    In this proof, we use the notation $\partial_\alpha f:=\partial_{\alpha_1} \partial_{\alpha_2} \cdots \partial_{\alpha_k} f$ for $\alpha \in [d]^k$ to make the expression clearer and more concise.
    Using the finite difference formula defined in Definition \ref{def: finite difference formula of Hessian}, we have
    $$
    \begin{aligned}
    f_{(2 m)}(\boldsymbol{x}) & =\sum_{t=-m}^m a_{t}^{(2 m)} f(t\boldsymbol{x}) \\
    & =\sum_{t=-m}^m a_{t}^{(2 m)} \sum_{k=0}^{\infty} \frac{1}{k !} \sum_{\alpha \in[d]^k}(t \boldsymbol{x})^\alpha \cdot \partial_\alpha f(\mathbf{0}) \\
    & =\sum_{k=0}^{\infty} \frac{1}{k !} \sum_{\alpha \in[d]^k} \boldsymbol{x}^\alpha \cdot \partial_\alpha f(\mathbf{0}) \underbrace{\sum_{t=-m}^m a_{t}^{(2 m)} t^k}_* .
    \end{aligned}
    $$
    Applying Lemma \ref{lem: bound of error for 1 dim fin diff formula} to the function $t^k$ with the choice $x=1$ tells us that the term $*$ is $0$ if $k\leq 2m$ except for $k=2$, in which case it is $2$. Therefore, using Lemma \ref{lem: bound of coefficients in fin diff formula}
    $$
    \begin{aligned}
        \left| \boldsymbol{x}^{\mathrm{T}} \H_f(\0) \boldsymbol{x}-f_{(2m)}(\boldsymbol{x})\right| &=\Big|\sum_{k=2 m+1}^{\infty} \frac{1}{k !} \sum_{\alpha \in[d]^k} \boldsymbol{x}^\alpha \cdot \partial_\alpha f(\mathbf{0}) \sum_{t=-m}^m a_{t}^{(2 m)} t^k\Big| \\
        & \leq \sum_{k=2 m+1}^{\infty}\left(\frac{e}{k}\right)^k \frac{1}{\sqrt{4 \pi m}}\Big|\sum_{\alpha \in[d]^k} \boldsymbol{x}^\alpha \cdot \partial_\alpha f(\0)\Big|\Big|\sum_{t=-m}^m a_{t}^{(2 m)} t^k\Big| \\
        & \leq \sum_{k=2 m+1}^{\infty}\Big|\sum_{\alpha \in[d]^k} \boldsymbol{x}^\alpha \cdot \partial_\alpha f(\mathbf{0})\Big|\left(\frac{e}{k}\right)^k \frac{12 e^{-\frac{7 m}{6}} m^{k+\frac{1}{2}}}{\sqrt{\pi m}} \\
        & \leq \sum_{k=2 m+1}^{\infty} 2\sqrt{2} \Big|\sum_{\alpha \in[d]^k} \boldsymbol{x}^\alpha \cdot \partial_\alpha f(\mathbf{0})\Big| \left(\frac{e m}{k}\right)^k .
    \end{aligned}
    $$
    Some probabilistic results \cite[Lemma 36 in the arXiv version]{gilyen2019optimizing} allow us to obtain that, if we choose uniformly random $\boldsymbol{x}\in aG_n^d $, then for all $k\in \mb N_+$, the ratio of $\boldsymbol{x}$ for which
    $$
    \Big|\sum_{\alpha \in[d]^k} \frac{\boldsymbol{x}^\alpha}{a^k} \cdot \frac{\partial_\alpha f(\mathbf{0})}{c^k k^{\frac{k}{2}}}\Big| \geq \sqrt{2}\Big(4 \sqrt{\frac{d k}{2}}\Big)^k
    $$
    is at most $4^{-2k}$. Since $\sum_{k=2 m+1}^{\infty} 4^{-2 k} \leq \sum_{k=3}^{\infty} 4^{-2 k}<1 / 1000$, it follows that for all but a $1/1000$ fraction of $\boldsymbol{x}\in aG_n^d$, we have
    $$
    \begin{aligned}
        \left| \boldsymbol{x}^{\mathrm{T}} \H_f(\0) \boldsymbol{x}-f_{(2m)}(\boldsymbol{x})\right| & \leq \sum_{k=2 m+1}^{\infty} 2\sqrt{2} \Big|\sum_{\alpha \in[d]^k} \boldsymbol{x}^\alpha \cdot \partial_\alpha f(\mathbf{0})\Big| \left(\frac{e m}{k}\right)^k \\
        & \leq \sum_{k=2 m+1}^{\infty} 4 \Big(4 \sqrt{\frac{d k}{2}}\Big)^k a^k c^k k^{\frac{k}{2}} \left(\frac{e m}{k}\right)^k \\
        & =\sum_{k=2 m+1}^{\infty}4\Big(\frac{4 \sqrt{d} acem }{\sqrt{2}}\Big)^k \\
        & <\sum_{k=2 m+1}^{\infty}(13 a c m \sqrt{d})^k .
    \end{aligned}
    $$
    This completes the proof.
\end{proof}

\end{appendices}

\bibliographystyle{plain}
\bibliography{reference}

@inproceedings{gilyen2019optimizing,
  title={Optimizing quantum optimization algorithms via faster quantum gradient computation},
  author={Gily{\'e}n, Andr{\'a}s and Arunachalam, Srinivasan and Wiebe, Nathan},
  booktitle={Proceedings of the 30th Annual ACM-SIAM Symposium on Discrete Algorithms},
  pages={1425--1444},
  year={2019},
  organization={SIAM}
}

@book{titchmarsh1939theory,
  title={{The Theory of Functions}},
  author={Titchmarsh, Edward C},
  year={1939},
  publisher={Oxford University Press}
}

@inproceedings{lee2021quantum,
  title={Quantum algorithms for graph problems with cut queries},
  author={Lee, Troy and Santha, Miklos and Zhang, Shengyu},
  booktitle={Proceedings of the 32nd Annual ACM-SIAM Symposium on Discrete Algorithms},
  pages={939--958},
  year={2021},
  organization={SIAM}
}

@article{karam2020complex,
  title={Why are complex numbers needed in quantum mechanics? {S}ome answers for the introductory level},
  author={Karam, Ricardo},
  journal={American Journal of Physics},
  volume={88},
  number={1},
  pages={39--45},
  year={2020},
  publisher={AIP Publishing}
}

@article{li2005general,
  title={General explicit difference formulas for numerical differentiation},
  author={Li, Jianping},
  journal={Journal of Computational and Applied Mathematics},
  volume={183},
  number={1},
  pages={29--52},
  year={2005},
  publisher={Elsevier}
}

@article{ruiz2017quantum,
  title={{Quantum arithmetic with the quantum Fourier transform}},
  author={Ruiz-Perez, Lidia and Garcia-Escartin, Juan Carlos},
  journal={Quantum Information Processing},
  volume={16},
  pages={1--14},
  year={2017},
  publisher={Springer}
}

@book{krantz2002primer,
  title={{A Primer of Real Analytic Functions}},
  author={Krantz, Steven G. and Parks, Harold R.},
  year={2002},
  publisher={Springer Science \& Business Media}
}

@article{fornberg1981numerical,
  title={Numerical differentiation of analytic functions},
  author={Fornberg, Bengt},
  journal={ACM Transactions on Mathematical Software},
  volume={7},
  number={4},
  pages={512--526},
  year={1981},
  publisher={ACM New York, NY, USA}
}

@book{trefethen2000spectral,
  title={{Spectral Methods in MATLAB}},
  author={Trefethen, Lloyd N.},
  year={2000},
  publisher={SIAM}
}

@article{huggins2022nearly,
  title={Nearly optimal quantum algorithm for estimating multiple expectation values},
  author={Huggins, William J and Wan, Kianna and McClean, Jarrod and O’Brien, Thomas E and Wiebe, Nathan and Babbush, Ryan},
  journal={Physical Review Letters},
  volume={129},
  number={24},
  pages={240501},
  year={2022},
  publisher={APS}
}

@inproceedings{gilyen2019quantum,
  title={Quantum singular value transformation and beyond: exponential improvements for quantum matrix arithmetics},
  author={Gily{\'e}n, Andr{\'a}s and Su, Yuan and Low, Guang Hao and Wiebe, Nathan},
  booktitle={Proceedings of the 51st Annual ACM SIGACT Symposium on Theory of Computing},
  pages={193--204},
  year={2019}
}

@article{lyness1971algorithm,
  title={{Algorithm 413: ENTCAF and ENTCRE: evaluation of normalized Taylor coefficients of an analytic function}},
  author={Lyness, James N and Sande, G},
  journal={Communications of the ACM},
  volume={14},
  number={10},
  pages={669--675},
  year={1971},
  publisher={ACM New York, NY, USA}
}

@article{lyness1967numerical,
  title={Numerical differentiation of analytic functions},
  author={Lyness, James N and Moler, Cleve B},
  journal={SIAM Journal on Numerical Analysis},
  volume={4},
  number={2},
  pages={202--210},
  year={1967},
  publisher={SIAM}
}

@article{lyness1968differentiation,
  title={Differentiation formulas for analytic functions},
  author={Lyness, J.N.},
  journal={Mathematics of Computation},
  volume={22},
  number={102},
  pages={352--362},
  year={1968}
}

@article{cornelissen2019quantum,
  title={{Quantum gradient estimation of Gevrey functions}},
  author={Cornelissen, Arjan},
  journal={arXiv:1909.13528},
  year={2019}
}

@book{coleman1984large,
  title={{Large Sparse Numerical Optimization}},
  author={Coleman, Thomas F},
  year={1984},
  address={New York},
  publisher={Springer-Verlag}
}

@article{gebremedhin2005color,
  title={{What color is your Jacobian? Graph coloring for computing derivatives}},
  author={Gebremedhin, Assefaw Hadish and Manne, Fredrik and Pothen, Alex},
  journal={SIAM Review},
  volume={47},
  number={4},
  pages={629--705},
  year={2005},
  publisher={SIAM}
}

@article{coleman1984estimation,
  title={{Estimation of sparse Hessian matrices and graph coloring problems}},
  author={Coleman, Thomas F and Mor{\'e}, Jorge J},
  journal={Mathematical Programming},
  volume={28},
  pages={243--270},
  year={1984},
  publisher={Springer}
}

@article{d2638b8e-8fdc-3d6d-a044-1384db1d1f3a,
 ISSN = {00361429},
 URL = {http://www.jstor.org/stable/2156656},
 author = {M. J. D. Powell and Ph. L. Toint},
 journal = {SIAM Journal on Numerical Analysis},
 number = {6},
 pages = {1060--1074},
 publisher = {Society for Industrial and Applied Mathematics},
 title = {On the Estimation of Sparse {H}essian Matrices},
 urldate = {2024-06-17},
 volume = {16},
 year = {1979}
}

@article{fletcher1997computing,
  title={{Computing sparse Hessian and Jacobian approximations with optimal hereditary properties}},
  author={Fletcher, Roger and Grothey, Andreas and Leyffer, Sven},
  journal={Large-Scale Optimization with Applications: Part II: Optimal Design and Control},
  pages={37--52},
  year={1997},
  publisher={Springer}
}

@inproceedings{bernstein1993quantum,
  title={Quantum complexity theory},
  author={Bernstein, Ethan and Vazirani, Umesh},
  booktitle={Proceedings of the 25th Annual ACM SIGACT Symposium on Theory of Computing},
  pages={11--20},
  year={1993}
}

@article{jordan2005fast,
  title={Fast quantum algorithm for numerical gradient estimation},
  author={Jordan, Stephen P},
  journal={Physical Review Letters},
  volume={95},
  number={5},
  pages={050501},
  year={2005},
  publisher={APS}
}

@article{rebentrost2019quantum,
  title={{Quantum gradient descent and Newton’s method for constrained polynomial optimization}},
  author={Rebentrost, Patrick and Schuld, Maria and Wossnig, Leonard and Petruccione, Francesco and Lloyd, Seth},
  journal={New Journal of Physics},
  volume={21},
  number={7},
  pages={073023},
  year={2019},
  publisher={IOP Publishing}
}

@article{gao2021quantum,
  title={Quantum gradient algorithm for general polynomials},
  author={Gao, Pan and Li, Keren and Wei, Shijie and Gao, Jiancun and Long, Guilu},
  journal={Physical Review A},
  volume={103},
  number={4},
  pages={042403},
  year={2021},
  publisher={APS}
}

@article{kerenidis2020quantum,
  title={Quantum gradient descent for linear systems and least squares},
  author={Kerenidis, Iordanis and Prakash, Anupam},
  journal={Physical Review A},
  volume={101},
  number={2},
  pages={022316},
  year={2020},
  publisher={APS}
}

@article{teo2023optimized,
  title={{Optimized numerical gradient and Hessian estimation for variational quantum algorithms}},
  author={Teo, Y.S.},
  journal={Physical Review A},
  volume={107},
  number={4},
  pages={042421},
  year={2023},
  publisher={APS}
}

@article{apers2023quantum,
  title={Quantum speedups for linear programming via interior point methods},
  author={Apers, Simon and Gribling, Sander},
  journal={SIAM Journal on Computing},
  volume={55},
  number={1},
  pages={93--134},
  year={2026},
  publisher={SIAM}
}

@article{low2019hamiltonian,
  title={Hamiltonian simulation by qubitization},
  author={Low, Guang Hao and Chuang, Isaac L},
  journal={Quantum},
  volume={3},
  pages={163},
  year={2019},
  publisher={Verein zur F{\"o}rderung des Open Access Publizierens in den Quantenwissenschaften}
}

@article{mari2021estimating,
  title={Estimating the gradient and higher-order derivatives on quantum hardware},
  author={Mari, Andrea and Bromley, Thomas R and Killoran, Nathan},
  journal={Physical Review A},
  volume={103},
  number={1},
  pages={012405},
  year={2021},
  publisher={APS}
}

@article{schuld2019evaluating,
  title={Evaluating analytic gradients on quantum hardware},
  author={Schuld, Maria and Bergholm, Ville and Gogolin, Christian and Izaac, Josh and Killoran, Nathan},
  journal={Physical Review A},
  volume={99},
  number={3},
  pages={032331},
  year={2019},
  publisher={APS}
}

@book{gunning2022analytic,
  title={{Analytic Functions of Several Complex Variables}},
  author={Gunning, Robert C and Rossi, Hugo},
  volume={368},
  year={2022},
  publisher={American Mathematical Society}
}

@book{range1998holomorphic,
  title={{Holomorphic Functions and Integral Representations in Several Complex Variables}},
  author={Range, R Michael},
  volume={108},
  year={1998},
  publisher={Springer Science \& Business Media}
}

@book{krantz2001function,
  title={{Function Theory of Several Complex Variables}},
  author={Krantz, Steven George},
  volume={340},
  year={2001},
  publisher={American Mathematical Society}
}

@book{boyd2004convex,
  title={{Convex Optimization}},
  author={Boyd, Stephen P and Vandenberghe, Lieven},
  year={2004},
  publisher={Cambridge University Press}
}

@book{sra2012optimization,
  title={{Optimization for Machine Learning}},
  author={Sra, Suvrit and Nowozin, Sebastian and Wright, Stephen J},
  year={2012},
  publisher={MIT Press}
}

@article{chakrabarti2020quantum,
  title={Quantum algorithms and lower bounds for convex optimization},
  author={Chakrabarti, Shouvanik and Childs, Andrew M and Li, Tongyang and Wu, Xiaodi},
  journal={Quantum},
  volume={4},
  pages={221},
  year={2020},
  publisher={Verein zur F{\"o}rderung des Open Access Publizierens in den Quantenwissenschaften}
}

@article{van2020convex,
  title={Convex optimization using quantum oracles},
  author={van Apeldoorn, Joran and Gily{\'e}n, Andr{\'a}s and Gribling, Sander and de Wolf, Ronald},
  journal={Quantum},
  volume={4},
  pages={220},
  year={2020},
  publisher={Verein zur F{\"o}rderung des Open Access Publizierens in den Quantenwissenschaften}
}

@inproceedings{van2023quantum,
  title={Quantum tomography using state-preparation unitaries},
  author={van Apeldoorn, Joran and Cornelissen, Arjan and Gily{\'e}n, Andr{\'a}s and Nannicini, Giacomo},
  booktitle={Proceedings of the 34th Annual ACM-SIAM Symposium on Discrete Algorithms},
  pages={1265--1318},
  year={2023},
  organization={SIAM}
}

@InProceedings{jerbi2023quantum,
  author =	{Jerbi, Sofiene and Cornelissen, Arjan and Ozols, Maris and Dunjko, Vedran},
  title =	{{Quantum policy gradient algorithms}},
  booktitle =	{18th Conference on the Theory of Quantum Computation, Communication and Cryptography (TQC 2023)},
  pages =	{13:1--13:24},
  series =	{Leibniz International Proceedings in Informatics (LIPIcs)},
  ISBN =	{978-3-95977-283-9},
  ISSN =	{1868-8969},
  year =	{2023},
  volume =	{266},
  editor =	{Fawzi, Omar and Walter, Michael},
  publisher =	{Schloss Dagstuhl -- Leibniz-Zentrum f{\"u}r Informatik},
  address =	{Dagstuhl, Germany},
  URL =		{https://drops-dev.dagstuhl.de/entities/document/10.4230/LIPIcs.TQC.2023.13},
  URN =		{urn:nbn:de:0030-drops-183230},
  doi =		{10.4230/LIPIcs.TQC.2023.13}
}

@article{berry2014high,
  title={High-order quantum algorithm for solving linear differential equations},
  author={Berry, Dominic W},
  journal={Journal of Physics A: Mathematical and Theoretical},
  volume={47},
  number={10},
  pages={105301},
  year={2014},
  publisher={IOP Publishing}
}

@article{childs2020quantum,
  title={Quantum spectral methods for differential equations},
  author={Childs, Andrew M and Liu, Jin-Peng},
  journal={Communications in Mathematical Physics},
  volume={375},
  number={2},
  pages={1427--1457},
  year={2020},
  publisher={Springer}
}

@InProceedings{montanaro_shao-tqc,
  author =	{Montanaro, Ashley and Shao, Changpeng},
  title =	{{Quantum Algorithms for Learning a Hidden Graph}},
  booktitle =	{17th Conference on the Theory of Quantum Computation, Communication and Cryptography (TQC 2022)},
  pages =	{1:1--1:22},
  series =	{Leibniz International Proceedings in Informatics (LIPIcs)},
  ISBN =	{978-3-95977-237-2},
  ISSN =	{1868-8969},
  year =	{2022},
  volume =	{232},
  editor =	{Le Gall, Fran\c{c}ois and Morimae, Tomoyuki},
  publisher =	{Schloss Dagstuhl -- Leibniz-Zentrum f{\"u}r Informatik},
  address =	{Dagstuhl, Germany},
  URL =		{https://drops-dev.dagstuhl.de/entities/document/10.4230/LIPIcs.TQC.2022.1},
  URN =		{urn:nbn:de:0030-drops-165081},
  doi =		{10.4230/LIPIcs.TQC.2022.1}
}

@article{montanaro2012quantum,
  title={The quantum query complexity of learning multilinear polynomials},
  author={Montanaro, Ashley},
  journal={Information Processing Letters},
  volume={112},
  number={11},
  pages={438--442},
  year={2012},
  publisher={Elsevier}
}

@book{laurent2010holomorphic,
  title={{Holomorphic Function Theory in Several Variables: An Introduction}},
  author={Laurent-Thi{\'e}baut, Christine},
  year={2010},
  publisher={Springer Science \& Business Media}
}

\end{document}